\documentclass[12pt,a4paper]{article}

\usepackage{lmodern}[lmr]

\usepackage{bm,amsfonts, amsmath, amssymb, amsthm, color, enumitem, float, mathrsfs, mathtools, multirow, nccmath, rotating, tikz, tocloft, setspace, subfig}

\usepackage[colorlinks=true, citecolor=., linkcolor=.]{hyperref}
\usepackage[nameinlink,noabbrev,capitalize]{cleveref}
\usepackage[hmargin=0.8in, vmargin=0.9in]{geometry}

\usepackage[round]{natbib}
\author{}
\title{}

\DeclareMathOperator*{\plim}{p\!\lim}

\newtheorem{assumption}{Assumption}

\newtheorem{lemma}{Lemma}
\newtheorem{proposition}{Proposition}

\newtheorem{theorem}{Theorem}

\renewcommand{\arraystretch}{1.1}

\allowdisplaybreaks[4]

\begin{document}

\setlength{\abovedisplayskip}{8pt}
\setlength{\belowdisplayskip}{8pt}

\begin{titlepage}

\title{A Robust Residual-Based Test for Structural Changes in Factor Models}
\author{Bin Peng$^{a}$, Liangjun Su$^{b},$ and Yayi Yan$^{c}\smallskip $ 
\thanks{The authors thank the co-editor, associate editor and two anonymous referees
for their constructive comments. Peng would like to acknowledge the
Australian Research Council Discovery Projects for its financial support
under Grant Number: DP250100063. Su acknowledges the National Natural
Science Foundation of China (NSFC) for financial support under the grant
number 72133002. Yan acknowledges the financial support by the NSFC under
the grant number 72303142 and the Fundamental Research Funds for the Central
Universities under grant numbers 2022110877 and 2023110099. The authors contributed equally to this paper and are credited in alphabetical order.} \\
{\small $^{a}$Department of Econometrics and Business Statistics, Monash
University, Australia}\\
{\small $^{b}$School of Economics and Management, Tsinghua University, China}%
\\
{\small $^{c}$School of Economics and Management, Shanghai University of
Finance and Economics, China} }
\date{\today }
\maketitle

\begin{abstract}
\noindent In this paper, we propose an easy-to-implement residual-based
specification testing procedure for detecting structural changes in factor
models, which is powerful against both smooth and abrupt structural changes
with unknown break dates. The proposed test is robust to the over-specified
number of factors, and serially and cross-sectionally correlated error
processes. A new central limit theorem is given for the quadratic forms of
panel data with dependence over both dimensions, thereby filling a gap in
the literature. We establish the asymptotic properties of the proposed test
statistic, and accordingly develop a simulation-based scheme to select
critical value in order to improve finite sample performance. Through
extensive simulations and a real-world application, we confirm our
theoretical results and demonstrate that the proposed test exhibits
desirable size and power in practice.

\medskip

\noindent\textbf{JEL Classification:} C14, C23, C33.\medskip

\noindent \textbf{Keywords:} Factor model; structural change; residual-based
test; serial correlation; cross-sectional dependence.
\end{abstract}

\end{titlepage}

\setlength{\abovedisplayskip}{7pt} \setlength{\belowdisplayskip}{7pt}

\section{Introduction \label{Sec1}}

As we are embracing the era of data rich environment, factor models have
been very popular in the past three decades or so. See \cite{Jolliffe} and 
\cite{FLL2021} for comprehensive reviews on a variety of applications of
factor models. One of the most attractive features is possibly the low rank
representation, which usually imposes that a set of high dimensional
observations can be driven by a series of low dimensional vectors. More
often than not, the low dimensional vectors are the foundation of many
factor augmented analyses (see, e.g., \citealp{STOCK2016415}, and references
therein).

To obtain these low dimensional vectors properly, extensive efforts have
been devoted to estimating the number of factors, viz., the length of
aforementioned low dimensional vectors. See, e.g., \cite{BN2002}, \cite{onatski2010}, \cite{NW2011}, \cite{LM2012}, \cite{ahn2013eigenvalue}, among
others. Notably, all of these works are based on the conventional factor
models where the loadings are assumed to remain constant over time. As
economists realize that the relationships between economic and financial
variables may suffer from structural changes over time, the factor loadings,
which capture the relationships between these variables and the latent
common factors, could be varying over time or some state variables. For this
reason, various methods have been proposed to test for the structural
changes in factor models or to estimate the factor models with breaks. For
example, \cite{breitung2011testing}, \cite{chen2014detecting}, \cite{han2015tests}, \cite{yamamoto_tanaka2015}, \cite{su2017time, su2020testing}, \cite{BKW2021}, \cite{fu2023testing}, and \cite{BDH2024} consider various
methods to detect structural changes in large dimensional factor models.
Except \cite{su2017time, su2020testing} and \cite{fu2023testing} who allow
for smooth changes, almost all the other papers in the literature consider
one or multiple abrupt changes. We provide a detailed summary of these works
in Section \ref{App.A1} for the purpose of comparison. In terms of
estimation, \cite{cheng2016shrinkage} study the estimation of a factor
models with a one-time break; \cite{ma_su2018} and \cite{BKW2021} consider
the estimation of factor models with multiple abrupt breaks; \cite%
{su2017time} and \cite{pelger_xiong2022} consider the kernel estimation of
factor models with loadings changing over time and stated variables,
respectively; \cite{fu_su_wang2024} study time-varying factor-augmented
regression models.

In this paper, we are also interested in testing whether the loadings change
over time or not. We aim to propose a new test statistic that has a superb
power in detecting smaller local alternatives than almost all existing tests
by utilizing richer information along both cross-sectional and time
dimensions, and remains robust against both cross-sectional and serial
dependence. One popular approach of detecting an abrupt change is via the
use of a Lagrange multiplier (LM) type of test statistic. For example, \cite{chen2014detecting}, \cite{BKW2021} and \cite{BDH2024} all consider LM-type
of tests based on the information of the unknown factors. When doing so, the
test statistic can only detect local alternatives that converge to the null
at a slow rate (e.g., $T^{-1/2}$) associated with the time dimension $\left(
T\right) $ alone, and it is generally not robust to misspecified factor
numbers (e.g., see p. 359 in \citealp{BKW2021}). Another popular approach
aims at detecting smooth structural changes and remains powerful against
abrupt breaks too. The tests of \cite{su2017time,su2020testing} and \cite%
{fu2023testing} are capable of examining the parametric loadings against the
nonparametric generalization via a smooth function. In general, such tests
are able to utilize the information along both dimensions of the residuals,
thereby detecting local alternatives converging to the null at a rate faster
than $T^{-1/2}.$ In particular, \cite{su2017time,su2020testing} consider a
kernel-based smoothing nonparametric test for structure changes that can
detect local alternatives converging to the null at the rate $%
N^{-1/4}T^{-1/2}h^{-1/4},$ where $N$ and $T$ denote cross-sectional and time
dimensions, respectively, and $h$ is a bandwidth parameter; \cite{fu2023testing} consider a nonsmoothing test\footnote{See Chapter 13 in \cite{li2007nonparametric} for the discussions on
smoothing and nonsmoothing tests.} for the structural changes in factor
models based on the discrete Fourier transform and their test can detect
local alternatives at a rate $(NT)^{-1/2}$. But to work out the asymptotic
distributional theory, early research typically assumes either cross
sectional independence or serial uncorrelation for the error terms $\left\{
\varepsilon _{it}\right\} $ in the factor model. The reason is that one
often has to deal with the following terms: 
\begin{equation}
\varepsilon _{it}\varepsilon _{js}\quad \text{with}\quad i,j\in [ N ] ,\quad t,s\in [T],  \label{EQ_BP1}
\end{equation}%
where $\left[ L\right] \coloneqq\left\{ 1,\ldots ,L\right\} $ for a given
positive integer $L$. When aggregating the information from both the
cross-sectional and time dimensions, early literature needs to balance the
rate of convergence yielded from the aggregation of the terms in %
\eqref{EQ_BP1} and the independence conditions imposed on $\left\{
\varepsilon _{it}\right\} $. To the best of our knowledge, no attempt has
been made to construct a test statistic using \eqref{EQ_BP1} and to study
its limit null distribution while allowing for dependence along both
dimensions of $\left\{ \varepsilon _{it}\right\} $. Despite the technical
challenge in dosing so, the issue is a realistic one. As \cite{CP2015}
remark, the weak dependence along both dimensions of $\left\{ \varepsilon
_{it}\right\} $ is likely to be the rule rather than the exception
practically.

In this paper, we focus on a residual-based test to examine whether one
should go beyond a traditional type of factor model practically while
accounting for dependence over both dimensions of $\left\{ \varepsilon
_{it}\right\} $. We shall derive the asymptotic properties of the newly
proposed test statistic, and conduct extensive numerical studies to confirm
our theoretical findings. To proceed, we briefly review the literature on
residual-based model specification tests. The basic idea is to use the
estimated residuals (typically from a parametric or semiparametric model
under the null) to construct a test statistic, in which 
\begin{equation*}
\text{estimated residuals}=\text{observations}-\text{fitted values}.
\end{equation*}%
Under the null, the estimated residuals are typically sufficiently close to
the true error terms so that the limiting distribution of a suitably
constructed residual-based test statistic can be established. Under the
alternative, the estimated residuals contain a term that reflects the
deviation from the null hypothesis and thus contribute to the power of the
test statistic. Along this line of research, early works focus on examining
different types of time series models (e.g., \citealp{DK1979, PP1988}).
Gradually, the idea is extended to justify different model specifications
for a variety of data types such as cross-sectional data (%
\citealp{Li1994,
JOHNXUZHENG1996263}), time series data (\citealp{GKLT2009}), and panel data (%
\citealp{su2015specification}; \citealp{WSX2022}).

To sum up, our main contributions are as follows. First, assuming the number
of factors is correctly specified, we establish the asymptotic distribution
of a residual-based test statistic under the null by allowing for dependence
along both dimensions of $\left\{ \varepsilon _{it}\right\} $. As mentioned
above, this is a highly technically challenging issue, and thus a new
central limit theorem is given for the quadratic forms of panel data with
dependence over both dimensions, thereby filling a gap in the literature.
This result is of independent interest and can be used in various
kernel-based specification tests (e.g., %
\citealp{su2015specification,chen2018nonparametric}). Second, we explore the
local power of the proposed test and show that it can detect local
alternatives converging to the null at a rate $(NT)^{-1/2}h^{-1/4},$ faster
than almost all existing tests in the literature but \cite{fu2023testing}
who consider a nonsmoothing test. It is well known a nonsmoothing test (like
that of \cite{fu2023testing}) tends to be more powerful than a smoothing
test (like ours) under low frequency local alternatives whereas\ a smoothing
test can be more powerful than a nonsmoothing test under high frequency
local alternatives as studied in \cite{rosenblatt1975quadratic}. See the
discussion on p. 400 of \cite{li2007nonparametric}. Third, we show the
robustness of our test by proving the limiting distributions under the null
and local alternatives continue to hold true when the number of factors is
over-specified. Fourth, we also study the global power behavior of the
proposed test and our investigation helps to explain where the
residual-based tests gain their power. Finally, we conduct numerical studies
to examine the finite sample performance of the test and the results
corroborate the theoretical findings.

The rest of the paper is organized as follows. Section \ref{Sec2} presents
the hypotheses and the test statistic. Section \ref{Sec3} investigates the
test statistic under the null, a sequence of Pitman local alternatives, and
a global alternative. Section \ref{Sec4} conducts numerical studies using
simulated data and real data examples. Section \ref{Sec5} concludes. The
proofs of the main results are relegated to the appendix. The online
supplement contains the proofs of the technical lemmas used in the proof of
the main result, the verification of Assumption \ref{Assumption4}, and some
additional simulation results.

\textit{Notation}. For a real-valued matrix $\mathbf{A=}\{a_{ij}\}$, let $%
\Vert \mathbf{A}\Vert $ and $\Vert \mathbf{A}\Vert _{2}$ denote the
Frobenius norm and the spectral norm, respectively. Let $\Vert \mathbf{A}%
\Vert _{1}=\max_{j}\sum_{i}|a_{ij}|$ and $\Vert \mathbf{A}\Vert _{\infty
}=\max_{i}\sum_{j}|a_{ij}|$. When $\mathbf{A}$ has full column rank, let $%
\mathbf{P}_{\mathbf{A}}\coloneqq \mathbf{A}(\mathbf{A}^{\top }\mathbf{A}%
)^{-1}\mathbf{A}^{\top }$ and $\mathbf{M}_{\mathbf{A}}\coloneqq \mathbf{I}-%
\mathbf{P}_{\mathbf{A}},$ where $^{\top }$ denotes transpose. When $\mathbf{A%
}$ is symmetric, we use $\lambda _{\min }(\mathbf{A})$ and $\lambda _{\max }(%
\mathbf{A})$ to denote its minimum and maximum eigenvalues, respectively.
For a random vector $\mathbf{v}$, let $|\mathbf{v}|_{q}\coloneqq \left(
E\Vert \mathbf{v}\Vert ^{q}\right) ^{1/q}$. For two constants $a$ and $b$, $%
a\wedge b\coloneqq\min (a,b)$, $a\vee b\coloneqq\max (a,b)$, and $a\simeq b$
stands for $a=O(b)$ and $b=O(a)$. For two random variables $c$ and $d$, $%
c\asymp d$ stands for $c=O_{P}(d)$ and $d=O_{P}(c)$. The operators $%
\rightarrow _{P}$ and $\rightarrow _{D}$ stand for convergence in
probability and convergence in distribution, respectively. $\left(N,T\right)
\rightarrow \infty $ signifies that $N$ and $T$ diverge to infinity jointly.

\section{The Null Hypothesis and the Test Statistic \label{Sec2}}

In this section we introduce the null hypothesis and the test statistic.

\subsection{Hypotheses \label{Sec2.1}}

Consider a series of high dimensional vectors $\{\mathbf{x}_{t}:t\in \left[ T%
\right] \}$, and assume that $\{\mathbf{x}_{t}\}$ is generated by the
following factor model that is potentially time-varying:

\begin{equation}
\mathbf{x}_{t}=\pmb{\Lambda}_{t}\mathbf{f}_{t}+\pmb{\varepsilon}_{t},
\label{Eq2.1}
\end{equation}
where $\pmb{\varepsilon}_{t}=(\varepsilon _{1t},\ldots ,\varepsilon
_{Nt})^{\top }$ is a vector of zero-mean error terms, $\pmb{\Lambda}_{t}=(%
\pmb{\lambda}_{1t},\ldots ,\pmb{\lambda}_{Nt})^{\top }$ is an $N\times r$
time-varying loading matrix, and $\mathbf{f}_{t}$ is an $r\times 1$ latent
factor. We focus on the case where $r$ is fixed, and $\left( N,T\right)
\rightarrow \infty $. To examine whether $\pmb{\Lambda}_{t}$'s evolve over
time, the null hypothesis is specified as follows:

\begin{equation}
\mathbb{H}_{0}:\ \pmb{\Lambda}_{t}=\pmb{\Lambda}\ \text{for all}\ t\in [T],
\label{Eq2.2}
\end{equation}%
where $\pmb{\Lambda}$ is a conformable constant matrix. The alternative
hypothesis $\mathbb{H}_{1}$ is the negation of $\mathbb{H}_{0}.$ Obviously,
under $\mathbb{H}_{0},$ the factor model in (\ref{Eq2.1}) degenerates to the
conventional time-invariant factor model that has been widely studied in the
literature. Under $\mathbb{H}_{1},$ the loadings $\pmb{\lambda}_{it}$ must
change over time for some $i\in \left[ N \right]$.

\subsection{The Test Statistic \label{Sec2.2}}

Under $\mathbb{H}_{0}$, we have 
\begin{equation}
\mathbf{X}=\mathbf{F}\pmb{\Lambda}^{\top }+\mathcal{E},  \label{Eq2.3a}
\end{equation}
where $\mathbf{X=(\mathbf{x}}_{1},\ldots,\mathbf{\mathbf{x}}_{T}\mathbf{)}%
^{\top }\mathbf{,}$ $\mathbf{F}=(\mathbf{f}_{1},\ldots ,\mathbf{f}%
_{T})^{\top },$ and $\mathcal{E}$ is defined accordingly. Following the lead
of \cite{BN2002}, we impose the identification condition $\frac{1}{T}\mathbf{%
F}^{\top }\mathbf{F}=\mathbf{I}_{r}$ to conduct the standard principal
component analysis (PCA): 
\begin{equation}
\widehat{\mathbf{F}}\widehat{\mathbf{V}}=\mathbf{X}\mathbf{X}^{\top }%
\widehat{\mathbf{F}},  \label{Eig1}
\end{equation}%
where $\widehat{\mathbf{V}}$ is an $r\times r$ diagonal matrix including the
first $r$ largest eigenvalues of $\mathbf{X}\mathbf{X}^{\top }$ on the main
diagonal, and $\frac{1}{T}\widehat{\mathbf{F}}^{\top }\widehat{\mathbf{F}}=%
\mathbf{I}_{r}$. Accordingly, we have the residual matrix $\widehat{\mathcal{%
E}}=\mathbf{X}-\widehat{\mathbf{F}}\widehat{\pmb{\Lambda}}^{\top }=\mathbf{M}%
_{\widehat{\mathbf{F}}}\mathbf{X},$ where $\widehat{\pmb{\Lambda}}=\frac{1}{T%
}\mathbf{X}^{\top }\widehat{\mathbf{F}}$. Having $\widehat{\mathcal{E}}=(%
\widehat{\mathcal{E}}_{1},\ldots ,\widehat{\mathcal{E}}_{N})=\{\widehat{%
\varepsilon }_{it}\}_{T\times N}$ in hand, we construct the residual-based
test statistic as follows: 
\begin{equation}
L_{NT}=\frac{1}{T^{2}N^{2}}\sum_{i,j=1}^{N}\widehat{\mathcal{E}}_{i}^{\top
}\,\mathbf{K}_{h}\,\widehat{\mathcal{E}}_{j},  \label{Eq2.3}
\end{equation}%
where $\mathbf{K}_{h}=\{\frac{1}{h}K(\frac{t-s}{Th})\}_{T\times T}$ with $t,s\in [ T]$ and $h\coloneqq h_{T}$ is a bandwidth parameter. The test
statistic in \eqref{Eq2.3} has a typical form of the residual-based test,
and the expression further concurs with \eqref{EQ_BP1} that one has to
account for dependence yielded by $\varepsilon _{it}\varepsilon _{js}$
appearing in the leading term in the expansion of $L_{NT}.$

\section{The Asymptotic Properties of the Test Statistic \label{Sec3}}

In this section, we establish the limiting distribution of $L_{NT}$ under
the null, and study its asymptotic local and global power properties. To
facilitate the presentation, we first provide the basic assumptions used for
our asymptotic analyses in Section \ref{Sec3.1}. In Section \ref{Sec3.2} we
establish the asymptotic distribution of the test statistic under the null,
and then present the result with the number of factors over-specified. In
Sections \ref{Sec3.3} and \ref{Sec3.4}, we study the power properties of our
test under different scenarios, of which one is also helpful to explain
where the residual-based tests gain their global power. Finally, we point
out that the newly proposed test can also be used to select the number of
factors practically.

\subsection{Basic Assumptions\label{Sec3.1}}

Note that we can write $\pmb{\Lambda}=(\pmb{\lambda}_{1},\ldots , %
\pmb{\lambda}_{N})^{\top }$ under $\mathbb{H}_{0}.$ As in \cite{su2017time},
we will assume that the loadings are nonrandom without loss of generality.
To facilitate the asymptotic analyses, we make the following assumptions.

\begin{assumption}\label{Assumption1}

\noindent (a) $K(\cdot )$ is a continuous and symmetric function with
bounded support $[-1,1]$, and $K(0)=1$.

\smallskip

\noindent (b) As $(N,T)\rightarrow \infty ,$ $h\rightarrow 0,$ $%
Th\rightarrow \infty ,$ $hT^{2}/N^{2}\rightarrow 0$ and $hN^{2}/T^{2}%
\rightarrow 0$.
\end{assumption}

\begin{assumption}
\label{Assumption2}

Suppose that $c_{1}\leq \lambda _{\min }(\pmb{\Sigma}_{\lambda})\leq \lambda
_{\max }(\pmb{\Sigma}_{\lambda})\leq c_{2}$ and $\max_{i\in \left[ N \right]
}\Vert \pmb{\lambda}_{i}\Vert \leq c_{0}$ for some constants $%
c_{0},c_{1},c_{2}>0$, where $\pmb{\Sigma}_{\lambda}\coloneqq \lim_{N\to
\infty}\frac{1}{N}\pmb{\Lambda}^{\top } \pmb{\Lambda}$.
\end{assumption}

\begin{assumption}
\label{Assumption3}

\noindent (a) Let $\{ \mathbf{f}_{t} \} $ be a stationary process such that $%
\mathbf{f}_{t}=\mathbf{G}(\mathbf{e}_{t},\mathbf{e}_{t-1},\ldots )$ with $%
\mathbf{e}_{t}$ being independent and identically distributed (i.i.d.)
random vectors, $E[\mathbf{f}_{t}\mathbf{f}_{t}^{\top }]=\pmb{\Sigma}_{f}>0$%
, and $\mathbf{G}(\cdot )=[G_{1}(\cdot),\ldots ,G_{r}(\cdot )]^{\top }$
being a vector of measurable functions.

\smallskip

\noindent (b) Let $|\mathbf{f}_{m}-\mathbf{f}_{m}^{\ast }|_{4} \eqqcolon \lambda _{4}^{f}(m)=O(m^{-\alpha })$ for some $\alpha >3$, where $\mathbf{f}
_{t}^{\ast }=\mathbf{G}(\mathbf{e}_{t},\ldots ,\mathbf{e}_{1},\mathbf{e}
_{0}^{\prime },\mathbf{e}_{-1},\ldots )$ is the coupled version of $\mathbf{f }_{t}$ with $\mathbf{e}_0^{\prime }$ being an independent copy of $\{\mathbf{e}_{t}\}$.
\end{assumption}

\begin{assumption}
\label{Assumption4}

\noindent (a) Let $\{\pmb{\varepsilon}_{t}\}$ be a stationary process such
that $\pmb{\varepsilon}_{t}=\mathbf{g}(\mathbf{e}_{t},\mathbf{e}
_{t-1},\ldots )$, $E(\pmb{\varepsilon}_{t}\, |\, \mathbf{f}_{t})=\mathbf{0}
_{N} $, $E[\pmb{\varepsilon}_{t}\pmb{\varepsilon}_{t}^{\top }]=\pmb{\Sigma}
_{\varepsilon }$, $\left\Vert \pmb{\Sigma}_{\varepsilon }\right\Vert
_{2}\leq c_{3}$ for some constant $c_{3}>0,$ and $\mathbf{g}(\cdot
)=[g_{1}(\cdot ),\ldots ,g_{N}(\cdot )]^{\top }$ is a vector of measurable
functions.

\smallskip

\noindent (b) Let $\sup_{\mathbf{v}\in \mathbb{R}^{N},\Vert \mathbf{v}\Vert
=1}|\overline{\varepsilon }_{m,\mathbf{v}}-\overline{\varepsilon }_{m,%
\mathbf{v}}^{\ast }|_{4}\eqqcolon\lambda _{4}^{\varepsilon }(m)=O(m^{-\alpha
})$ for some $\alpha >3$, where $\overline{\varepsilon }_{t,\mathbf{v}}=%
\pmb{\varepsilon}_{t}^{\top }\mathbf{v}$, $\overline{\varepsilon }_{t,%
\mathbf{v}}^{\ast }=\pmb{\varepsilon}_{t}^{\ast \top }\mathbf{v}$ with $%
\pmb{\varepsilon}_{t}^{\ast }=\mathbf{g}(\mathbf{e}_{t},\ldots ,\mathbf{e}%
_{1},\mathbf{e}_{0}^{\prime },\mathbf{e}_{-1},\ldots )$, and $\mathbf{e}%
_{0}^{\prime }$ is as defined in Assumption \ref{Assumption3}(b).

\smallskip

\noindent (c) $|\mathcal{P}_{t-m}(\mathbf{f}_{t}\frac{1}{\sqrt{N}}\sum_{i=1}^{N} \varepsilon _{it}\varepsilon _{is})|_{2}=O(m^{-\beta })$ for
some $\beta >1 $, where $\mathcal{P}_{t}[\cdot ]=E[\cdot \,|\, \mathscr{E}%
_{t}]$ $-E[\cdot \,|\, \mathscr{E}_{t-1}]$ and $\mathscr{E}_{t}=(\mathbf{e}%
_{t}, \mathbf{e}_{t-1},\ldots )$.
\end{assumption}

Assumption \ref{Assumption1} imposes some standard conditions on the kernel
function $K\left( \cdot \right) $ and the bandwidth parameter $h.$ In
addition, we allow for $N$ and $T$ diverging to infinity at different or
same rates. Without loss of generality, Assumption \ref{Assumption2} assumes
that the loadings under $\mathbb{H}_{0}$ are non-random; see, e.g., \cite%
{su2017time} and \cite{fu2023testing} and the references therein. It is easy
to see that the asymptotic analyses remain valid when the loadings are
random and independent of the factors and error terms, and satisfy some
moment conditions as in \cite{BN2002}.

Assumption \ref{Assumption3} imposes weak dependence conditions on $\{ 
\mathbf{f}_{t}\}$, allowing for serial correlations and conditional
heteroscedasticity of unknown forms. The strength of temporal dependence is
defined by the functional dependence measure of \cite{wu2005nonlinear},
which nests many stationary processes as special cases (e.g., %
\citealp[p.204]{tong1990non}; \citealp{wu2005nonlinear}). Here, we only
require an algebraic decay rate of the temporal dependence. Intuitively, the
physical dependence coefficients $\{\lambda_4^f(m)\}$ quantify the
dependence of outputs $\{\mathbf{f}_{m}\}$ on $\mathbf{e}_0$.

Assumption \ref{Assumption4} regulates the decay rate of temporal dependence
after taking weighted average cross-sectionally. By doing so, $\varepsilon
_{it}$'s are allowed to exhibit cross-sectional dependence of various
unknown forms. The constants \textquotedblleft $\alpha $\textquotedblright\
and \textquotedblleft $\beta $\textquotedblright\ play a role similar to the 
$\alpha $-mixing or $\beta $-mixing coefficients in the time series
literature and we only require an algebraic decay rate of the temporal
dependence. Assumption \ref{Assumption4}(c) requires the underlying data
generating process to satisfy the following moment restriction (see Lemma %
\ref{Lemma1}(c) for details) 
\begin{equation*}
\left\vert \frac{1}{\sqrt{TN}}\sum_{t=1}^{T}\mathbf{f}_{t}\sum_{i=1}^{N}[%
\varepsilon _{it}\varepsilon _{is}-E(\varepsilon _{it}\varepsilon
_{is})]\right\vert _{2}=O(1).
\end{equation*}%
As a high level assumption, the above condition has been widely adopted in
the literature on factor models (see, e.g., Assumption F.1 in %
\citealp{bai2003inferential}, Assumption A.1 (vii) in \citealp{su2017time}
and Assumption A.7 in \citealp{fu2023testing}). Intuitively, Assumption \ref%
{Assumption4}(c) requires that $\mathbf{f}_{t}\frac{1}{\sqrt{N}}%
\sum_{i=1}^{N}\varepsilon _{it}\varepsilon _{is}$ and $\mathbf{e}_{t-m}$ are
asymptotically independent as $m\rightarrow \infty $ for any fixed $s.$ In
Section \ref{App.B2} of the online supplement, we provide an example to show
that Assumption \ref{Assumption4} is fulfilled by a well known data
generating process.

\subsection{Asymptotic Null Distribution \label{Sec3.2}}

Define 
\begin{equation*}
\sigma _{\varepsilon ,a}^{2}=\lim_{(N,T)\rightarrow \infty }\frac{1}{Th}%
\sum_{t,s=1}^{T}E(\overline{\varepsilon }_{t,a}\overline{\varepsilon }%
_{s,a})K\left( \frac{t-s}{Th}\right) ,
\end{equation*}%
where $\overline{\varepsilon }_{t,a}\coloneqq \frac{1}{\sqrt{N}}%
\sum_{i=1}^{N}(1-a_{i})\varepsilon _{it}$, $a_{i}\coloneqq \pmb{\lambda}_{i}^{\top }(%
\frac{1}{N}\pmb{\Lambda}^{\top }\pmb{\Lambda})^{-1}\overline{\pmb{\lambda}}$%
, and $\overline{\pmb{\lambda}}\coloneqq\frac{1}{N}\sum_{i=1}^{N}\pmb{\lambda}_{i}$%
. Clearly, $\sigma _{\varepsilon ,a}^{2}$ denotes a kernel-weighted version
of the long-run variance of ${\varepsilon }_{it}.$ We are now
ready to present the first main result of this paper.

\begin{theorem}
\label{theorem1} Suppose that Assumptions \ref{Assumption1}--\ref{Assumption4} hold, and $\sigma _{\varepsilon ,a}^{2}>0$. As $(N,T)\rightarrow \infty $, 
\begin{equation*}
TN\sqrt{h}[L_{NT}-(TNh)^{-1}\sigma _{\varepsilon ,a}^{2}]\rightarrow
_{D}N(0,2\nu _{0}\sigma _{\varepsilon ,a}^{4})\text{ under }\mathbb{H}_{0},
\end{equation*}%
where $\nu _{0}=\int_{-1}^{1}K^{2}(u)\mathrm{d}u$.
\end{theorem}

As indicated in Theorem \ref{theorem1}, in order to achieve an
asymptotically nondegenerate normal distribution, the long-run variance $%
\sigma _{\varepsilon ,a}^{2}$ cannot be 0. This immediately rules out the
following special case: 
\begin{equation*}
\mathbf{x}_{t}=\mathbf{1}_{N}f_{t}+\pmb{\varepsilon}_{t},
\end{equation*}%
where $f_{t}$ is a common trend function with loadings given identically by
1. In this case, we have $1-a_{i}=0$ for all $i\in \lbrack N]$. Notably, the
term $1-a_{i}$ is equivalent to $B_{i}$ of \citet[p. 307]{fu2023testing},
which also suffers the issue raised here. The reason of having a structure
as in $1-a_{i}$ is due to the use of projection matrix involved in the PCA.
Our test has the usual asymptotic normal distribution when one has
idiosyncratic trends as follows 
\begin{equation*}
\mathbf{x}_{t}=\mathbf{\Lambda }f_{t}+\pmb{\varepsilon}_{t},\text{ }
\end{equation*}%
or equivalently in scalar notations: $x_{it}=\lambda _{i}f_{t}+\varepsilon
_{it},$ where $\lambda _{i}$ is not a constant across $i.$ Interestingly,
the common trend issue has been addressed in another literature. Various
tests have been proposed to test for common trend in the literature; see
e.g., \cite{zhang2012testing} and \cite{WSX2022}. As discussed in \cite%
{zhang2012testing}, high-dimensional time series rarely share a common trend
in the real world and idiosyncratic trending behavior is a norm rather than
the exception. For the factor model, if one wishes, one can first eliminate
the possibility of having a common time trend in the dataset, say, by using
the test of \cite{WSX2022}. But this is subject to the usual pretesting
issue. For this reason we do not recommend testing the presence of common
trends before using our test.

Building on Theorem \ref{theorem1}, we investigate the case when the number
of factors is over-specified. That said, suppose that when conducting PCA, $%
\widetilde{r}\ (\geq r)$ is specified as the number of factors. Then the
eigenvalue problem in (\ref{Eig1}) becomes 
\begin{equation*}
\widetilde{\mathbf{F}}\widetilde{\mathbf{V}}=\mathbf{X}\mathbf{X}^{\top }%
\widetilde{\mathbf{F}},
\end{equation*}%
where $\widetilde{\mathbf{V}}$ is an $\widetilde{r}\times \widetilde{r}$
diagonal matrix including the first $\widetilde{r}$ largest eigenvalues of $%
\mathbf{X}\mathbf{X}^{\top }$ on the main diagonal, and $\frac{1}{T}%
\widetilde{\mathbf{F}}^{\top }\widetilde{\mathbf{F}}=\mathbf{I}_{\widetilde{r%
}}$. Accordingly, we have 
\begin{equation*}
\widetilde{\mathcal{E}}=\mathbf{X}-\widetilde{\mathbf{F}}\widetilde{%
\pmb{\Lambda}}^{\top }=\mathbf{M}_{\widetilde{\mathbf{F}}}\mathbf{X},
\end{equation*}%
where $\widetilde{\pmb{\Lambda}}=\frac{1}{T}\mathbf{X}^{\top }\widetilde{%
\mathbf{F}}$ and $\widetilde{\mathcal{E}}=(\widetilde{\mathcal{E}}%
_{1},\ldots ,\widetilde{\mathcal{E}}_{N})=\{\widetilde{\varepsilon }%
_{it}\}_{T\times N}$. With a bit abuse of notation, we continue to use $%
L_{NT}$ to denote the test statistic: 
\begin{equation}
L_{NT}=\frac{1}{T^{2}N^{2}}\sum_{i,j=1}^{N}\widetilde{\mathcal{E}}_{i}^{\top
}\,\mathbf{K}_{h}\,\widetilde{\mathcal{E}}_{j}.  \label{Eq2.5}
\end{equation}

We summarize the asymptotic property of $L_{NT}$ under the null in the
following theorem.

\begin{theorem}
\label{theorem2} Let Assumptions \ref{Assumption1}--\ref{Assumption4} hold
and $L_{NT}$ be defined in (\ref{Eq2.5}) with $\widetilde{r}\geq r$ being
fixed. Then as $(N,T)\rightarrow \infty $, 
\begin{equation*}
TN\sqrt{h}[L_{NT}-(TNh)^{-1}\sigma _{\varepsilon ,a}^{2}]\rightarrow
_{D}N(0,2\nu _{0}\sigma _{\varepsilon ,a}^{4})\text{ under }\mathbb{H}_{0}.
\end{equation*}
\end{theorem}

Theorem \ref{theorem2} infers that when the number of factors is
over-specified, we are still able to achieve the same limit null
distribution as in Theorem \ref{theorem1}.

\medskip

To implement our test, one needs to estimate the long-run variance $\sigma
_{\varepsilon ,a}^{2}$. Define a panel heteroskedasticity and
autocorrelation consistent (HAC) variance estimator by 
\begin{equation}
\widehat{\sigma }_{\varepsilon ,a}^{2}=\sum_{k=-l}^{l}\widehat{\sigma }
_{\varepsilon ,a,k}^{2}a(k/l),  \label{Eq2.4}
\end{equation}
where $\widehat{\sigma }_{\varepsilon ,a,k}^{2} \coloneqq \frac{1}{T}
\sum_{t=1}^{T-|k|} \widehat{\overline{\varepsilon }}_{t}\widehat{\overline{
\varepsilon }} _{t+|k|}$, $\widehat{\overline{\varepsilon }}_{t}\coloneqq%
\frac{1}{ \sqrt{N}} \sum_{i=1}^{N}\widetilde{\varepsilon }_{it}$, $a(\cdot )$
is a kernel function, and $l\coloneqq l_{T}$ is a bandwidth parameter for
the HAC estimation. We further impose the following assumption.

\begin{assumption}
\label{Assumption5}

\noindent (a) Suppose that $a(\cdot )$ is a symmetric positive kernel
function that is Lipschitz continuous on $[-1,1]$. For $q\in \{1,2\}$, $%
\lim_{|x|\rightarrow 0}\frac{1-a(x)}{|x|^{q}}=\bar{c}_{q}$ for some constant 
$\bar{c}_{q}\in (0,\infty )$.

\smallskip

\noindent (b) As $\left(N,T\right) \rightarrow \infty ,$ $1/l+l/T\rightarrow
0$, $Tl/N^{2}\rightarrow 0$ and $lN^{2}/T^{3}\rightarrow 0$.
\end{assumption}

Assumption \ref{Assumption5} includes a set of typical conditions for the
kernel function $a(\cdot)$ and the bandwidth $l$, and we omit the discussion and refer interested readers
to \cite{andrews1991heteroskedasticity} for details. Then the following
result holds.

\begin{proposition}
\label{Coro1} Let Assumptions \ref{Assumption1}--\ref{Assumption5} hold.
Under the null hypothesis in \eqref{Eq2.2}, for any fixed $\widetilde{r}\geq
r$, 
\begin{equation*}
\widehat{\sigma }_{\varepsilon ,a}^{2}=\sigma _{\varepsilon
,a}^{2}+O_{P}(l^{-q}+\sqrt{l/T}).
\end{equation*}%
In addition, if $l\simeq T^{1/(2q+1)}$ and $hT^{2q/(2q+1)}\rightarrow \infty 
$, then 
\begin{equation}
\widehat{L}_{T}=\frac{TN\sqrt{h}[L_{NT}-(TNh)^{-1}\widehat{\sigma }%
_{\varepsilon ,a}^{2}]}{\sqrt{2\nu _{0}}\widehat{\sigma }_{\varepsilon
,a}^{2}}\rightarrow _{D}N(0,1).  \label{L_T_hat}
\end{equation}
\end{proposition}

Proposition \ref{Coro1} infers that $\widehat{L}_{T}$ is asymptotically
pivotal under the null. By the first part of Proposition \ref{Coro1}, $%
l\simeq T^{1/(2q+1)}$ corresponds to the optimal rate of bandwidth in terms
of minimizing the asymptotic mean squared error of $\widehat{\sigma }%
_{\varepsilon ,a}^{2}.$

\subsection{Asymptotic Local Power}\label{Sec3.3}

To study the asymptotic local power property of our test, we consider a
sequence of local alternatives: 
\begin{equation}
\mathbb{H}_{1}(a_{TN}):\pmb{\lambda}_{it}=\pmb{\lambda}_{i}+a_{TN}\mathbf{g}%
_{it}\ \ \text{for all}\ \ i\in \lbrack N],  \label{Eq2.8}
\end{equation}%
where $a_{TN}\rightarrow 0$, $\mathbf{g}_{it}=\mathbf{g}_{i}(t/T)$, and $%
\int_{0}^{1}\mathbf{g}_{i}(\tau )\mathrm{d}\tau =\mathbf{0}_{r}$ for the
purpose of normalization. The term $a_{TN}\mathbf{g}_{it}$ characterizes the
departure of $\pmb{\lambda}_{it}$ from the constant $\pmb{\lambda}_{i}$, and
our setup allows for both abrupt and smooth structural changes. For example,
if we have only one break at $T_{1}$ with $\tau _{1}=T_{1}/T:$%
\begin{equation*}
\mathbf{\lambda }_{it}=\left\{ 
\begin{array}{cc}
\mathbf{\lambda }_{i\left( 1\right) } & \text{ for }t<T_{1} \\ 
\mathbf{\lambda }_{i\left( 2\right) } & \text{for }t\geq T_{1}%
\end{array}%
\right. ,
\end{equation*}%
then we can define 
\begin{eqnarray*}
\mathbf{g}_{i} ( \tau ) &=&\left\{ 
\begin{array}{cc}
\mathbf{g}_{i\left( 1\right) } & \text{ for }\tau <T_{1}/T \\ 
\mathbf{g}_{i\left( 2\right) } & \text{for }\tau \geq T_{1}/T%
\end{array}%
\right. 
\end{eqnarray*}
with $\mathbf{g}_{i\left( 1\right) } =-\frac{\mathbf{\lambda }%
_{i\left( 2\right) }-\mathbf{\lambda }_{i\left( 1\right) }}{a_{TN}}\tau _{2}$,
$\mathbf{g}_{i\left( 2\right) }=\frac{\mathbf{\lambda }_{i\left(
2\right) }-\mathbf{\lambda }_{i\left( 1\right) }}{a_{TN}}\tau _{1}$ and $\tau _{2}=1-\tau _{1}$ to ensure $\int_{0}^{1}\mathbf{g}_{i}(\tau )\mathrm{d}\tau =\tau _{1}\mathbf{%
g}_{i\left( 1\right) }+\tau _{2}\mathbf{g}_{i\left( 2\right) }=0.$ Here, $%
a_{TN}$ controls the speed at which $\mathbf{\lambda }_{i\left( 2\right) }-%
\mathbf{\lambda }_{i\left( 1\right) }$ shrinks to 0.

The following theorem indicates that our test statistic can detect a class
of Pitman local alternatives at a departure rate of $(TN)^{-1/2}h^{-1/4}$.
In addition, our test is robust to the over-specified number of factors
under the local alternatives.

\begin{theorem}
\label{theorem3} Suppose that Assumptions \ref{Assumption1}--\ref%
{Assumption4} hold, and $Th^{3/2}\rightarrow \infty $. Then under $\mathbb{H}%
_{1}(a_{TN})$ with $a_{TN}=(TN)^{-1/2}h^{-1/4}$ and for any fixed $%
\widetilde{r}\geq r$, we have 
\begin{equation*}
TN\sqrt{h}[L_{NT}-(TNh)^{-1}\sigma _{\varepsilon ,a}^{2}]\rightarrow
_{D}N(\mu _{1},2\nu _{0}\sigma _{\varepsilon ,a}^{4}),
\end{equation*}%
where $\mu _{1}=\int_{-1}^{1}K(\tau )\mathrm{d}\tau \int_{0}^{1}[\overline{%
\mathbf{g}}_{a}(\tau )^{\top }E(\mathbf{f}_{t})]^{2}\mathrm{d}\tau $, and $%
\overline{\mathbf{g}}_{a}(\tau )=\lim_{N\rightarrow \infty }\frac{1}{N}%
\sum_{i=1}^{N}(1-a_{i})\mathbf{g}_{i}(\tau )$.
\end{theorem}

Theorem \ref{theorem3} indicates that $L_{NT}$ has asymptotically
non-negligible power in detecting local alternatives converging to the null
at the rate $(TN)^{-1/2}h^{-1/4}$ provided $\mu _{1}>0.$ This rate is much
faster than almost all existing tests in the literature by noticing that $%
Nh^{1/2}\rightarrow \infty $ as $\left( N,T\right) \rightarrow \infty $. As
mentioned in the introduction, existing time-series-based tests, such as 
\cite{breitung2011testing}, \cite{chen2014detecting}, \cite{han2015tests}, 
\cite{yamamoto_tanaka2015}, \cite{BKW2021}, and \cite{BDH2024}, have power
in detecting local alternatives at rate $T^{-1/2}.$ Similarly, the
kernel-based smoothing tests of \cite{su2017time} and \cite{su2020testing}
can detect local alternatives at the rate $N^{-1/4}T^{-1/2}h^{-1/4}.$ The
only test that can detect local alternatives converging to the null faster
than ours is the non-smoothing test of \cite{fu2023testing}, which is based
on a discrete Fourier transform and can detect local alternatives at rate $%
(NT)^{-1/2}.$ Again, neither the smoothing kernel test of \cite{su2017time}
and \cite{su2020testing} nor the nonsmoothing test of \cite{fu2023testing}
allows for serial dependence in the error terms.

It is worth mentioning that the result in Theorem \ref{theorem3} holds for
any fixed $\widetilde{r}\geq r.$ That is, the over-specification of the
number of factors does not affect the local power property of our test.
Intuitively, this is due to the fact that the local alternatives converge to
the null at a very fast rate (viz., $(TN)^{-1/2}h^{-1/4}$), and the
over-specified $\widetilde{r}-r$ factors are essentially very weak factors
so that they cannot be identified by the PCA and their effect is still
present through the residuals in the PCA.

Note that the requirement that $\mu _{1}>0$ rules out the case where $E(%
\mathbf{f}_{t})=0$. We believe that this requirement is not restrictive for
practical purpose. For example, in the case of asset pricing (e.g., %
\citealp{lettau2020factors}), the risk factors, say, in the Fama-French
three-factor model (and also the estimated latent pricing factors) usually
do not have zero mean due to the risk premium required by investors. The
inclusion of an intercept may not address the issue of non-zero mean
factors. As a matter of factor, when an idiosyncratic intercept term $\alpha
_{i}$ is included in the regression, it can be absorbed into the vector of
factor loadings which has the associated factor $\mathbf{1}_{T}$. In this
case, the population mean of $\mathbf{f}_{t}$ is surely nonzero.

Below we show that when $E(\mathbf{f}_{t})=0$, our test statistic is able to
detect a class of local alternatives at a departure rate of $%
N^{-1/2}h^{1/4}. $ This rate converges to zero slower than $%
(TN)^{-1/2}h^{-1/4}$, indicating the price we have to pay when all the
factors in $\mathbf{f}_{t}$ have zero mean.

\begin{proposition}
\label{prop2} Suppose that Assumptions \ref{Assumption1}--\ref{Assumption4}
hold with $E(\mathbf{f}_t)=0$ and $E(\mathbf{f}_t\pmb{\varepsilon}%
_{s}^\top)=0$ for any $1\leq t,s \leq T$. Then under $\mathbb{H}_{1}(a_{TN})$
with $a_{TN}=N^{-1/2}h^{1/4}$ and for any fixed $\widehat{r}\geq r$, we have 
\begin{equation*}
TN\sqrt{h}[L_{NT}-(TNh)^{-1}\sigma _{\varepsilon ,a}^{2}]\rightarrow
_{D}N(\mu_{2},2\nu _{0}\sigma _{\varepsilon ,a}^{4}),
\end{equation*}
where $\mu _{2}=\lim_{T\to\infty}\frac{1}{T}\sum_{t,s=1}^{T}E\left(\overline{%
\mathbf{g}}_a(t/T)^\top \mathbf{f}_t\mathbf{f}_s^\top\overline{\mathbf{g}}%
_a(s/T)\right)K\left( \frac{t-s}{Th}\right)$.
\end{proposition}

\subsection{Global Power}

\label{Sec3.4}

In this section, we explain the source of power of our test when data
include global structural breaks. Additionally, we point out that the newly
proposed test can be used to select the number of factors practically.

To accommodate some typical settings of structural breaks such as those in 
\cite{breitung2011testing}, we write the loadings as $\mathbb{A}=[%
\pmb{\Lambda}_{1},\ldots ,\pmb{\Lambda}_{T}],$ where $\mathbb{A}$ is a $%
N\times Tr$ matrix. Let $R\coloneqq(Tr)\wedge N$ and consider the singular
value decomposition (SVD): $\mathbb{A}=\mathbf{U}\mathbf{S}\mathbf{V}^{\top
} $, where $\mathbf{U}$ and $\mathbf{V}$ are $N\times R$ and $Tr\times R$
unitary matrices, and $\mathbf{S}$ denotes an $R\times R$ diagonal matrix
with singular values $s_{TN,1},\ldots ,s_{TN,R}$ along the diagonal in
descending order. Then we can rewrite \eqref{Eq2.1} in a compact form: 
\begin{equation}
\mathbf{X}=\mathbb{F}\mathbb{A}^{\top }+\mathcal{E},  \label{Eq2.6}
\end{equation}%
where $\mathbb{F}=\mathrm{diag}(\mathbf{f}_{1}^{\top },\ldots ,\mathbf{f}%
_{T}^{\top })$ is a $T\times Tr$ block diagonal matrix.

To facilitate the development, we add the following assumption.

\begin{assumption}
\label{Assumption6}

\noindent (a) $\max_{i,t}\Vert \pmb{\lambda}_{it}\Vert \leq c_{0}$ for some
constant $c_{0}>0$, where $\pmb{\lambda}_{it}$ is the $i^{th}$ column of $%
\pmb{\Lambda}_{t}^{\top }$.

\smallskip

\noindent (b) The first $J$ singular values of $\mathbf{S}$ is of order $%
\sqrt{TN}$, while the remaining singular values satisfy $(TN)^{-1}
\sum_{j=J+1}^{R}s_{TN,j}^{2}=o(1)$, where $J$ is fixed.
\end{assumption}

This assumption encompasses many existing works as special cases. For
example, the model \eqref{Eq2.6} with a single structural break as in \cite%
{breitung2011testing} has $r<\mathrm{rank}(\mathbb{A})=J\leq2r.$ Assumption %
\ref{Assumption6}(a) regulates the factor loadings such that $\Vert \mathbb{A%
}\Vert /\sqrt{TN}=O(1)$. Thus, we have $\mathrm{tr}(\mathbb{A}^{\top }%
\mathbb{A}/(TN))=(TN)^{-1} \sum_{j=1}^{R}s_{TN,j}^{2}=O(1)$. Assumption \ref%
{Assumption6}(b) suggests that the first $J$ singular values diverge to
infinity at the rate of $\sqrt{ TN}$.

Now, we rewrite the model in \eqref{Eq2.6} as 
\begin{equation}
\mathbf{X}=\mathbb{F}\mathbb{A}^{(J),\top }+\mathbb{F}\mathbb{A}^{(-J),\top
}+\mathcal{E}=\mathbb{F}\mathbb{A}^{(J),\top }+\mathcal{E}^{\dagger },
\label{Eq2.7}
\end{equation}
where $\mathbb{A}^{(J)}=\mathbf{U}^{(J)}\mathbf{S}^{(J)}(\mathbf{V}^{(J)} 
\mathbf{)}^{\top }$, $\mathbb{A}^{(-J)}=\mathbf{U}^{(-J)}\mathbf{S}^{(-J)}( 
\mathbf{V}^{(-J)})^{\top }$, $\mathcal{E}^{\dagger }=\mathcal{E}+\mathbb{F(} 
\mathbb{A}^{(-J)})^{\top }$, $\mathbf{U}^{(J)}$ (resp. $\mathbf{V}^{(J)}$)
and $\mathbf{U}^{(-J)}$ (resp. $\mathbf{V}^{(-J)}$) contain the first $J$
and the last $R-J$ columns of $\mathbf{U}$ (resp. $\mathbf{V}$)
respectively, and $\mathbf{S}^{(J)}$ (resp. $\mathbf{S}^{(-J)}$) contains
the first $J$ (resp. the remaining $R-J$) singular values of $\mathbf{S}$.
Therefore, \eqref{Eq2.6} has a representation involving the time-invariant
factor loadings only: 
\begin{equation*}
\mathbf{X}=\mathcal{F}\pmb{\Theta}^{\top }+\mathcal{E}^{\dagger },
\end{equation*}
where $\mathcal{F}=\sqrt{T}\mathbb{F}\mathbf{V}^{(J)}$ with $\Vert \mathcal{%
F }^{\top }\mathcal{F}/T\Vert =O_{P}(1)$, and $\pmb{\Theta}=\mathbf{U}^{(J)} 
\mathbf{S}^{(J)}/\sqrt{T}$. Hence, the above analysis also explains where
the power of LM-type tests come from under a general setup since these tests
(e.g., \citealp{chen2014detecting}) try to test if the covariance matrix of
the larger finite-dimensional set of estimated factors is stable.

Define the $J\times J$ symmetric positive definite matrix: 
\begin{equation*}
\pmb{\Delta}=\plim(TN)^{-1}\mathbf{S}^{(J)}(\mathbf{V}^{(J)})^{\top }\mathbb{%
F}^{\top }\mathbb{F}\mathbf{V}^{(J)}\mathbf{S}^{(J)}.
\end{equation*}
Let $\mathbb{V}$ be an $\widetilde{r}\times \widetilde{r}$ diagonal matrix
containing the $\widetilde{r}$ largest eigenvalues of $\pmb{\Delta}$ in
descending order and $\pmb{\Upsilon}$ be a $J\times \widetilde{r}$
eigenvector matrix of $\pmb{\Delta}$ corresponding to the $\widetilde{r}$
eigenvalues in $\mathbb{V}$. The following theorem studies the asymptotic
behavior of $L_{NT}$ under the global alternative.

\begin{theorem}
\label{theorem4} Suppose that Assumptions \ref{Assumption1}, \ref%
{Assumption3}, \ref{Assumption4} and \ref{Assumption6} hold. For $\widetilde{%
r}<J$, there exists a positive constant $c$ such that 
\begin{equation*}
L_{NT}=\overline{\pmb{\theta}}^{\top }(\mathbf{I}_{J}-\mathbf{S}_{J}^{-1}%
\pmb{\Upsilon}\pmb{\Upsilon}^{\top }\mathbf{S}_{J})\pmb{\Omega}(\mathbf{I}%
_{J}-\mathbf{S}_{J}\pmb{\Upsilon}\pmb{\Upsilon}^{\top }\mathbf{S}_{J}^{-1})%
\overline{\pmb{\theta}}+o_{P}(1)\geq c+o_{P}(1),
\end{equation*}%
where 
\begin{eqnarray*}
\mathbf{S}_{J} &=&\lim_{(N,T)\rightarrow \infty }\mathrm{diag}(s_{TN,1}/%
\sqrt{TN},\ldots ,s_{TN,J}/\sqrt{TN}),  \notag \\
\pmb{\Omega} &=&\plim\frac{1}{T^{2}h}\sum_{t,s=1}^{T}\mathcal{F}_{t}\mathcal{%
F}_{s}^{\top }K((t-s)/(Th)),
\end{eqnarray*}%
$\overline{\pmb{\theta}}=\frac{1}{N}\sum_{i=1}^{N}\pmb{\theta}_{i}$ with $%
\pmb{\theta}_{i}$ being the $i^{th}$ column of $\pmb{\Theta}^{\top }$, and $%
\mathcal{F}_{t}$ is the $t^{th}$ column of $\mathcal{F}^{\top }$.
\end{theorem}

We make three comments on Theorem \ref{theorem4}. First, Theorem \ref%
{theorem4} implies that $\widehat{L}_{T}$ defined in (\ref{L_T_hat})
generally diverges to infinity at a rate $NTh^{1/2}$ if one selects $%
\widetilde{r}<J$ under the alternative. Second, Theorem \ref{theorem4}
explains where the power of residual-based tests comes from. In the presence
of $k$ breaks, $J\coloneqq$rank$\left( \mathbb{A}\right) \leq k\cdot r.$
Therefore as long as one specifies a factor model with $\widetilde{r}<J$ (so
that the true rank of the represented model is under-specified), the
residual-based tests would have power in detecting the global alternative.
Third, in connection with the results in Section \ref{Sec3.2}, Theorem \ref%
{theorem4} also indicates that we can propose sequential procedure to test
the composite null hypotheses 
\begin{equation*}
\mathbb{H}_{0}^{\prime }:r=j\quad \text{for}\quad j\in \lbrack r_{\max }]
\end{equation*}%
to identify the number of factors practically if a low rank representation
is believed to be true. The estimated value of factor number, $\widehat{r}$,
will be the value of $j$ when we fail to reject the null at the first time.
Of course, in the presence of abrupt structural changes, $\widehat{r}$ would
over-estimate the underlying number of factors, $r$, for the original factor
model.

To proceed, we emphasize that as explained in 
\citet[pp.
310--311]{fu2023testing}, the information criterion of \cite{su2017time}
still offers a consistent estimation of $r$ rather than $J$. We therefore
suggest that one always turns to the information criterion of \cite%
{su2017time} when estimating $r$ to avoid any power loss.

\subsection{Simulating the Critical Value}

Note that $\widehat{L}_{T}$ is asymptotically pivotal under $\mathbb{H}_{0}$
defined in (\ref{Eq2.2}). It is well known that kernel-based nonparametric
tests may not perform well in finite samples if one uses the normal critical
values and they can be sensitive to the choice of bandwidth. To improve the
finite sample performance of the test, we propose a simulation-based scheme
for the selection of the critical value. A similar procedure has also been
adopted by \cite{zhang2012inference} for the same purpose in the context of
testing univariate time-varying regression models.

The algorithm of a simulation-assisted testing procedure is as follows.

\begin{description}
\item[Step 1:] Use $\{\mathbf{x}_{t}\}$ to estimate the constant factor
model, and compute $\widehat{L}_{T}$ based on Proposition \ref{Coro1}.

\item[Step 2:] Generate i.i.d. $\widetilde{r}$-dimensional standard normal
random vectors $\{\mathbf{f}_{t}^{\ast }\}_{t=1}^{T}$ and $\{\pmb{\lambda}
_{i}^{\ast }\}_{i=1}^{N}$, and i.i.d. $N$-dimensional standard normal random
vectors $\{\pmb{\varepsilon}_{t}^{\ast }\}$, and generate\footnote{%
In our own experiments, we also generate i.i.d. $N$-dimensional standard
normal random vectors $\{\mathbf{x}_{t}^{*}\}_{t=1}^{T}$ to compute the
simulated test statistics, which yields similar simulation results.} $%
x_{it}^{\ast }= \pmb{\lambda}_{i}^{\ast ,\top }\mathbf{f}_{t}^{\ast
}+\varepsilon _{it}^{\ast }$.

\item[Step 3:] Compute the bootstrap statistic $\widehat{L}_{T}^{b}$ in the
same way as $\widehat{L}_{T}$ using the simulated sample $\{x_{it}^{\ast }\}$%
.

\item[Step 4:] Repeat Steps 2-3 $B$ times to obtain $B$ test statistics $\{%
\widehat{L}_{T}^{b}\}_{b=1}^{B}$, as well as its empirical quantile $%
\widehat{q}_{1-\alpha }$. We reject the null hypothesis \eqref{Eq2.2} at the
significance level $\alpha $ if $\widehat{L}_{T}>\widehat{q}_{1-\alpha }$.
\end{description}

By Theorems \ref{theorem1} and \ref{theorem2}, $\widehat{L}_{T}^{b}$ and $%
\widehat{L}_{T}$ have the same asymptotic distribution under the null. This
suggests that instead of using the normal critical value, we can obtain the
simulated critical value $\widehat{q}_{1-\alpha }$ for $\widehat{L}_{T}$.
Under the global alternatives, $\widehat{L}_{T}^{b}$ remains to be
asymptotically normal whereas $\widehat{L}_{T}$ diverges to infinity in
probability.

Note that in the above simulation procedure, we do not need to use the
estimated factors and loadings under the null to generate $x_{it}^{\ast }.$
This is similar in spirit to the fixed-regressor wild bootstrap of \cite%
{hansen2000testing}, where one does not need to mimic the exact dependence
structure in the bootstrap world and may still obtain a valid bootstrap test
that shares the same limit null distribution as the original test statistic.
For our test, one may theoretically provide the approximation rate for the
simulation procedure by establishing a deep Gaussian approximation theory
for a quadratic form of panel data with complex dependence structure, based
on which one can judge whether the simulation can provide a better
approximation to the finite sample distribution than the asymptotic normal.
But this certainly goes beyond the scope of the current paper, and has to be
left for future research.

\section{Numerical Studies \label{Sec4}}

In this section, we conduct Monte Carlo simulations and consider a real data
example to examine the theoretical findings.

\subsection{Simulation \label{Sec4.1}}

In this section, we examine the finite sample performance of the proposed
testing procedure via extensive simulation experiments. We also compare our
test with the parametric tests of \cite{chen2014detecting} and \cite%
{han2015tests} (referred to as CDG and HI, respectively) designed for a
single structural break with an unknown break date in the factor loadings
and the nonparametric tests of \cite{su2020testing} and \cite{fu2023testing}
(referred to as SW and FHW, respectively) that allow for both single or
multiple abrupt breaks and smooth changes under the alternative. We refer
interested readers to these studies for detailed implementation, which is
omitted here to save space.

Throughout the numerical studies, the Bartlett kernel is adopted. When
estimating the long-run covariance, we follow \cite{stock2020introduction}
and set the bandwidth $\ell $ to be $\lceil 0.75T^{1/3}\rceil $, where $%
\lceil \cdot \rceil $ denotes a ceiling function. For the bandwidth $h$, we
use the rule of thumb bandwidth $h=(TN)^{-1/5}$. In Section \ref{App.B3} of
the online supplement, we also examine the performance of our test for
different choices of bandwidth sequences, and the results reveal that our
test is not sensitive to the choice of bandwidth. In addition, we consider
the case with $\widetilde{r}=2,\ldots ,5$ to verify the theoretical results
in Sections \ref{Sec3.2}--\ref{Sec3.4}. For comparison, we also examine the
performance of the above tests with the number of factors estimated by using 
\cite{su2017time}'s information criteria $IC_{h1}$.

\subsubsection{Size Performance}


We first study the size performance of the proposed test. We generate the
data by using the following high-dimensional factor model with two common
factors 
\begin{equation*}
\mathbf{x}_{t}=\pmb{\Lambda}_{t}\mathbf{f}_{t}+\pmb{\varepsilon}_{t},
\end{equation*}
where $\mathbf{f}_{t}=0.5+0.3\mathbf{f}_{t-1}+N(\mathbf{0},\mathbf{I}_{2})$.
To examine the size performance, we consider the following designs for the
factor loading $\pmb{\lambda}_{it}$ and the error term $\varepsilon _{it}$: 
\begin{eqnarray*}
&&\text{DGP.S1:}\quad \pmb{\lambda}_{it}=\pmb{\lambda}_{i0},\quad %
\pmb{\varepsilon}_{t}\sim N(\mathbf{0},\pmb{\Sigma}_{\varepsilon }),\quad %
\pmb{\Sigma}_{\varepsilon }=\{0.3^{|i-j|}\}_{N\times N},  \notag \\
&&\text{DGP.S2:}\quad \pmb{\lambda}_{it}=\pmb{\lambda}_{i0},\quad %
\pmb{\varepsilon}_{t}=0.2\pmb{\varepsilon}_{t-1}+N(\mathbf{0},\mathbf{I}
_{N}),  \notag \\
&&\text{DGP.S3:}\quad \pmb{\lambda}_{it}=\pmb{\lambda}_{i0},\quad %
\pmb{\varepsilon}_{t}=0.2\pmb{\varepsilon}_{t-1}+N(\mathbf{0},\pmb{\Sigma}
_{\varepsilon }),\quad \pmb{\Sigma}_{\varepsilon }=\{0.3^{|i-j|}\}_{N\times
N},
\end{eqnarray*}
where each element of $\pmb{\lambda}_{i0}$ is from $N(1,1)$. Here,
DGP.S1--DGP.S3 satisfy the null hypothesis of time-invariant factor loading
and is used to examine the size performance of the test statistics.
Specifically, DGP.S3 examines the performance under both time series
autocorrelation (TSA) and cross-sectional dependence (CSD), while DGP.S1
assumes the absence of TSA and DGP.S2 does not include CSD. For each data
generating process (DGP), we simulate 1000 data sets with $T=50,100,200$ and 
$N=50,100$ respectively, and conduct the test as aforementioned. For the
simulation-assisted testing procedure, we set $B=1000$.

\begin{table}[tbp]
\caption{Empirical Rejection Rates for DGP.S1--DGP.S3}
\label{Table1.1}\centering
\vspace{-3mm} \setlength{\tabcolsep}{4pt} \renewcommand{\arraystretch}{0.95} 
\scalebox{0.9}{\begin{tabular}{l cc ccccc | ccccc}
\hline\hline
& &  & \multicolumn{5}{c}{$N = 50$} & \multicolumn{5}{c}{$N = 100$} \\
\hline
& & $\widetilde{r}$ &PSY & CDG & HI& SW & FHW &PSY & CDG & HI& SW & FHW\\
\multirow{15}{*}{DGP.S1} 
&\multirow{5}{*}{$T = 50$}  & 2            &0.039  & 0.006  & 0.001  & 0.088  & 0.039  & 0.037  & 0.009  & 0.003  & 0.127  & 0.040  \\
& & 3            &0.031  & 0.003  & 0.000  & 0.096  & 0.043  & 0.049  & 0.001  & 0.000  & 0.150  & 0.031  \\
& & 4            &0.047  & 0.000  & 0.000  & 0.126  & 0.032  & 0.037  & 0.000  & 0.000  & 0.167  & 0.030  \\
& & 5            &0.040  & 0.001  & 0.000  & 0.108  & 0.028  & 0.031  & 0.000  & 0.000  & 0.150  & 0.028  \\
& & $\widehat{r}$&0.039  & 0.000  & 0.000  & 0.119  & 0.029  & 0.038  & 0.000  & 0.000  & 0.131  & 0.027  \\
\cline{3-13} 
& \multirow{5}{*}{$T = 100$}  & 2            &     0.062  & 0.019  & 0.005  & 0.070  & 0.041  & 0.053  & 0.023  & 0.004  & 0.086  & 0.050  \\
& & 3            &     0.049  & 0.007  & 0.000  & 0.069  & 0.051  & 0.053  & 0.009  & 0.000  & 0.105  & 0.041  \\
& & 4            &     0.049  & 0.003  & 0.000  & 0.080  & 0.035  & 0.055  & 0.004  & 0.000  & 0.127  & 0.043  \\
& & 5            &     0.054  & 0.001  & 0.000  & 0.098  & 0.045  & 0.042  & 0.002  & 0.000  & 0.160  & 0.041  \\
& & $\widehat{r}$&     0.053  & 0.022  & 0.007  & 0.089  & 0.036  & 0.046  & 0.031  & 0.008  & 0.083  & 0.056  \\
\cline{3-13}                     
                                                                                 
& \multirow{5}{*}{$T = 200$}  & 2            &    0.050  & 0.023  & 0.015  & 0.064  & 0.048  & 0.048  & 0.041  & 0.024  & 0.077  & 0.060  \\
&& 3            &    0.040  & 0.019  & 0.003  & 0.054  & 0.045  & 0.049  & 0.016  & 0.004  & 0.080  & 0.054  \\
&& 4            &    0.042  & 0.011  & 0.000  & 0.073  & 0.044  & 0.041  & 0.010  & 0.000  & 0.087  & 0.048  \\
&& 5            &    0.040  & 0.008  & 0.000  & 0.067  & 0.025  & 0.046  & 0.003  & 0.000  & 0.107  & 0.044  \\
&& $\widehat{r}$&    0.047  & 0.032  & 0.024  & 0.077  & 0.045  & 0.048  & 0.030  & 0.030  & 0.074  & 0.060  \\
\hline
\multirow{15}{*}{DGP.S2}	
&	\multirow{5}{*}{$T = 50$} & 2            &     0.081 & 0.011  & 0.000  & 0.606  & 0.087  & 0.119  & 0.012  & 0.000  & 0.728  & 0.071  \\
&& 3            &     0.079 & 0.004  & 0.000  & 0.505  & 0.069  & 0.083  & 0.001  & 0.000  & 0.631  & 0.064  \\
&& 4            &     0.050  & 0.001  & 0.000  & 0.317  & 0.072  & 0.078  & 0.001  & 0.000  & 0.437  & 0.050  \\
&& 5            &     0.055 & 0.000  & 0.000  & 0.175  & 0.056  & 0.049  & 0.000  & 0.000  & 0.114  & 0.050  \\
&& $\widehat{r}$&     0.058  & 0.001  & 0.000  & 0.290  & 0.056  & 0.064  & 0.000  & 0.000  & 0.326  & 0.063  \\
\cline{3-13} 
&\multirow{5}{*}{$T = 100$}  & 2            &    0.068 & 0.033  & 0.009  & 0.823  & 0.095  & 0.093  & 0.028  & 0.008  & 0.956  & 0.089  \\
& & 3            &    0.078 & 0.011  & 0.000  & 0.787  & 0.097  & 0.072  & 0.012  & 0.000  & 0.952  & 0.094  \\
& & 4            &    0.055 & 0.004  & 0.000  & 0.743  & 0.087  & 0.061  & 0.003  & 0.000  & 0.910  & 0.083  \\
& & 5            &    0.066 & 0.001  & 0.000  & 0.682  & 0.073  & 0.060  & 0.002  & 0.000  & 0.842  & 0.079  \\
& & $\widehat{r}$&    0.069  & 0.017  & 0.013  & 0.855  & 0.101  & 0.089  & 0.018  & 0.005  & 0.949  & 0.086  \\
\cline{3-13}                                                                                                                                                            
&\multirow{5}{*}{$T = 200$}   & 2            &     0.065  & 0.039  & 0.030  & 0.916  & 0.105  & 0.047  & 0.037  & 0.034  & 0.996  & 0.101  \\
& & 3            &     0.061  & 0.028  & 0.007  & 0.927  & 0.111  & 0.059  & 0.021  & 0.007  & 1.000  & 0.098  \\
& & 4            &     0.062  & 0.012  & 0.001  & 0.904  & 0.100  & 0.050  & 0.017  & 0.000  & 0.999  & 0.086  \\
& & 5            &     0.062  & 0.004  & 0.000  & 0.893  & 0.110  & 0.044  & 0.011  & 0.000  & 0.997  & 0.095  \\
& & $\widehat{r}$&     0.064  & 0.028  & 0.020  & 0.934  & 0.105  & 0.048  & 0.034  & 0.027  & 0.996  & 0.106  \\
\hline
\multirow{15}{*}{DGP.S3} 
&\multirow{5}{*}{$T = 50$}  & 2            &    0.099 & 0.011  & 0.000  & 0.575  & 0.089  & 0.109  & 0.013  & 0.001  & 0.720  & 0.062  \\
& & 3            &    0.098 & 0.005  & 0.000  & 0.431  & 0.066  & 0.078  & 0.002  & 0.000  & 0.597  & 0.069  \\
& & 4            &    0.084 & 0.000  & 0.000  & 0.308  & 0.055  & 0.105  & 0.001  & 0.000  & 0.360  & 0.054  \\
& & 5            &    0.084 & 0.000  & 0.000  & 0.185  & 0.056  & 0.083  & 0.000  & 0.000  & 0.152  & 0.051  \\
& & $\widehat{r}$&    0.085  & 0.001  & 0.000  & 0.241  & 0.060  & 0.089  & 0.000  & 0.000  & 0.261  & 0.055  \\
\cline{3-13} 
&\multirow{5}{*}{$T = 100$}  & 2            &    0.078 & 0.022  & 0.007  & 0.761  & 0.094  & 0.079  & 0.016  & 0.010  & 0.939  & 0.104  \\
& & 3            &    0.079 & 0.007  & 0.000  & 0.722  & 0.078  & 0.066  & 0.010  & 0.000  & 0.932  & 0.102  \\
& & 4            &    0.074 & 0.004  & 0.000  & 0.665  & 0.071  & 0.063  & 0.001  & 0.000  & 0.881  & 0.092  \\
& & 5            &    0.062 & 0.000  & 0.000  & 0.590  & 0.066  & 0.071  & 0.001  & 0.000  & 0.766  & 0.092  \\
& & $\widehat{r}$&    0.078  & 0.008  & 0.009  & 0.792  & 0.094  & 0.074  & 0.011  & 0.004  & 0.942  & 0.096  \\
\cline{3-13}                                                                                                                                                           
&\multirow{5}{*}{$T =200$}    & 2            &    0.066  & 0.037  & 0.032  & 0.893  & 0.102  & 0.056  & 0.029  & 0.024  & 0.993  & 0.109  \\
& & 3            &    0.062  & 0.027  & 0.007  & 0.867  & 0.086  & 0.069  & 0.020  & 0.006  & 0.991  & 0.104  \\
& & 4            &    0.059  & 0.014  & 0.001  & 0.849  & 0.095  & 0.056  & 0.016  & 0.000  & 0.991  & 0.091  \\
& & 5            &    0.065  & 0.010  & 0.000  & 0.824  & 0.085  & 0.044  & 0.006  & 0.000  & 0.986  & 0.098  \\
& & $\widehat{r}$&    0.060  & 0.033  & 0.034  & 0.896  & 0.089  & 0.056  & 0.033  & 0.022  & 0.994  & 0.106  \\
\hline\hline
\end{tabular}}
\end{table}

We first examine the size performance of our test (referred to as PSY).
Table \ref{Table1.1} reports the empirical rejection rates based on 1000
replications at the 5\% nominal level. We summarize some important findings
from these tables. First, PSY tends to be slightly oversized (undersized)
with the presence (absence) of time series correlation. But as the sample
size $T$ increases, the empirical rejection rates of PSY get closer to the
nominal levels. Second, when $\widetilde{r}$ varies from $2$ to $5$, the
size of PSY remains stable if the sample size is not so small. This verifies
the theoretical findings in Theorems \ref{theorem1}--\ref{theorem2}. Third,
CDG and HI are quite undersized, especially when the factor number is over
specified and the sample size is small, which is consistent with the
statements in \citet[p. 359]{BKW2021}. Fourth, SW and FHW performs
reasonably well in DGP.S1 in the absence of serial dependence in the error
terms as required by their theories, but the sizes of SW and FHW lose
control in the presence of serial dependence. Fifth, SW (resp. FHW) tends to
be oversized (resp. undersized) when the factor number is over specified and
the sample size is relatively small.

\subsubsection{Local Power Performance}

To examine the local power performance, we first consider the following
sequence of local alternatives with smooth structural changes 
\begin{equation*}
\pmb{\lambda}_{it}=\pmb{\lambda}_{i0}+10\times a_{TN}\times \lbrack \mathcal{%
G}(10\tau _{t};0.1,(1,3,7,9)^{\top }),0]^{\top },
\end{equation*}%
where $a_{TN}=(TN)^{-1/2}h^{-1/4}$ and $\mathcal{G}(y;\varsigma ,\pmb{\beta}%
)=[1+\exp (-\varsigma \prod_{l=1}^{p}(y-\beta _{l}))]^{-1}$ is the logistic
function with the scale parameter $\varsigma $ and location parameter vector 
$\pmb{\beta}=[\beta _{1},\ldots ,\beta _{p}]^{\top }$. The above setting has
non-monotonic smooth structural changes, which is the same as DGP.P3 in \cite%
{fu2023testing}. We also consider the following designs for the error term $%
\varepsilon _{it}$: 
\begin{eqnarray*}
&&\text{DGP.L1:}\quad \pmb{\varepsilon}_{t}\sim N(\mathbf{0},\pmb{\Sigma}%
_{\varepsilon }),\quad \pmb{\Sigma}_{\varepsilon }=\{0.3^{|i-j|}\}_{N\times
N},\notag \\
&&\text{DGP.L2:}\quad \pmb{\varepsilon}_{t}=0.2\pmb{\varepsilon}_{t-1}+N(%
\mathbf{0},\mathbf{I}_{N}),\notag \\
&&\text{DGP.L3:}\quad \pmb{\varepsilon}_{t}=0.2\pmb{\varepsilon}_{t-1}+N(%
\mathbf{0},\pmb{\Sigma}_{\varepsilon }),\quad \pmb{\Sigma}_{\varepsilon
}=\{0.3^{|i-j|}\}_{N\times N}.
\end{eqnarray*}%
In each case, $\mathbf{f}_{t}$ and $\pmb{\lambda}_{i0}$ are generated in the
same way as in DGP.S1--DGP.S3.

\begin{table}[tbp]
\caption{Empirical Rejection Rates for DGP.L1--DGP.L3}
\label{Table2.1}\centering
\vspace{-3mm} \setlength{\tabcolsep}{4pt} \renewcommand{\arraystretch}{0.95} 
\scalebox{0.9}{\begin{tabular}{c cc ccccc | ccccc}
\hline\hline
& &  & \multicolumn{5}{c}{$N = 50$} & \multicolumn{5}{c}{$N = 100$} \\
\hline
& & $\widetilde{r}$ &PSY & CDG & HI& SW & FHW &PSY & CDG & HI& SW & FHW\\
\multirow{15}{*}{DGP.L1} 
&\multirow{5}{*}{$T = 50$} & 2            &       0.285  & 0.008  & 0.001  & 0.166  & 0.354  & 0.241  & 0.010  & 0.000  & 0.184  & 0.371  \\
&& 3            &    0.192  & 0.001  & 0.000  & 0.137  & 0.208  & 0.170  & 0.000  & 0.000  & 0.149  & 0.249  \\
&& 4            &    0.169  & 0.000  & 0.000  & 0.145  & 0.120  & 0.147  & 0.001  & 0.000  & 0.187  & 0.154  \\
&& 5            &    0.135  & 0.000  & 0.000  & 0.131  & 0.079  & 0.120  & 0.000  & 0.000  & 0.157  & 0.122  \\
&& $\widehat{r}$&    0.154  & 0.000  & 0.000  & 0.132  & 0.116  & 0.129  & 0.000  & 0.000  & 0.154  & 0.124  \\
\cline{3-13} 
&\multirow{5}{*}{$T = 100$} & 2            &     0.289  & 0.016  & 0.013  & 0.158  & 0.443  & 0.255  & 0.021  & 0.010  & 0.174  & 0.491  \\
&& 3            &     0.228  & 0.012  & 0.000  & 0.140  & 0.309  & 0.226  & 0.014  & 0.000  & 0.166  & 0.341  \\
&& 4            &     0.182  & 0.002  & 0.000  & 0.156  & 0.201  & 0.207  & 0.003  & 0.000  & 0.165  & 0.257  \\
&& 5            &     0.157  & 0.001  & 0.000  & 0.135  & 0.157  & 0.224  & 0.000  & 0.000  & 0.191  & 0.215  \\
&& $\widehat{r}$&     0.280  & 0.032  & 0.004  & 0.187  & 0.428  & 0.253  & 0.023  & 0.010  & 0.162  & 0.484  \\
\cline{3-13}                                                                                                                                                                             
&\multirow{5}{*}{$T = 200$}  & 2            &    0.291  & 0.031  & 0.019  & 0.206  & 0.595  & 0.312  & 0.027  & 0.017  & 0.165  & 0.614  \\
&& 3            &    0.279  & 0.022  & 0.003  & 0.147  & 0.427  & 0.279  & 0.026  & 0.007  & 0.155  & 0.499  \\
&& 4            &    0.257  & 0.013  & 0.000  & 0.114  & 0.309  & 0.301  & 0.014  & 0.000  & 0.144  & 0.396  \\
&& 5            &    0.234  & 0.004  & 0.000  & 0.133  & 0.242  & 0.264  & 0.003  & 0.000  & 0.152  & 0.326  \\
&& $\widehat{r}$&    0.295  & 0.036  & 0.021  & 0.208  & 0.589  & 0.317  & 0.039  & 0.023  & 0.168  & 0.606  \\
\hline
\multirow{15}{*}{DGP.L2}&\multirow{5}{*}{$T = 50$} & 2            &      0.292  & 0.013  & 0.001  & 0.667  & 0.491  & 0.316  & 0.012  & 0.000  & 0.785  & 0.503  \\
&& 3            &    0.185  & 0.003  & 0.000  & 0.540  & 0.289  & 0.188  & 0.001  & 0.000  & 0.696  & 0.317  \\
&& 4            &    0.119  & 0.000  & 0.000  & 0.369  & 0.187  & 0.178  & 0.001  & 0.000  & 0.440  & 0.216  \\
&& 5            &    0.125  & 0.000  & 0.000  & 0.185  & 0.128  & 0.130  & 0.000  & 0.000  & 0.163  & 0.173  \\
&& $\widehat{r}$&    0.126  & 0.000  & 0.000  & 0.350  & 0.131  & 0.147  & 0.000  & 0.000  & 0.328  & 0.192  \\
\cline{3-13} 
&\multirow{5}{*}{$T = 100$} & 2            &     0.323  & 0.019  & 0.010  & 0.897  & 0.637  & 0.323  & 0.021  & 0.006  & 0.985  & 0.686  \\
&& 3            &     0.206  & 0.005  & 0.000  & 0.854  & 0.450  & 0.241  & 0.006  & 0.002  & 0.967  & 0.543  \\
&& 4            &     0.186  & 0.002  & 0.000  & 0.821  & 0.341  & 0.206  & 0.004  & 0.000  & 0.936  & 0.445  \\
&& 5            &     0.158  & 0.000  & 0.000  & 0.730  & 0.274  & 0.160  & 0.000  & 0.000  & 0.865  & 0.360  \\
&& $\widehat{r}$&     0.316  & 0.018  & 0.006  & 0.918  & 0.620  & 0.319  & 0.011  & 0.003  & 0.978  & 0.623  \\
\cline{3-13}                                                                                              
                                                                                
&\multirow{4}{*}{$T = 200$}   & 2            &      0.363  & 0.029  & 0.023  & 0.976  & 0.716  & 0.373  & 0.034  & 0.025  & 0.999  & 0.796  \\
&& 3            &    0.277  & 0.026  & 0.006  & 0.965  & 0.603  & 0.301  & 0.026  & 0.007  & 0.999  & 0.695  \\
&& 4            &    0.279  & 0.021  & 0.000  & 0.937  & 0.498  & 0.275  & 0.015  & 0.000  & 0.999  & 0.595  \\
&& 5            &    0.231  & 0.006  & 0.000  & 0.933  & 0.422  & 0.288  & 0.004  & 0.000  & 0.998  & 0.534  \\
&& $\widehat{r}$&    0.354  & 0.034  & 0.026  & 0.975  & 0.762  & 0.368  & 0.027  & 0.023  & 0.999  & 0.798  \\
\hline
\multirow{15}{*}{DGP.L3}
&\multirow{5}{*}{$T = 50$} & 2            &     0.225  & 0.008  & 0.000  & 0.255  & 0.254  & 0.218  & 0.008  & 0.001  & 0.283  & 0.294  \\
&& 3            &    0.189  & 0.001  & 0.000  & 0.287  & 0.152  & 0.161  & 0.000  & 0.000  & 0.413  & 0.212  \\
&& 4            &    0.127  & 0.000  & 0.000  & 0.245  & 0.118  & 0.134  & 0.000  & 0.000  & 0.335  & 0.164  \\
&& 5            &    0.114  & 0.000  & 0.000  & 0.155  & 0.082  & 0.135  & 0.000  & 0.000  & 0.145  & 0.136  \\
&& $\widehat{r}$&    0.122  & 0.001  & 0.000  & 0.200  & 0.107  & 0.145  & 0.000  & 0.000  & 0.236  & 0.143  \\
\cline{3-13} 

&\multirow{5}{*}{$T = 100$} & 2            &    0.179  & 0.023  & 0.007  & 0.410  & 0.370  & 0.217  & 0.023  & 0.007  & 0.506  & 0.427  \\
&& 3            &     0.158  & 0.006  & 0.000  & 0.465  & 0.254  & 0.171  & 0.006  & 0.000  & 0.687  & 0.326  \\
&& 4            &     0.154  & 0.002  & 0.000  & 0.476  & 0.221  & 0.151  & 0.003  & 0.000  & 0.716  & 0.268  \\
&& 5            &     0.134  & 0.000  & 0.000  & 0.424  & 0.169  & 0.126  & 0.000  & 0.000  & 0.616  & 0.225  \\
&& $\widehat{r}$&     0.161  & 0.015  & 0.003  & 0.476  & 0.305  & 0.213  & 0.009  & 0.002  & 0.574  & 0.372  \\
\cline{3-13}                                                                                                
&\multirow{5}{*}{$T = 200$}   & 2            &     0.220  & 0.036  & 0.026  & 0.550  & 0.450  & 0.185  & 0.026  & 0.016  & 0.648  & 0.509  \\
&& 3            &     0.174  & 0.020  & 0.005  & 0.621  & 0.337  & 0.156  & 0.020  & 0.005  & 0.840  & 0.409  \\
&& 4            &     0.197  & 0.013  & 0.000  & 0.660  & 0.267  & 0.153  & 0.011  & 0.000  & 0.887  & 0.355  \\
&& 5            &     0.144  & 0.004  & 0.000  & 0.679  & 0.212  & 0.169  & 0.006  & 0.000  & 0.909  & 0.300  \\
&& $\widehat{r}$&     0.212  & 0.025  & 0.010  & 0.550  & 0.446  & 0.187  & 0.035  & 0.017  & 0.687  & 0.484  \\
\hline\hline
\end{tabular}}
\end{table}

\begin{table}[tbp]
\caption{Empirical Rejection Rates for DGP.L4--DGP.L6}
\label{Table2.2}\centering
\vspace{-3mm} \setlength{\tabcolsep}{4pt} \renewcommand{\arraystretch}{0.95} 
\scalebox{0.9}{\begin{tabular}{c cc ccccc | ccccc}
\hline\hline
& &  & \multicolumn{5}{c}{$N = 50$} & \multicolumn{5}{c}{$N = 100$} \\
\hline
& & $\widetilde{r}$ &PSY & CDG & HI& SW & FHW &PSY & CDG & HI& SW & FHW\\
\multirow{15}{*}{DGP.L4}
&\multirow{5}{*}{$T = 50$} & 2            &    0.250  & 0.014  & 0.001  & 0.130  & 0.269  & 0.195  & 0.006  & 0.003  & 0.139  & 0.284  \\
&& 3            &    0.176  & 0.000  & 0.000  & 0.118  & 0.156  & 0.173  & 0.002  & 0.000  & 0.153  & 0.176  \\
&& 4            &    0.137  & 0.001  & 0.000  & 0.133  & 0.102  & 0.131  & 0.000  & 0.000  & 0.186  & 0.117  \\
&& 5            &    0.152  & 0.000  & 0.000  & 0.143  & 0.070  & 0.120  & 0.000  & 0.000  & 0.254  & 0.086  \\
&& $\widehat{r}$&    0.148  & 0.000  & 0.000  & 0.139  & 0.073  & 0.183  & 0.002  & 0.000  & 0.147  & 0.188  \\
\cline{3-13} 
&\multirow{5}{*}{$T = 100$} & 2            &    0.265  & 0.013  & 0.010  & 0.118  & 0.383  & 0.224  & 0.023  & 0.010  & 0.125  & 0.387  \\
&& 3            &    0.234  & 0.012  & 0.000  & 0.101  & 0.244  & 0.197  & 0.010  & 0.000  & 0.111  & 0.275  \\
&& 4            &    0.181  & 0.002  & 0.000  & 0.112  & 0.196  & 0.164  & 0.002  & 0.000  & 0.146  & 0.217  \\
&& 5            &    0.167  & 0.002  & 0.000  & 0.104  & 0.126  & 0.147  & 0.000  & 0.000  & 0.159  & 0.177  \\
&& $\widehat{r}$&    0.257  & 0.020  & 0.011  & 0.112  & 0.373  & 0.219  & 0.020  & 0.011  & 0.127  & 0.385  \\
\cline{3-13}                                                                                                                                                                
&                  
\multirow{5}{*}{$T = 200$} & 2            &    0.298  & 0.035  & 0.034  & 0.132  & 0.429  & 0.250  & 0.037  & 0.018  & 0.099  & 0.431  \\
&& 3            &    0.249  & 0.043  & 0.006  & 0.103  & 0.294  & 0.187  & 0.030  & 0.005  & 0.084  & 0.323  \\
&& 4            &    0.239  & 0.021  & 0.000  & 0.093  & 0.202  & 0.168  & 0.009  & 0.000  & 0.095  & 0.252  \\
&& 5            &    0.220  & 0.007  & 0.000  & 0.076  & 0.145  & 0.145  & 0.008  & 0.000  & 0.096  & 0.202  \\
&& $\widehat{r}$&    0.285  & 0.038  & 0.025  & 0.128  & 0.413  & 0.246  & 0.039  & 0.017  & 0.098  & 0.433  \\
\hline
\multirow{15}{*}{DGP.L5}
&\multirow{5}{*}{$T = 50$} & 2            &    0.224  & 0.009  & 0.002  & 0.655  & 0.490  & 0.278  & 0.008  & 0.000  & 0.833  & 0.539  \\
&& 3            &    0.195  & 0.000  & 0.000  & 0.652  & 0.357  & 0.200  & 0.002  & 0.000  & 0.855  & 0.380  \\
&& 4            &    0.135  & 0.000  & 0.000  & 0.609  & 0.256  & 0.161  & 0.001  & 0.000  & 0.803  & 0.300  \\
&& 5            &    0.127  & 0.000  & 0.000  & 0.549  & 0.190  & 0.148  & 0.000  & 0.000  & 0.702  & 0.209  \\
& &$\widehat{r}$&    0.128  & 0.000  & 0.000  & 0.548  & 0.186  & 0.264  & 0.005  & 0.000  & 0.851  & 0.441  \\
\cline{3-13} 
&\multirow{5}{*}{$T = 100$} & 2            &    0.245  & 0.016  & 0.005  & 0.812  & 0.623  & 0.276  & 0.020  & 0.004  & 0.956  & 0.660  \\
&& 3            &    0.219  & 0.024  & 0.000  & 0.812  & 0.496  & 0.178  & 0.011  & 0.000  & 0.971  & 0.537  \\
&& 4            &    0.176  & 0.002  & 0.000  & 0.788  & 0.418  & 0.192  & 0.002  & 0.000  & 0.964  & 0.443  \\
&& 5            &    0.178  & 0.002  & 0.000  & 0.771  & 0.363  & 0.188  & 0.001  & 0.000  & 0.958  & 0.390  \\
&& $\widehat{r}$&    0.233  & 0.014  & 0.007  & 0.816  & 0.617  & 0.278  & 0.029  & 0.003  & 0.976  & 0.594  \\
\cline{3-13}                                                                                                 
&\multirow{5}{*}{$T = 200$}                                                                                   
& 2            &     0.246  & 0.036  & 0.025  & 0.908  & 0.668  & 0.248  & 0.034  & 0.023  & 0.993  & 0.679  \\
&& 3            &     0.243  & 0.029  & 0.008  & 0.911  & 0.593  & 0.221  & 0.025  & 0.010  & 0.996  & 0.620  \\
&& 4            &     0.208  & 0.016  & 0.001  & 0.906  & 0.513  & 0.236  & 0.019  & 0.000  & 0.999  & 0.558  \\
&& 5            &     0.204  & 0.011  & 0.000  & 0.905  & 0.441  & 0.223  & 0.008  & 0.000  & 0.996  & 0.488  \\
&& $\widehat{r}$&     0.232  & 0.037  & 0.026  & 0.912  & 0.637  & 0.244  & 0.035  & 0.021  & 0.999  & 0.665  \\
\hline
\multirow{15}{*}{DGP.L6}
&\multirow{5}{*}{$T = 50$} & 2            &    0.262  & 0.010  & 0.002  & 0.612  & 0.342  & 0.282  & 0.008  & 0.000  & 0.807  & 0.357  \\
&& 3            &    0.219  & 0.004  & 0.000  & 0.615  & 0.198  & 0.217  & 0.001  & 0.000  & 0.821  & 0.243  \\
&& 4            &    0.135  & 0.000  & 0.000  & 0.532  & 0.119  & 0.152  & 0.000  & 0.000  & 0.754  & 0.154  \\
&& 5            &    0.152  & 0.000  & 0.000  & 0.494  & 0.076  & 0.131  & 0.000  & 0.000  & 0.611  & 0.112  \\
&& $\widehat{r}$&    0.148  & 0.000  & 0.000  & 0.499  & 0.113  & 0.263  & 0.000  & 0.000  & 0.817  & 0.288  \\
\cline{3-13} 
&\multirow{5}{*}{$T = 100$} & 2            &    0.301  & 0.023  & 0.009  & 0.746  & 0.404  & 0.277  & 0.016  & 0.007  & 0.920  & 0.438  \\
&& 3            &    0.217  & 0.017  & 0.000  & 0.736  & 0.255  & 0.215  & 0.011  & 0.000  & 0.940  & 0.333  \\
&& 4            &    0.222  & 0.007  & 0.000  & 0.720  & 0.178  & 0.163  & 0.006  & 0.000  & 0.940  & 0.259  \\
&& 5            &    0.162  & 0.000  & 0.000  & 0.678  & 0.137  & 0.159  & 0.003  & 0.000  & 0.935  & 0.194  \\
&& $\widehat{r}$&    0.297  & 0.020  & 0.011  & 0.742  & 0.403  & 0.274  & 0.012  & 0.011  & 0.937  & 0.385  \\
\cline{3-13}                                                                                                 
&\multirow{4}{*}{$T = 200$}                                                                                   
& 2            &    0.326  & 0.033  & 0.023  & 0.848  & 0.448  & 0.250  & 0.035  & 0.023  & 0.982  & 0.473  \\
&& 3            &    0.224  & 0.039  & 0.008  & 0.836  & 0.317  & 0.222  & 0.032  & 0.007  & 0.985  & 0.369  \\
&& 4            &    0.213  & 0.023  & 0.001  & 0.830  & 0.235  & 0.207  & 0.017  & 0.001  & 0.986  & 0.300  \\
&& 5            &    0.223  & 0.006  & 0.000  & 0.817  & 0.180  & 0.157  & 0.009  & 0.000  & 0.978  & 0.267  \\
&& $\widehat{r}$&    0.315  & 0.038  & 0.025  & 0.848  & 0.443  & 0.246  & 0.033  & 0.022  & 0.988  & 0.473  \\
\hline\hline
\end{tabular}}
\end{table}

Table \ref{Table2.1} reports the empirical rejection rates for the 5\% test
for DGP.L1--DGP.L3, which are smooth time-varying factor models with the
break size shrinking to zero. We summarize some findings from these tables.
First, as $\widetilde{r}$ increases from $2$ to $5$, the local power of PSY
generally decreases for fixed $\left( N,T\right) $, which indicates the side
effect of overspecifying the umber of factors. In theory (see Theorem \ref%
{theorem2}, Proposition \ref{Coro1} and Theorem \ref{theorem3}), the
over-specification of the number of factors does not yield any difference in
the asymptotic normal distribution under the null hypothesis and under the
local alternatives. However, in finite sample study, especially when the
sample size is small, it is expected that the empirical powers are affected.
Although the powers are slightly different, the numerical behavior is indeed
consistent with the theoretical results established in the paper when the
sample size is relatively large. Specifically, for DGP.L1, when $N=100$ and $%
T=200$, Table \ref{Table2.1} shows that the empirical rejection rates of PSY
test are 0.312, 0.279, 0.301, 0.264 for $\widetilde{r}=2,\ldots ,5$, in
which the differences are quite small. Second, CDG and HI seem to have no
local power against smooth structural changes, which is not surprising since
these two tests are designed to detect the presence of big structural
breaks. Third, for DGP.L1, PSY has a larger (resp. smaller) power than SW
(resp. FHW), which is consistent with the theory that our test is able to
detect a class of local alternatives at the rate $(TN)^{-1/2}h^{-1/4}$,
converging to 0 faster (resp. slower) than the rate $%
T^{-1/2}N^{-1/4}h^{-1/4} $ (resp. $T^{-1/2}N^{-1/2}$) in SW (resp. FHW).
Fourth, when the sample size is small ($T=50$), \cite{su2017time}'s
information criteria $IC_{h1}$ tends to overestimate the number of common
factors (indeed the effective sample size is $Th$, which is extremely
small), in which case PSY test has better local power performance than FHW.
Fifth, the local power of SW is very close to one if there is TSA in error
terms, which is consistent with the simulation results for DGP.S2--DGP.S3.

To examine the local power performance of the proposed test under abrupt
structural breaks, we also consider the following design for $\pmb{\lambda}%
_{it}$ 
\begin{equation*}
\pmb{\lambda}_{it}=\left\{ 
\begin{array}{ll}
\pmb{\lambda}_{i0}, & \text{for}\ 1\leq t\leq T/2 \\ 
\pmb{\lambda}_{i0}+a_{NT}\times \mathbf{b}, & \text{for}\ T/2+1\leq t\leq T%
\end{array}%
\right. .
\end{equation*}%
where $a_{TN}=(TN)^{-1/2}h^{-1/4}$ and $\mathbf{b}=[2,2]^{\top }$. This
setup is the same as DGP.P1 in \cite{fu2023testing}. In addition, similar to
DGP.L1--DGP.L3, we also consider the three designs for the error term $%
\varepsilon _{it}$ denoted by DGP.L4--DGP.L6 respectively, while $\mathbf{f}%
_{t}$ and $\pmb{\lambda}_{i0}$ are generated in the same way as in
DGP.S1--DGP.S3. Table \ref{Table2.2} reports the empirical rejection rates
for the 5\% test for DGP.L4--DGP.L6, from which we find similar results as
those for DGP.L1--DGP.L3.

\subsubsection{Global Power Performance}

To examine the global power, we consider the following specification for the
time-varying loadings: 
\begin{equation*}
\pmb{\lambda}_{it}=\left\{ 
\begin{array}{ll}
\pmb{\lambda}_{i0}, & \text{for}\ 1\leq t\leq T/2 \\ 
\pmb{\lambda}_{i0}+\mathbf{b}, & \text{for}\ T/2+1\leq t\leq T%
\end{array}%
\right. ,
\end{equation*}%
where $\mathbf{b}=[0.25,0.25]^{\top }$. We also consider the following
designs for the error term $\varepsilon _{it}$: 
\begin{eqnarray*}
&&\text{DGP.G1:}\quad \pmb{\varepsilon}_{t}\sim N(\mathbf{0},\pmb{\Sigma}%
_{\varepsilon }),\quad \pmb{\Sigma}_{\varepsilon }=\{0.3^{|i-j|}\}_{N\times
N}, \notag\\
&&\text{DGP.G2:}\quad \pmb{\varepsilon}_{t}=0.2\pmb{\varepsilon}_{t-1}+N(%
\mathbf{0},\mathbf{I}_{N}), \notag\\
&&\text{DGP.G3:}\quad \pmb{\varepsilon}_{t}=0.2\pmb{\varepsilon}_{t-1}+N(%
\mathbf{0},\pmb{\Sigma}_{\varepsilon }),\quad \pmb{\Sigma}_{\varepsilon
}=\{0.3^{|i-j|}\}_{N\times N}.
\end{eqnarray*}%
In each case, $\pmb{\lambda}_{i0}$ and $\mathbf{f}_{t}$ are generated in the
same way as DGP.S1--DGP.S3. Clearly, we have a one-time big break in each of
the above three DGPs. In this case, each DGP can be written as a
time-invariant factor model with three factors since the break in the factor
loadings is not idiosyncratic across $i$'s (see \citealp{han2015tests} for
more discussions on this type of abrupt structural break). Hence, our test
would have power to detect the breaks as long as $\widetilde{r}<3$.

\begin{table}[tbp]
\caption{Empirical Rejection Rates for DGP.G1--DGP.G3}
\label{Table3.1}\centering
\vspace{-3mm} \setlength{\tabcolsep}{4pt} \renewcommand{\arraystretch}{0.95} 
\scalebox{0.9}{\begin{tabular}{c cc ccccc | ccccc}
\hline\hline
& &  & \multicolumn{5}{c}{$N = 50$} & \multicolumn{5}{c}{$N = 100$} \\
\hline
& & $\widetilde{r}$ &PSY & CDG & HI& SW & FHW &PSY & CDG & HI& SW & FHW\\
\multirow{15}{*}{DGP.G1} 
&\multirow{5}{*}{$T = 50$} & 2            &     0.988  & 0.015  & 0.000  & 0.940  & 0.999  & 0.997  & 0.011  & 0.000  & 0.991  & 1.000  \\
&& 3            &     0.195  & 0.012  & 0.000  & 0.411  & 0.166  & 0.144  & 0.018  & 0.000  & 0.621  & 0.073  \\
&& 4            &     0.121  & 0.002  & 0.000  & 0.403  & 0.087  & 0.084  & 0.001  & 0.000  & 0.596  & 0.045  \\
&& 5            &     0.123  & 0.001  & 0.000  & 0.409  & 0.061  & 0.099  & 0.001  & 0.000  & 0.648  & 0.038  \\
&& $\widehat{r}$&     0.124  & 0.002  & 0.000  & 0.414  & 0.078  & 0.090  & 0.002  & 0.000  & 0.616  & 0.060  \\
\cline{3-13} 
&\multirow{5}{*}{$T = 100$} 
& 2            &     1.000  & 0.028  & 0.010  & 1.000  & 1.000  & 1.000  & 0.026  & 0.014  & 1.000  & 1.000  \\
&& 3            &     0.138  & 0.942  & 0.001  & 0.389  & 0.116  & 0.079  & 0.983  & 0.003  & 0.560  & 0.066  \\
&& 4            &     0.111  & 0.004  & 0.000  & 0.383  & 0.078  & 0.068  & 0.008  & 0.000  & 0.582  & 0.058  \\
&& 5            &     0.104  & 0.001  & 0.000  & 0.392  & 0.060  & 0.065  & 0.001  & 0.000  & 0.611  & 0.055  \\
&& $\widehat{r}$&     0.943  & 0.095  & 0.029  & 0.983  & 0.954  & 1.000  & 0.031  & 0.013  & 1.000  & 1.000  \\
\cline{3-13}                                                                                                 
&\multirow{5}{*}{$T = 200$}                                                                                   
& 2            &    1.000  & 0.033  & 0.032  & 1.000  & 1.000  & 1.000  & 0.053  & 0.035  & 1.000  & 1.000  \\
&& 3            &    0.141  & 1.000  & 0.992  & 0.364  & 0.132  & 0.078  & 1.000  & 1.000  & 0.548  & 0.068  \\
&& 4            &    0.102  & 0.020  & 0.306  & 0.348  & 0.099  & 0.071  & 0.013  & 0.485  & 0.543  & 0.054  \\
&& 5            &    0.115  & 0.007  & 0.000  & 0.327  & 0.085  & 0.077  & 0.004  & 0.000  & 0.536  & 0.043  \\
&& $\widehat{r}$&    1.000  & 0.049  & 0.239  & 1.000  & 1.000  & 1.000  & 0.060  & 0.042  & 1.000  & 1.000  \\
\hline
\multirow{15}{*}{DGP.G2} 
&\multirow{5}{*}{$T = 50$} 
& 2            &    0.997  & 0.014  & 0.000  & 0.996  & 1.000  & 1.000  & 0.005  & 0.002  & 0.998  & 1.000  \\
&& 3            &    0.196  & 0.006  & 0.000  & 0.794  & 0.234  & 0.126  & 0.012  & 0.000  & 0.891  & 0.128  \\
&& 4            &    0.091  & 0.001  & 0.000  & 0.745  & 0.155  & 0.113  & 0.000  & 0.000  & 0.856  & 0.095  \\
&& 5            &    0.118  & 0.000  & 0.000  & 0.699  & 0.107  & 0.082  & 0.000  & 0.000  & 0.821  & 0.069  \\
&& $\widehat{r}$&    0.088  & 0.001  & 0.000  & 0.696  & 0.161  & 0.096  & 0.004  & 0.000  & 0.828  & 0.092  \\
\cline{3-13} 
&\multirow{5}{*}{$T = 100$} 
& 2            &    1.000  & 0.034  & 0.015  & 1.000  & 1.000  & 1.000  & 0.031  & 0.020  & 1.000  & 1.000  \\
&& 3            &    0.101  & 0.938  & 0.003  & 0.851  & 0.217  & 0.082  & 0.970  & 0.002  & 0.917  & 0.093  \\
&& 4            &    0.076  & 0.009  & 0.000  & 0.843  & 0.184  & 0.050  & 0.009  & 0.000  & 0.916  & 0.093  \\
&& 5            &    0.100  & 0.001  & 0.000  & 0.837  & 0.148  & 0.056  & 0.000  & 0.000  & 0.913  & 0.098  \\
&& $\widehat{r}$&    0.927  & 0.161  & 0.034  & 0.976  & 0.939  & 0.961  & 0.153  & 0.031  & 0.984  & 0.971  \\
\cline{3-13}                                                                                                 
&\multirow{5}{*}{$T = 200$}                                                                                   
& 2            &    1.000  & 0.045  & 0.029  & 1.000  & 1.000  & 1.000  & 0.041  & 0.044  & 1.000  & 1.000  \\
&& 3            &    0.069  & 1.000  & 0.999  & 0.885  & 0.206  & 0.057  & 1.000  & 1.000  & 0.942  & 0.101  \\
&& 4            &    0.055  & 0.017  & 0.402  & 0.893  & 0.172  & 0.059  & 0.027  & 0.525  & 0.944  & 0.108  \\
&& 5            &    0.050  & 0.011  & 0.000  & 0.891  & 0.144  & 0.047  & 0.007  & 0.000  & 0.951  & 0.108  \\
&& $\widehat{r}$&    1.000  & 0.050  & 0.036  & 1.000  & 1.000  & 1.000  & 0.045  & 0.043  & 1.000  & 1.000  \\
\hline
\multirow{15}{*}{DGP.G3} 
&\multirow{4}{*}{$T = 50$} 
& 2            &     0.977  & 0.011  & 0.002  & 0.993  & 0.997  & 0.998  & 0.013  & 0.002  & 0.999  & 1.000  \\
&& 3            &     0.215  & 0.005  & 0.000  & 0.710  & 0.164  & 0.173  & 0.016  & 0.000  & 0.829  & 0.109  \\
&& 4            &     0.136  & 0.000  & 0.000  & 0.661  & 0.076  & 0.124  & 0.001  & 0.000  & 0.797  & 0.073  \\
&& 5            &     0.104  & 0.000  & 0.000  & 0.603  & 0.052  & 0.106  & 0.000  & 0.000  & 0.742  & 0.045  \\
&& $\widehat{r}$&     0.117  & 0.000  & 0.000  & 0.649  & 0.056  & 0.149  & 0.004  & 0.000  & 0.823  & 0.113  \\
\cline{3-13} 
&\multirow{5}{*}{$T = 100$} 
& 2            &     1.000  & 0.025  & 0.015  & 1.000  & 1.000  & 1.000  & 0.023  & 0.010  & 1.000  & 1.000  \\
&& 3            &     0.138  & 0.904  & 0.002  & 0.793  & 0.166  & 0.096  & 0.960  & 0.003  & 0.874  & 0.092  \\
& &4            &     0.112  & 0.006  & 0.000  & 0.781  & 0.125  & 0.073  & 0.005  & 0.000  & 0.890  & 0.081  \\
& &5            &     0.079  & 0.000  & 0.000  & 0.761  & 0.103  & 0.100  & 0.000  & 0.000  & 0.885  & 0.067  \\
& &$\widehat{r}$&     0.769  & 0.311  & 0.008  & 0.966  & 0.808  & 0.834  & 0.193  & 0.019  & 0.988  & 0.855  \\
\cline{3-13}                                                                                                 
&\multirow{5}{*}{$T = 200$}                                                                                   
& 2            &    1.000  & 0.060  & 0.041  & 1.000  & 1.000  & 1.000  & 0.044  & 0.040  & 1.000  & 1.000  \\
&& 3            &    0.091  & 1.000  & 0.992  & 0.812  & 0.173  & 0.051  & 1.000  & 1.000  & 0.926  & 0.101  \\
& &4            &    0.077  & 0.017  & 0.234  & 0.822  & 0.150  & 0.058  & 0.015  & 0.439  & 0.938  & 0.083  \\
& &5            &    0.094  & 0.015  & 0.000  & 0.823  & 0.136  & 0.043  & 0.006  & 0.000  & 0.945  & 0.072  \\
& &$\widehat{r}$&    1.000  & 0.063  & 0.031  & 1.000  & 1.000  & 1.000  & 0.050  & 0.041  & 1.000  & 1.000  \\
\hline\hline
\end{tabular}}
\end{table}

Table \ref{Table3.1} reports the empirical rejection rates for the 5\% test
for DGP.G1--DGP.G3. We summarize some findings here. First, our test is
powerful to reject the null hypothesis when $\widetilde{r}<J=3$, while the
empirical rejection rates of both PSY and FHW converge to the 5\% nominal
level when $\widetilde{r}\geq J$, which is consistent with our theory in
Section \ref{Sec3.4}. In addition, when the sample size is small ($T=50$), 
\cite{su2017time}'s information criteria $IC_{h1}$ tends to overestimate the
number of common factors, in which case PSY and FHW tend to lose power.
Second, CDG and HI have desirable power performance if the factor number is
specified to be $\widetilde{r}=J=3$ and the sample size is not so small, and
their tests lose power when $\widetilde{r}>3$ since their tests are designed
to detect the structural breaks in the variance of common factors (c.f.,
p.39 in \citealp{chen2014detecting}). Third, SW is generally more powerful
than PSY and FHW for detecting big structural breaks when the factor number
is over-specified, although these empirical rejection rates decrease with
the increase of sample size when the factor number is over-specified.

\subsection{An Empirical Study \label{Sec4.2}}

In this section, we revisit the U.S. macroeconomic dataset assembled by \cite%
{MN2016}. Specifically, we examine the time period from October 2003 to
September 2023. Therefore, in total, we have $T=240$ along the time
dimension. After removing variables with missing values, $N=127$ macro
variables are available for our investigation.

As in the simulation, we implement PSY, CDG, HI, SW and FHW tests. We let $%
\widetilde{r}=1,\ldots ,8$ in what follows. To examine the necessity of
accounting for CSD and TSA, for any given $\widetilde{r}$, we conduct the
PCA and obtain $\{\widetilde{\varepsilon }_{it}\}$ as in Section \ref{Sec2.1}%
. Subsequently, we consider the following two tests:

\begin{enumerate}
\item Examine TSA by conducting the Ljung-Box Q-test (referred to as LBQ)
for each time series (i.e., $\{\widetilde{\varepsilon }_{i1},\ldots , 
\widetilde{\varepsilon }_{iT}\}$), and report the percentage of individuals
having non-negligible autocorrelation at the 5\% level;

\item Examine CSD by conducting the CD test of \cite{Pesaran2004} on $%
\widetilde{\varepsilon }_{it}$'s, and report the test statistics.
\end{enumerate}

The results for the CD and LBQ tests are reported in the last two columns in
Table \ref{table.emp1}. It is obvious that both cross-sectional dependence
and serial dependence are present in the residuals.

Now, we focus on the structural change tests in Table \ref{table.emp1} ,
where $\surd $ and $\times $ denote failing to detect break and detecting
breaks, respectively at the 5\% nominal level. First, most of the test
results suggest the presence of structural changes in the factor models for $%
\widetilde{r}\in \{1,\ldots,7\}$. Second, as in the simulations, PSY and SW
are largely consistent in terms of rejecting the null with only one
exception occurring at $\widetilde{r}=7$. Third, PSY suggests that there are
eight common factors in this U.S. macroeconomic dataset if one wants to use
a conventional time-invariant factor model to characterize the U.S. macro
economy, which is consistent with the argument of \cite{MN2016}. Fourth, FHW
tends to reject the null hypothesis for any $\widetilde{r} \in
\{1,\ldots,8\} $, possibly due to its higher power to detect smaller
structural changes or its tendency to over-reject in the presence of serial
correlations in the error terms. To further confirm the above findings, we
use \cite{su2017time}'s information criteria $IC_{h1}$ to determine the
number of factors as documented in the end of Section \ref{Sec3.4}, which
indicates two common factors. Putting these numerical evidences together, we
conclude that the dataset features a factor model with time-varying loadings.

\begin{table}[tbp]
\caption{Test Results}
\label{table.emp1}\centering
\setlength{\tabcolsep}{4pt} \renewcommand{\arraystretch}{1} 
\begin{tabular}{rccccclrr}
\hline\hline
$\widetilde{r}$ & PSY & FHW & SW & CDG & HI &  & CD & LBQ \\ 
1 & $\times $ & $\times $ & $\times $ & \checkmark & $\times $ &  & 341.32 & 
0.99 \\ 
2 & $\times $ & $\times $ & $\times $ & $\times $ & $\times $ &  & 116.06 & 
0.99 \\ 
3 & $\times $ & $\times $ & $\times $ & $\times $ & $\times $ &  & 142.27 & 
1.00 \\ 
4 & $\times $ & $\times $ & $\times $ & $\times $ & $\times $ &  & 113.55 & 
1.00 \\ 
5 & $\times $ & $\times $ & $\times $ & $\times $ & $\times $ &  & 71.25 & 
1.00 \\ 
6 & $\times $ & $\times $ & $\times $ & $\times $ & $\times $ &  & 87.95 & 
0.99 \\ 
7 & $\times $ & $\times $ & \checkmark & $\times $ & $\times $ &  & 100.87 & 
0.97 \\ 
8 & \checkmark & $\times$ & \checkmark & \checkmark & \checkmark &  & 126.20
& 0.94 \\ \hline\hline
\end{tabular}%
\end{table}

\section{Conclusion \label{Sec5}}

In this paper, we propose an easy-to-implement residual-based specification
testing procedure to detect structural changes in factor models, which is
powerful against both smooth and abrupt structural changes with unknown
break dates. The proposed test is robust against the over-specified number
of factors under the null and local alternatives, and serially and
cross-sectionally correlated error processes. A new central limit theorem is
derived for the quadratic forms of panel data with dependence over both
dimensions, thereby filling a gap in the literature. We establish the
asymptotic properties of the proposed test statistic under the null and a
sequence of Pitman local alternative, and accordingly develop a
simulation-based scheme to select critical value in order to improve finite
sample performance. Through extensive simulations and a real-world
application, we confirm our theoretical results and demonstrate that the
proposed test exhibits desirable size and power in practice.

Some extensions are possible. First, despite the superb asymptotic power
properties of the proposed test among all existing kernel-based
nonparametric smooth test, its power decreases in the case when all factors
have zero mean in the model. It is worthwhile to consider a remedy to avoid
such a power loss. Second, we only consider the case of pervasive factors in
the paper. Given the recent burgeoning literature on weak factors (%
\citealp{bai2023approximate}), it is interesting to extend the current work
to allow for the presence of weak factors. We leave these topics for future
research.

\bigskip

\renewcommand{\theequation}{A.\arabic{equation}}
\renewcommand{\thesection}{A.\arabic{section}}
\renewcommand{\thefigure}{A.\arabic{figure}}
\renewcommand{\thetable}{A.\arabic{table}}
\renewcommand{\thelemma}{A.\arabic{lemma}}
\renewcommand{\theremark}{A.\arabic{remark}}
\renewcommand{\thecorollary}{A.\arabic{corollary}}
\renewcommand{\theassumption}{A.\arabic{assumption}}

\setcounter{equation}{0}
\setcounter{lemma}{0}
\setcounter{section}{0}
\setcounter{table}{0}
\setcounter{figure}{0}
\setcounter{remark}{0}
\setcounter{corollary}{0}
\setcounter{assumption}{0}

\medskip

\noindent { {\Large\bf Appendix A} }

\medskip

{\small This appendix includes three sections. Section \ref{App.A1}
summaries the details of the relevant tests in the literature. Section \ref{App.A2} states some technical lemmas that are used in the proofs of the main
results. Section \ref{App.A3} contains the proofs of the main results in
Section \ref{Sec3}. }

\section{Relevant Tests\label{App.A1}}

{\small In what follows, we summarize the details of the relevant tests in
the literature, and use the following notations. }

\begin{enumerate}
\item {\small Paper (by publication year): BE stands for \cite%
{breitung2011testing}, CDG stands for \cite{chen2014detecting}, HI stands
for \cite{han2015tests}, YT stands for \cite{yamamoto_tanaka2015}, SW\_17
stands for \cite{su2017time}, SW\_20 stands for \cite{su2020testing}, BKW
stands for \cite{BKW2021}, FHW stands for \cite{fu2023testing}, BDH stands
for \cite{BDH2024}. }

\item {\small Framework: `C' stands for constant factor model, `NTV' stands
for factor model with nonparametric time-varying loadings, `A' stands for
abrupt structural breaks. }

\item {\small Dependence: `CSD' stands for allowing for cross-sectional
dependence in the error terms, `TSA' stands for allowing for serial
correlations in the error terms. }

\item {\small Local: `LP' stands for some local departure from the constant
loadings. Whenever the local power analysis is not available, we use $\times 
$. }

\item {\small Rate: This column reports that the rate at which the test
statistic diverges to infinity under the global alternatives. The square
root of one over this rate denotes the rate at which the local alternatives
is allowed to converge to the null and the test has power to detect. }

\item {\small Robustness: `Yes' indicates that the asymptotic properties
remain valid under the null hypothesis or local alternatives, even when the
number of factors is over-specified. The symbol `$\times $' denotes that the
corresponding test is not robust to an over-specified factor number (e.g.,
see p. 359 in \citealp{BKW2021}). Meanwhile, `$?$' represents cases where it
is unclear whether the test is robust due to the absence of related
discussions or theorems in the literature. }
\end{enumerate}

 {\footnotesize \centering
\begin{tabular}{lcrrrrr}
\hline\hline
Paper & Framework & Dependence & Approach & Local & Rate & Robustness \\ 
\multicolumn{1}{l|}{BE} & \multicolumn{1}{r|}{C vs A} & \multicolumn{1}{r|}{
CSD} & \multicolumn{1}{r|}{LM test} & \multicolumn{1}{r|}{$\times $} & 
\multicolumn{1}{r|}{$T$} & $\times $ \\ 
\multicolumn{1}{l|}{CDG} & \multicolumn{1}{r|}{C vs A} & \multicolumn{1}{r|}{
CSD + TSA} & \multicolumn{1}{r|}{LM test (via factors)} & 
\multicolumn{1}{r|}{$\times $} & \multicolumn{1}{r|}{$T$} & $\times $ \\ 
\multicolumn{1}{l|}{HI} & \multicolumn{1}{r|}{C vs A} & \multicolumn{1}{r|}{
CSD + TSA} & \multicolumn{1}{r|}{LM test (via factors)} & 
\multicolumn{1}{r|}{$\times $} & \multicolumn{1}{r|}{$T$} & $\times $ \\ 
\multicolumn{1}{l|}{YT} & \multicolumn{1}{r|}{C vs A} & \multicolumn{1}{r|}{
CSD} & \multicolumn{1}{r|}{LM test} & \multicolumn{1}{r|}{LP} & 
\multicolumn{1}{r|}{$T$} & $\times $ \\ 
\multicolumn{1}{l|}{SW\_17} & \multicolumn{1}{r|}{C vs (NTV+A)} & 
\multicolumn{1}{r|}{CSD} & \multicolumn{1}{r|}{Quadratic form statistic} & 
\multicolumn{1}{r|}{LP} & \multicolumn{1}{r|}{$T(Nh)^{1/2}$} & $?$ \\ 
\multicolumn{1}{l|}{SW\_20} & \multicolumn{1}{r|}{C vs (NTV+A)} & 
\multicolumn{1}{r|}{CSD} & \multicolumn{1}{r|}{Quadratic form statistic} & 
\multicolumn{1}{r|}{LP} & \multicolumn{1}{r|}{$T(Nh)^{1/2}$} & $?$ \\ 
\multicolumn{1}{l|}{BKW} & \multicolumn{1}{r|}{C vs A} & \multicolumn{1}{r|}{
CSD + TSA} & \multicolumn{1}{r|}{LM test (via factors)} & 
\multicolumn{1}{r|}{$\times $} & \multicolumn{1}{r|}{$T$} & $\times $ \\ 
\multicolumn{1}{l|}{FHW} & \multicolumn{1}{r|}{C vs (NTV+A)} & 
\multicolumn{1}{r|}{CSD} & \multicolumn{1}{r|}{Residual-based test} & 
\multicolumn{1}{r|}{LP} & \multicolumn{1}{r|}{$TN$} & $\times $ \\ 
\multicolumn{1}{l|}{BDH} & \multicolumn{1}{r|}{C vs A} & \multicolumn{1}{r|}{
CSD + TSA} & \multicolumn{1}{r|}{LR test (via factors)} & 
\multicolumn{1}{r|}{$\times $} & \multicolumn{1}{r|}{$T$} & $\times $ \\ 
\multicolumn{1}{l|}{Our paper} & \multicolumn{1}{r|}{C vs (NTV+A)} & 
\multicolumn{1}{r|}{CSD + TSA} & \multicolumn{1}{r|}{Residual-based test} & 
\multicolumn{1}{r|}{LP} & \multicolumn{1}{r|}{$TN(h)^{1/2}$} & Yes \\ 
\hline\hline
\end{tabular}%
} 

\section{Some Technical Lemmas \label{App.A2}}

\begin{lemma}
{\small \label{Lemma1} Suppose that Assumptions \ref{Assumption3}--\ref%
{Assumption4} hold. Then }

{\small \smallskip }

{\small \noindent (a) $|\frac{1}{T}\sum_{t=1}^{T}[\varepsilon
_{it}\varepsilon _{jt}-E(\varepsilon _{it}\varepsilon _{jt})]|_{2}=O(1/\sqrt{%
T})$ for any $1\leq i,j\leq N$; }

{\small \smallskip }

{\small \noindent (b) $|\frac{1}{T}\sum_{t=1}^{T}\mathbf{f}_{t}\overline{%
\varepsilon }_{t, \mathbf{v}}|_{2}=O(1/\sqrt{T})$ for any non-random $%
\mathbf{v}\in \mathbb{R}^{N}$ satisfying $\Vert \mathbf{v}\Vert <\infty $; }

{\small \smallskip }

{\small \noindent (c) $|\frac{1}{TN}\sum_{t=1}^{T}\mathbf{f}%
_{t}\sum_{i=1}^{N}[\varepsilon _{it}\varepsilon _{is}-E(\varepsilon
_{it}\varepsilon _{is})]|_{2}=O(1/ \sqrt{TN})$ for any $1\leq s\leq T$. }
\end{lemma}

\begin{lemma}
{\small \label{Lemma2} Suppose that Assumptions \ref{Assumption1}--\ref%
{Assumption4} hold. Let $\mathcal{E}_{i}$ denote the $i^{th}$ column of $%
\mathcal{E}$. Then }

{\small \smallskip }

{\small \noindent (a) $\Vert \mathbf{K}_{h}\Vert _{2}=O(T)$; }

{\small \smallskip }

{\small \noindent (b) $\Vert \frac{1}{TN}\mathcal{E}\mathcal{E}^{\top }\Vert
=O_{P}(\frac{1}{ \sqrt{T\wedge N}})$; }

{\small \smallskip }

{\small \noindent (c) $\Vert \frac{1}{TN}\pmb{\Lambda}^{\top }\mathcal{E}%
^{\top }\mathbf{F} \Vert =O_{P}(\frac{1}{\sqrt{TN}})$; }

{\small \smallskip }

{\small \noindent (d) $\Vert \frac{1}{TN}\mathcal{E}\pmb{\Lambda}\Vert
=O_{P}(\frac{1}{\sqrt{ TN }})$; }

{\small \smallskip }

{\small \noindent (e) $\Vert \frac{1}{TN}\mathcal{E}\mathcal{E}^{\top }%
\mathbf{F}\Vert =O_{P}( \frac{1}{\sqrt{T\wedge N}})$; }

{\small \smallskip }

{\small \noindent (f) $\Vert \frac{1}{TN}\sum_{i=1}^{N}(1-a_{i})\mathcal{E}%
_{i}^{\top }\mathbf{F}\Vert =O_{P}(\frac{1}{\sqrt{TN}})$; }

{\small \smallskip }

{\small \noindent (g) $\Vert \frac{1}{TN}\sum_{i=1}^{N}(1-a_{i})\mathcal{E}%
_{i}^{\top }\Vert =O_{P}(\frac{1}{\sqrt{TN}})$. }
\end{lemma}

\begin{lemma}
{\small \label{Lemma3} Suppose that Assumptions \ref{Assumption2}--\ref%
{Assumption4} hold. Let $\mathbf{H}=\frac{1}{N}\pmb{\Lambda}^{\top }%
\pmb{\Lambda}\cdot \frac{1}{T}\mathbf{F}^{\top }\widehat{\mathbf{F}}\cdot (%
\frac{1}{TN}\widehat{ \mathbf{V}})^{-1}.$ Then }

{\small \smallskip }

{\small \noindent (a) $\frac{1}{\sqrt{T}}\Vert \widehat{\mathbf{F}}-\mathbf{F%
}\mathbf{H}\Vert =O_{P}(\frac{1}{\sqrt{T\wedge N}})$; }

{\small \smallskip }

{\small \noindent (b) $\Vert \frac{1}{T}\mathbf{F}^{\top }(\widehat{\mathbf{F%
}}-\mathbf{F} \mathbf{H})\Vert =O_{P}(\frac{1}{T\wedge N})$; }

{\small \smallskip }

{\small \noindent (c) $\frac{1}{\sqrt{T}}\Vert \frac{1}{TN}\mathcal{E}%
\mathcal{E}^{\top }( \widehat{\mathbf{F}}-\mathbf{F}\mathbf{H})\Vert =O_{P}(%
\frac{1}{(T\wedge N)^{3/2}})$; }

{\small \smallskip }

{\small \noindent (d) $\Vert \frac{1}{TN}\sum_{i=1}^{N}(1-a_{i})\mathcal{E}%
_{i}^{\top }( \widehat{\mathbf{F}}-\mathbf{F}\mathbf{H})\Vert =O_{P}(\frac{1%
}{N})$. }
\end{lemma}

\begin{lemma}
{\small \label{Lemma4} Suppose Assumption \ref{Assumption4} holds. Let $%
\widetilde{\overline{\varepsilon }}_{t,\mathbf{v}}=E(\overline{\varepsilon }
_{t,\mathbf{v}}\,|\, \mathscr{E}_{t-m,t})$ and $\mathscr{E}_{t-m,t}=(\mathbf{%
e }_{t},\mathbf{e}_{t-1},\ldots ,\mathbf{e}_{t-m})$. Let $d_{m}=\sum_{t=0}^{
\infty }\left( \lambda _{4}^{\varepsilon }(t)\wedge (\sum_{j=m+1}^{\infty
}\lambda _{4}^{\varepsilon ,2}(j))^{1/2}\right) $ and $K_{ts}=K\left( \frac{
t-s}{Th}\right) .$ Define 
\begin{eqnarray*}
&&J_{T}=\frac{1}{N}\sum_{i,j=1}^{N}\sum_{1\leq t<s\leq
T}(1-a_{i})(1-a_{j})\varepsilon _{it}\varepsilon _{js}K_{ts},  \notag \\
&&\widetilde{J}_{T}=\frac{1}{N}\sum_{i,j=1}^{N}\sum_{1\leq t<s\leq
T}(1-a_{i})(1-a_{j})\widetilde{\varepsilon }_{it}\widetilde{\varepsilon }
_{js}K_{ts}.
\end{eqnarray*}
Let $h_{t}=\sum_{j=0}^{\infty }E[\widetilde{\overline{\varepsilon }}
_{t+j,a}\,|\, \mathscr{E}_{t}]$, $H_{t}=h_{t}-E[h_{t}\,|\, \mathscr{E}%
_{t-1}] $ , $M_{T}=\sum_{t=1}^{T}H_{t}\sum_{s=1}^{t-1}K_{ts}H_{s}$, and $%
\widetilde{ \overline{\varepsilon }}_{t,a}=\frac{1}{\sqrt{N}}%
\sum_{i=1}^{N}(1-a_{i}) \widetilde{\varepsilon }_{it}$. Then }

{\small \smallskip }

{\small \noindent (a) $\left\vert \sum_{t=1}^{T}b_{t}(\overline{\varepsilon }%
_{t,\mathbf{v}}- \widetilde{\overline{\varepsilon }}_{t,\mathbf{v}%
})\right\vert _{4}=O\left(
(\sum_{t=1}^{T}b_{t}^{2})^{1/2}(\sum_{j=m+1}^{\infty }\lambda
_{4}^{\varepsilon }(j))\right) $ for any $m\geq 1$, any fixed $\mathbf{v}\in 
\mathbb{R}^{N}$ with $\Vert \mathbf{v}\Vert <\infty $ and any fixed sequence 
$\{b_{t}\}$; }

{\small \smallskip }

{\small \noindent (b) $|J_{T}-E(J_{T})-\widetilde{J}_{T}-E(\widetilde{J}%
_{T})|_{2}=O\left( \sqrt{T}(\sum_{t=1}^{T}K^{2}(\frac{t}{Th}%
))^{1/2}d_{m}\right) $; }

{\small \smallskip }

{\small \noindent (c) $|\widetilde{J}_{T}-E(\widetilde{J}_{T})-M_{T}|_{2}%
\leq O(1)T^{1/2}m\left( \max_{1\leq t\leq T-1}K^{2}(\frac{t}{Th}
)+m\sum_{t=1}^{T-1}\left( K(\frac{t}{Th})-K(\frac{t-1}{Th})\right)
^{2}\right) ^{1/2}.$ }
\end{lemma}

\begin{lemma}
{\small \label{Lemma5} Suppose that Assumptions \ref{Assumption1} and \ref%
{Assumption4} hold. Let $U_{T}=\sum_{t=2}^{T}H_{t}
\sum_{s=1}^{t-1}H_{s}w_{t,s}$ where $w_{t,s}=\frac{1}{T\sqrt{h}}K\left( 
\frac{t-s}{Th}\right) $. Then $U_{T}\rightarrow _{D}N(0,\sigma _{H}^{2}),$
where $\sigma _{H}^{2}=\lim_{T\rightarrow \infty
}\sum_{t=2}^{T}E(H_{t}^{2})\sum_{s=1}^{t-1}E(H_{s}^{2})w_{t,s}^{2}$. }
\end{lemma}

\begin{lemma}
{\small \label{Lemma6} Suppose that Assumptions \ref{Assumption1} and \ref%
{Assumption4} hold. Then $\sqrt{1/h}[Q_{T}-E(Q_{T})]\rightarrow _{D}N(0,2\nu
_{0}\sigma _{\varepsilon ,a}^{4}),$ where 
\begin{eqnarray*}
Q_{T} &=&\frac{1}{TN}\sum_{i,j=1}^{N}\sum_{t,s=1}^{T}(1-a_{i})(1-a_{j})
\varepsilon _{it}\varepsilon _{js}K\left( \frac{t-s}{Th}\right) ,  \notag \\
\sigma _{\varepsilon ,a}^{2} &=&\lim_{(N,T)\rightarrow \infty }\frac{1}{Th}
\sum_{t,s=1}^{T}E(\overline{\varepsilon }_{t,a}\overline{\varepsilon }
_{s,a})K\left( \frac{t-s}{Th}\right) ,\text{ }\nu _{0}=\int_{-1}^{1}K^{2}(u) 
\mathrm{d}u,
\end{eqnarray*}
and $\overline{\varepsilon }_{t,a}=\frac{1}{\sqrt{N}}\sum_{i=1}^{N}(1-a_{i})
\varepsilon _{it}$. }
\end{lemma}

\begin{lemma}
{\small \label{Lemma7} Suppose that $\mathbf{A}$ and $\mathbf{A}+\mathbf{E}$
are $n\times n$ symmetric matrices and that $\mathbf{Q}=(\mathbf{Q}_{1},%
\mathbf{\ Q}_{2})$ with $\mathbf{Q}_{1}$ and $\mathbf{Q}_{2}$ being $n\times
r$ and $n\times (n-r)$, respectively, is an orthogonal matrix such that $%
\normalfont 
\text{span}(\mathbf{Q}_{1})$ is an invariant subspace for $\mathbf{A}$.
Decompose $\mathbf{Q}^{\top }\mathbf{A}\mathbf{Q}$ and $\mathbf{Q}^{\top } 
\mathbf{E}\mathbf{Q}$ as $\mathbf{Q}^{\top }\mathbf{A}\mathbf{Q}=\text{diag}
( \mathbf{D}_{1},\mathbf{D}_{2})$ and $\mathbf{Q}^{\top }\mathbf{E}\mathbf{Q}
=\{\mathbf{E}_{ij}\}_{2\times 2}$. Let $\mathrm{sep}(\mathbf{D}_{1},\mathbf{%
D }_{2})=\min_{\lambda _{1}\in \lambda (\mathbf{D}_{1}),\ \lambda _{2}\in
\lambda (\mathbf{D}_{2})}|\lambda _{1}-\lambda _{2}|$. If $\mathrm{sep}( 
\mathbf{D}_{1},\mathbf{D}_{2})>0$ and $\Vert \mathbf{E}\Vert \leq \mathrm{\
sep }(\mathbf{D}_{1},\mathbf{D}_{2})/5$, then there exists a $(n-r)\times r$
matrix $\mathbf{P}$ with $\Vert \mathbf{P}\Vert \leq 4\Vert \mathbf{E}
_{21}\Vert /\mathrm{sep}(\mathbf{D}_{1},\mathbf{D}_{2})$, such that the
columns of $\mathbf{Q}_{1}^{0}=(\mathbf{Q}_{1}+\mathbf{Q}_{2}\mathbf{P})( 
\mathbf{I}_{r}+\mathbf{P}^{\top }\mathbf{P})^{-1/2}$ define an orthonormal
basis for a subspace that is invariant for $\mathbf{A}+\mathbf{E}$. }
\end{lemma}

\begin{lemma}
{\small \label{Lemma8} Suppose that Assumptions \ref{Assumption1}, \ref%
{Assumption3}, \ref{Assumption4}, and \ref{Assumption5} hold and the
selected number of factors satisfies $\widetilde{r}\leq J$. Let $\pmb{\Delta}%
=\plim(TN)^{-1} \mathbf{S}^{(J)}\mathbf{V}^{(J),\top }\mathbb{F}^{\top }%
\mathbb{F}\mathbf{V} ^{(J)}\mathbf{S}^{(J)}$ be a $J\times J$ symmetric
positive definite matrix, $\mathbb{V}$ be an $\widetilde{r}\times \widetilde{%
r}$ diagonal matrix containing the $\widetilde{r}$ largest eigenvalues of $%
\pmb{\Delta}$ in descending order and $\pmb{\Upsilon}$ be a $J\times 
\widetilde{r}$ eigenvector matrix of $\pmb{\Delta}$ corresponding to the $%
\widetilde{r}$ eigenvalues in $\mathbb{V}$. Then we have }

{\small \smallskip }

{\small \noindent (a) $\Vert \frac{1}{TN}\sum_{i=1}^{N}\mathcal{E}%
_{i}^{\dagger }\Vert =o_{P}(1/\sqrt{T})$, where $\mathcal{E}_{i}^{\dagger }$
is the $i^{th}$ column of $\mathcal{E}^{\dagger ,\top }$; }

{\small \smallskip }

{\small \noindent (b) $\frac{1}{\sqrt{T}}\Vert \widehat{\mathbf{F}}-\mathcal{%
F}\mathcal{H}\Vert =o_{P}(1)$, where $\mathcal{H}=\left( \frac{\pmb{\Theta}%
^{\top }\pmb{\Theta} }{N}\right) \left( \frac{\mathcal{F}^{\top }\widehat{%
\mathbf{F}}}{T}\right) \left( \frac{1}{TN}\widehat{\mathbf{V}}\right) ^{-1}$%
; }

{\small \smallskip }

{\small \noindent (c) $\frac{1}{TN}\widehat{\mathbf{V}}=\mathbb{V}+o_{P}(1)$%
; }

{\small \smallskip }

{\small \noindent (d) $\mathbf{R}_{TN}=(\pmb{\Theta}^{\top }\pmb{\Theta}%
/N)^{1/2}\mathcal{F} ^{\top }\widehat{\mathbf{F}}/T=\pmb{\Upsilon}\mathbb{V}%
^{1/2}+o_{P}(1)$. }
\end{lemma}

\begin{lemma}
{\small \label{Lemma9} Suppose that Assumptions \ref{Assumption1}--\ref%
{Assumption4} hold and let $(T\wedge N) \sqrt{h} \to \infty$. Let $\mathcal{E%
}_{i}^{\ast }$ be the $i^{th}$ column of $\mathcal{E}^{\ast }\coloneqq%
\{\varepsilon _{it}+a_{TN} \mathbf{g}_{it}^{\top }\mathbf{f}_{t}\}_{T\times
N}.$ Let $\mathbf{H}=( \pmb{\Lambda}^{\top }\pmb{\Lambda}/N)(\mathbf{F}%
^{\top }\widehat{\mathbf{F}} /T)(\frac{1}{TN}\widehat{\mathbf{V}})^{-1}$ and 
$\pmb{\Lambda}=[\pmb{\lambda} _{1},\ldots ,\pmb{\lambda}_{p}]^{\top }.$ Then
under the local alternative $\mathbb{H}_{1}(a_{TN})$ with $%
a_{TN}=(TN)^{-1/2}h^{-1/4}$, we have }

{\small \smallskip }

{\small \noindent (a) $\Vert \frac{1}{TN}\mathcal{E}^{\ast }\mathcal{E}%
^{\ast ,\top }\Vert =O_{P}(\frac{1}{\sqrt{T\wedge N}})$; }

{\small \smallskip }

{\small \noindent (b) $\Vert \frac{1}{TN}\pmb{\Lambda}^{\top }\mathcal{E}%
^{\ast ,\top }\mathbf{F}\Vert =O_{P}(\frac{1}{\sqrt{TN}})$; }

{\small \smallskip }

{\small \noindent (c) $\Vert \frac{1}{TN}\mathcal{E}^{\ast }\pmb{\Lambda}%
\Vert =O_{P}(\frac{1}{ \sqrt{TN}})$; }

{\small \smallskip }

{\small \noindent (d) $\Vert \frac{1}{TN}\mathcal{E}^{\ast }\mathcal{E}%
^{\ast ,\top }\mathbf{F} \Vert =O_{P}(\frac{1}{\sqrt{T\wedge N}})$; }

{\small \smallskip }

{\small \noindent (e) $\Vert \frac{1}{TN}\sum_{i=1}^{N}(1-a_{i})\mathcal{E}%
_{i}^{\ast ,\top } \mathbf{F}\Vert =O_{P}(\frac{1}{\sqrt{TN}})$; }

{\small \smallskip }

{\small \noindent (f) $T^{-1/2}\Vert \widehat{\mathbf{F}}-\mathbf{F}\mathbf{H%
}\Vert =O_{P}( \frac{1}{\sqrt{T\wedge N}})$; }

{\small \smallskip }

{\small \noindent (g) $\Vert \frac{1}{T}\mathbf{F}^{\top }(\widehat{\mathbf{F%
}}-\mathbf{F} \mathbf{H})\Vert =O_{P}(\frac{1}{T\wedge N})$; }

{\small \smallskip }

{\small \noindent (h) $\frac{1}{\sqrt{T}}\Vert \frac{1}{TN}\mathcal{E}^{\ast
}\mathcal{E}^{\ast ,\top }(\widehat{\mathbf{F}}-\mathbf{F}\mathbf{H})\Vert
=O_{P}(\frac{1}{ (T\wedge N)^{3/2}})$; }

{\small \smallskip }

{\small \noindent (i) $\left\Vert \frac{1}{TN}\sum_{i=1}^{N}(1-a_{i})%
\mathcal{E}_{i}^{\ast ,\top }(\widehat{\mathbf{F}}-\mathbf{F}\mathbf{H}%
)\right\Vert =O_{P}(1/N)$. }
\end{lemma}

\begin{lemma}
{\small \label{Lemma10} Suppose that Assumptions \ref{Assumption1}--\ref%
{Assumption4} hold with $E(\mathbf{f}_t)=\mathbf{0}$. Let $\mathcal{E}%
_{i}^{\ast }$ is the $i^{th}$ column of $\mathcal{E}^{\ast }\coloneqq%
\{\varepsilon _{it}+a_{TN} \mathbf{g}_{it}^{\top }\mathbf{f}_{t}\}_{T\times
N}.$ Let $\mathbf{H}=( \pmb{\Lambda}^{\top }\pmb{\Lambda}/N)(\mathbf{F}%
^{\top }\widehat{\mathbf{F}} /T)(\frac{1}{TN}\widehat{\mathbf{V}})^{-1}$ and 
$\pmb{\Lambda}=[\pmb{\lambda} _{1},\ldots ,\pmb{\lambda}_{p}]^{\top }.$ Then
under the local alternative $\mathbb{H}_{1}(a_{TN})$ with $%
a_{TN}=N^{-1/2}h^{1/4}$, we have }

{\small \smallskip }

{\small \noindent (a) $\Vert \frac{1}{TN}\mathcal{E}^{\ast }\mathcal{E}%
^{\ast ,\top }\Vert =O_{P}(\frac{1}{\sqrt{T\wedge N}})$; }

{\small \smallskip }

{\small \noindent (b) $\Vert \frac{1}{TN}\pmb{\Lambda}^{\top }\mathcal{E}%
^{\ast ,\top }\mathbf{\ F}\Vert =O_{P}(\frac{1}{\sqrt{TN}})$; }

{\small \smallskip }

{\small \noindent (c) $\Vert \frac{1}{TN}\mathcal{E}^{\ast }\pmb{\Lambda}%
\Vert =O_{P}(\frac{1}{ \sqrt{TN}})$; }

{\small \smallskip }

{\small \noindent (d) $\Vert \frac{1}{TN}\mathcal{E}^{\ast }\mathcal{E}%
^{\ast ,\top }\mathbf{F} \Vert =O_{P}(\frac{1}{\sqrt{T\wedge N}})$; }

{\small \smallskip }

{\small \noindent (e) $\Vert \frac{1}{TN}\sum_{i=1}^{N}(1-a_{i})\mathcal{E}%
_{i}^{\ast ,\top } \mathbf{F}\Vert =O_{P}(\frac{1}{\sqrt{TN}})$; }

{\small \smallskip }

{\small \noindent (f) $T^{-1/2}\Vert \widehat{\mathbf{F}}-\mathbf{F}\mathbf{H%
}\Vert =O_{P}( \frac{1}{\sqrt{T\wedge N}})$; }

{\small \smallskip }

{\small \noindent (g) $\Vert \frac{1}{T}\mathbf{F}^{\top }(\widehat{\mathbf{F%
}}-\mathbf{F} \mathbf{H})\Vert =O_{P}(\frac{1}{T\wedge N})$; }

{\small \smallskip }

{\small \noindent (h) $\frac{1}{\sqrt{T}}\Vert \frac{1}{TN}\mathcal{E}^{\ast
}\mathcal{E}^{\ast ,\top }(\widehat{\mathbf{F}}-\mathbf{F}\mathbf{H})\Vert
=O_{P}(\frac{1}{ (T\wedge N)^{3/2}})$; }

{\small \smallskip }

{\small \noindent (i) $\left\Vert \frac{1}{TN}\sum_{i=1}^{N}(1-a_{i})%
\mathcal{E}_{i}^{\ast ,\top }(\widehat{\mathbf{F}}-\mathbf{F}\mathbf{H}%
)\right\Vert =O_{P}(1/N)$. }
\end{lemma}

\section{Proofs of the Main Results in Section \ref{Sec3} \label{App.A3}}

{\small In this section, we prove the main results in Section \ref{Sec3}. }

\medskip

\begin{proof}[Proof of Theorem \ref{theorem1}]

{\small First, by the proof of Lemma \ref{Lemma3}, we have 
\begin{eqnarray*}
\mathbf{F}-\widehat{\mathbf{F}}\mathbf{H}^{-1}&=&-\frac{1}{TN}\mathbf{F} %
\pmb{\Lambda}^{\top }\mathcal{E}^{\top }\widehat{\mathbf{F}}\pmb{\Sigma}_{ 
\widehat{\mathbf{F}}}^{-1}\pmb{\Sigma}_{\pmb{\Lambda}}^{-1}-\frac{1}{N} 
\mathcal{E}\pmb{\Lambda}\pmb{\Sigma}_{\pmb{\Lambda}}^{-1}-\frac{1}{TN} 
\mathcal{E}\mathcal{E}^{\top }\widehat{\mathbf{F}}\pmb{\Sigma}_{\widehat{ 
\mathbf{F}}}^{-1}\pmb{\Sigma}_{\pmb{\Lambda}}^{-1}  \notag \\
&\eqqcolon& -\mathbf{W}_{1}-\mathbf{W }_{2}-\mathbf{W}_{3},
\end{eqnarray*}
where $\mathbf{H}^{-1}=\frac{1}{TN}\widehat{\mathbf{V}}\pmb{\Sigma}_{ 
\widehat{\mathbf{F}}}^{-1}\pmb{\Sigma}_{\pmb{\Lambda}}^{-1}$, $\pmb{\Sigma}
_{ \widehat{\mathbf{F}}}=\frac{1}{T}\mathbf{F}^{\top }\widehat{\mathbf{F}}$
and $\pmb{\Sigma}_{\pmb{\Lambda}}=\frac{1}{N}\pmb{\Lambda}^{\top } %
\pmb{\Lambda}$ . Thus, by definition, we expand $L_{NT}$ as follows: 
\begin{eqnarray*}
L_{NT} &=&\frac{1}{T^{2}N^{2}}\sum_{i,j=1}^{N}\pmb{\lambda}_{i}^{\top }%
\mathbf{F} ^{\top }\mathbf{M}_{\widehat{\mathbf{F}}}\mathbf{K}_{h}\mathbf{M}%
_{\widehat{ \mathbf{F}}}\mathbf{F}\pmb{\lambda}_{j}+\frac{1}{T^{2}N^{2}}%
\sum_{i,j=1}^{N} \mathcal{E}_{i}^{\top }\mathbf{M}_{\widehat{\mathbf{F}}}%
\mathbf{K}_{h} \mathbf{M}_{\widehat{\mathbf{F}}}\mathcal{E}_{j}  \notag \\
&&+\frac{2}{T^{2}N^{2}}\sum_{i,j=1}^{N}\mathcal{E}_{i}^{\top }\mathbf{M}_{ 
\widehat{\mathbf{F}}}\mathbf{K}_{h}\mathbf{M}_{\widehat{\mathbf{F}}}\mathbf{%
F }\pmb{\lambda}_{j}  \notag \\
&\eqqcolon & L_{T,1}+L_{T,2}+L_{T,3}.
\end{eqnarray*}
}

{\small For $L_{T,1},$ we make the following decomposition: 
\begin{eqnarray*}
L_{T,1} &=&\frac{1}{T^{2}N^{2}}\sum_{i,j=1}^{N}\pmb{\lambda}_{i}^{\top }( 
\mathbf{F}-\widehat{\mathbf{F}}\mathbf{H}^{-1})^{\top }\mathbf{M}_{\widehat{ 
\mathbf{F}}}\mathbf{K}_{h}\mathbf{M}_{\widehat{\mathbf{F}}}(\mathbf{F}- 
\widehat{\mathbf{F}}\mathbf{H}^{-1})\pmb{\lambda}_{j}  \notag \\
&=&\frac{1}{T^{2}}\overline{\pmb{\lambda}}^{\top }\mathbf{W}_{1}^{\top } 
\mathbf{M}_{\widehat{\mathbf{F}}}\mathbf{K}_{h}\mathbf{M}_{\widehat{\mathbf{%
F }}}\mathbf{W}_{1}\overline{\pmb{\lambda}}+\frac{1}{T^{2}}\overline{ %
\pmb{\lambda}}^{\top }\mathbf{W}_{2}^{\top }\mathbf{M}_{\widehat{\mathbf{F}}
} \mathbf{K}_{h}\mathbf{M}_{\widehat{\mathbf{F}}}\mathbf{W}_{2}\overline{ %
\pmb{\lambda}}  \notag \\
&&+\frac{1}{T^{2}}\overline{\pmb{\lambda}}^{\top }\mathbf{W}_{3}^{\top } 
\mathbf{M}_{\widehat{\mathbf{F}}}\mathbf{K}_{h}\mathbf{M}_{\widehat{\mathbf{%
F }}}\mathbf{W}_{3}\overline{\pmb{\lambda}}+\frac{2}{T^{2}}\overline{ %
\pmb{\lambda}}^{\top }\mathbf{W}_{1}^{\top }\mathbf{M}_{\widehat{\mathbf{F}}
} \mathbf{K}_{h}\mathbf{M}_{\widehat{\mathbf{F}}}\mathbf{W}_{2}\overline{ %
\pmb{\lambda}}  \notag \\
&&+\frac{2}{T^{2}}\overline{\pmb{\lambda}}^{\top }\mathbf{W}_{1}^{\top } 
\mathbf{M}_{\widehat{\mathbf{F}}}\mathbf{K}_{h}\mathbf{M}_{\widehat{\mathbf{%
F }}}\mathbf{W}_{3}\overline{\pmb{\lambda}}+\frac{2}{T^{2}}\overline{ %
\pmb{\lambda}}^{\top }\mathbf{W}_{2}^{\top }\mathbf{M}_{\widehat{\mathbf{F}}
} \mathbf{K}_{h}\mathbf{M}_{\widehat{\mathbf{F}}}\mathbf{W}_{3}\overline{ %
\pmb{\lambda}}  \notag \\
&\eqqcolon &L_{T,11}+L_{T,12}+L_{T,13}+2L_{T,14}+2L_{T,15}+2L_{T,16},
\end{eqnarray*}
where $\overline{\pmb{\lambda}}=\frac{1}{N}\sum_{i=1}^{N}\pmb{\lambda}_{i}$.
Below, we consider the terms on the right hand side (r.h.s.) of the last
equation one by one. In particular, we will keep the leading term $L_{T,12}$
and show that the other terms are asymptotically negligible. }

{\small For $L_{T,11}$, by using Lemmas \ref{Lemma2} and \ref{Lemma3}, 
\begin{eqnarray*}
|L_{T,11}| &\leq &\frac{1}{T^{2}N^{2}}\Vert \mathbf{F}^{\top }\mathbf{M}_{ 
\widehat{\mathbf{F}}}\mathbf{K}_{h}\mathbf{M}_{\widehat{\mathbf{F}}}\mathbf{%
F }\Vert _{2}\left\Vert \frac{1}{T}\pmb{\Lambda}^{\top }\mathcal{E}^{\top } 
\widehat{\mathbf{F}}\pmb{\Sigma}_{\widehat{\mathbf{F}}}^{-1}\pmb{\Sigma}_{ %
\pmb{\Lambda}}^{-1}\overline{\pmb{\lambda}}\right\Vert ^{2}  \notag \\
&\leq &\frac{2}{T^{2}N^{2}}\Vert \mathbf{K}_{h}\Vert _{2}\Vert \mathbf{M}_{ 
\widehat{\mathbf{F}}}\Vert _{2}^{2}\Vert \mathbf{F}-\widehat{\mathbf{F}} 
\mathbf{H}^{-1}\Vert ^{2}\left\Vert \frac{1}{T}\pmb{\Lambda}^{\top }\mathcal{%
\ E}^{\top }\mathbf{F}\mathbf{H}\pmb{\Sigma}_{\widehat{\mathbf{F}}}^{-1} %
\pmb{\Sigma}_{\pmb{\Lambda}}^{-1}\overline{\pmb{\lambda}}\right\Vert ^{2} 
\notag \\
&&+\frac{2}{T^{2}N^{2}}\Vert \mathbf{K}_{h}\Vert _{2}\Vert \mathbf{M}_{ 
\widehat{\mathbf{F}}}\Vert _{2}^{2}\Vert \mathbf{F}-\widehat{\mathbf{F}} 
\mathbf{H}^{-1}\Vert ^{2}\left\Vert \frac{1}{T}\pmb{\Lambda}^{\top }\mathcal{%
E}^{\top }(\widehat{\mathbf{F}}-\mathbf{F}\mathbf{H})\pmb{\Sigma}_{ \widehat{
\mathbf{F}}}^{-1}\pmb{\Sigma}_{\pmb{\Lambda}}^{-1}\overline{ \pmb{\lambda}}
\right\Vert ^{2}  \notag \\
&\leq &O_{P}(1)\frac{1}{T^{2}N^{2}}\cdot T\cdot \frac{T}{T\wedge N}\cdot 
\frac{N}{T}+O_{P}(1)\frac{1}{T^{2}N^{2}}\cdot T\cdot \frac{T}{T\wedge N}
\cdot \frac{N}{T}\frac{T}{T\wedge N}  \notag \\
&=&O_{P}\left( \frac{1}{TN(T\wedge N)}+\frac{1}{N(T\wedge N)^{2}}\right)
=o_{P}\left( \frac{1}{TN\sqrt{h}}\right) .
\end{eqnarray*}
For $L_{T,12}$, we have 
\begin{eqnarray*}
L_{T,12} &=&\frac{1}{T^{2}N^{2}}\sum_{i,j=1}^{N}\overline{\pmb{\lambda}}%
^{\top } \pmb{\Sigma}_{\pmb{\Lambda}}^{-1}\pmb{\lambda}_{i}\mathcal{E}%
_{i}^{\top } \mathbf{M}_{\widehat{\mathbf{F}}}\mathbf{K}_{h}\mathbf{M}_{%
\widehat{\mathbf{F}}}\mathcal{E}_{j}\pmb{\lambda}_{j}^{\top }\pmb{\Sigma}_{%
\pmb{\Lambda} }^{-1} \overline{\pmb{\lambda}},
\end{eqnarray*}
which can be merged with $L_{T,2}$ later on. For $L_{T,13}$, we have 
\begin{eqnarray*}
|L_{T,13}| &\leq &O_{P}(1)\frac{2}{T^{2}N^{2}}\Vert \mathbf{M}_{\widehat{ 
\mathbf{F}}}\mathbf{K}_{h}\mathbf{M}_{\widehat{\mathbf{F}}}\Vert
_{2}\left\Vert \frac{1}{T}\mathcal{E}\mathcal{E}^{\top }\mathbf{F}
\right\Vert ^{2}  \notag \\
&&+O_{P}(1)\frac{2}{T^{2}N^{2}}\Vert \mathbf{M}_{\widehat{ \mathbf{F}}}%
\mathbf{K}_{h}\mathbf{M}_{\widehat{\mathbf{F}}}\Vert _{2}\left\Vert \frac{1}{%
T}\mathcal{E}\mathcal{E}^{\top }(\widehat{\mathbf{F}} -\mathbf{F}\mathbf{H}%
)\right\Vert ^{2}  \notag \\
&\leq &O_{P}(1)\frac{2}{T^{2}N^{2}}\cdot T\cdot \frac{N^{2}}{T\wedge N}
+O_{P}(1)\frac{2}{T^{2}N^{2}}\cdot T\cdot \frac{TN^{2}}{(T\wedge N)^{3}} 
\notag \\
&=&O_{P}\left( \frac{1}{T(T\wedge N)}+\frac{1}{(T\wedge N)^{3}}\right)
=o_{P}\left( \frac{1}{TN\sqrt{h}}\right) ,
\end{eqnarray*}
where the last equality holds under Asss\ref{Assumption1}(b). Based on the
developments of $L_{T,11}$, $L_{T,12},$ and $L_{T,13}$ and by the
Cauchy-Schwarz (CS) inequality, $L_{T,14}$, $L_{T,15}$ and $L_{T,16}$ are
all asymptotically negligible. }

{\small Similarly, we can show that 
\begin{equation*}
L_{T,3}=-\frac{2}{T^{2}N^{2}}\sum_{i,j=1}^{N}\mathcal{E}_{i}^{\top }\mathbf{%
M }_{\widehat{\mathbf{F}}}\mathbf{K}_{h}\mathbf{M}_{\widehat{\mathbf{F}}} 
\mathcal{E}_{j}\pmb{\lambda}_{j}^{\top }\pmb{\Sigma}_{\pmb{\Lambda}}^{-1} 
\overline{\pmb{\lambda}}+o_{P}\left( \frac{1}{TN\sqrt{h}}\right) .
\end{equation*}
Then we have $L_{NT}=\frac{1}{T^{2}N^{2}}\sum_{i,j=1}^{N}(1-a_{i})(1-a_{j}) 
\mathcal{E}_{i}^{\top }\mathbf{M}_{\widehat{\mathbf{F}}}\mathbf{K}_{h} 
\mathbf{M}_{\widehat{\mathbf{F}}}\mathcal{E}_{j}+o_{P}((TN\sqrt{h})^{-1}),$
where $a_{i}=\pmb{\lambda}_{i}^{\top }\pmb{\Sigma}_{\pmb{\Lambda}}^{-1} 
\overline{\pmb{\lambda}}.$ }

{\small Now, we observe that 
\begin{eqnarray*}
&&\frac{1}{T^{2}N^{2}}\sum_{i,j=1}^{N}(1-a_{i})(1-a_{j})\mathcal{E}
_{i}^{\top }\mathbf{M}_{\widehat{\mathbf{F}}}\mathbf{K}_{h}\mathbf{M}_{ 
\widehat{\mathbf{F}}}\mathcal{E}_{j}  \notag \\
&=&\frac{1}{T^{2}N^{2}}\sum_{i,j=1}^{N}(1-a_{i})(1-a_{j})\mathcal{E}
_{i}^{\top }\mathbf{K}_{h}\mathcal{E}_{j}+\frac{1}{T^{2}N^{2}}
\sum_{i,j=1}^{N}(1-a_{i})(1-a_{j})\mathcal{E}_{i}^{\top }\frac{1}{T}\widehat{
\mathbf{F}}\widehat{\mathbf{F}}^{\top }\mathbf{K}_{h}\frac{1}{T}\widehat{ 
\mathbf{F}}\widehat{\mathbf{F}}^{\top }\mathcal{E}_{j}  \notag \\
&&-\frac{2}{T^{2}N^{2}}\sum_{i,j=1}^{N}(1-a_{i})(1-a_{j})\mathcal{E}
_{i}^{\top }\mathbf{K}_{h}\frac{1}{T}\widehat{\mathbf{F}}\widehat{\mathbf{F}}
^{\top }\mathcal{E}_{j}  \notag \\
&\eqqcolon &B_{1}+B_{2}-2B_{3}.
\end{eqnarray*}
We consider $B_{2}$ first. Note that 
\begin{eqnarray*}
B_{2} &=&\frac{1}{T^{4}N^{2}}\sum_{i,j=1}^{N}(1-a_{i})(1-a_{j})\mathcal{E}
_{i}^{\top }\left\{ (\widehat{\mathbf{F}}-\mathbf{F}\mathbf{H})\widehat{ 
\mathbf{F}}^{\top }\mathbf{K}_{h}\widehat{\mathbf{F}}(\widehat{\mathbf{F}}- 
\mathbf{F}\mathbf{H})^{\top }+(\mathbf{F}\mathbf{H})\widehat{\mathbf{F}}
^{\top }\mathbf{K}_{h}\widehat{\mathbf{F}}(\mathbf{F}\mathbf{H})^{\top
}\right.  \notag \\
&&\left. +2(\widehat{\mathbf{F}}-\mathbf{F}\mathbf{H})\widehat{\mathbf{F}}
^{\top }\mathbf{K}_{h}\widehat{\mathbf{F}}(\mathbf{F}\mathbf{H})^{\top
}\right\} \mathcal{E}_{j}  \notag \\
&\eqqcolon &B_{2,1}+B_{2,2}+2B_{2,3}.
\end{eqnarray*}
By Lemma \ref{Lemma3}(d), Lemma \ref{Lemma2}(f) and the condition $%
hT^{2}/N^{2}\rightarrow 0$ in Assumption \ref{Assumption1}(b), we have 
\begin{eqnarray*}
\left\vert B_{2,1}\right\vert &\leq &\frac{1}{T^{4}N^{2}}\Vert \mathbf{K}
_{h}\Vert _{2}\cdot \left\Vert \sum_{i=1}^{N}(1-a_{i})\mathcal{E}_{i}^{\top
}(\widehat{\mathbf{F}}-\mathbf{F}\mathbf{H})\right\Vert ^{2}\Vert \widehat{ 
\mathbf{F}}\Vert ^{2}=O_{P}(1/N^{2})=o_{P}\left( \frac{1}{TN\sqrt{h}}\right) 
\text{ and}  \notag \\
\left\vert B_{2,2}\right\vert &\lesssim &\frac{1}{T^{4}N^{2}}\Vert \mathbf{K}
_{h}\Vert _{2}\Vert \widehat{\mathbf{F}}\Vert ^{2}\left\Vert
\sum_{i=1}^{N}(1-a_{i})\mathcal{E}_{i}^{\top }\mathbf{F}\right\Vert
^{2}=O_{P}(1/(TN))=o_{P}\left( \frac{1}{TN\sqrt{h}}\right) .
\end{eqnarray*}
By the CS inequality, $B_{2,3}=o_{P}((TN\sqrt{h})^{-1}).$ Then $%
|B_{2}|=o_{P}((TN\sqrt{h})^{-1})$. Analogously, $|B_{3}|=o_{P}((TN\sqrt{h}
)^{-1})$. }

{\small Finally, we can write 
\begin{equation*}
L_{NT}=\frac{1}{T^{2}N^{2}}\sum_{i,j=1}^{N}%
\sum_{t,s=1}^{T}(1-a_{i})(1-a_{j}) \varepsilon _{it}\varepsilon
_{js}h^{-1}K\left( \frac{t-s}{Th}\right) +o_{P}\left( \frac{1}{TN\sqrt{h}}%
\right) .
\end{equation*}
By Lemma \ref{Lemma6}, we have 
\begin{equation*}
TN\sqrt{h}[L_{NT}-(TNh)^{-1}E(Q_{T})]\rightarrow _{D}N(0,2\nu _{0}\sigma
_{\varepsilon ,a}^{4}),
\end{equation*}
where $Q_{T}=\frac{1}{TNh}\sum_{i,j=1}^{N}\sum_{t,s=1}^{T}(1-a_{i})(1-a_{j})
\varepsilon _{it}\varepsilon _{js}K\left( \frac{t-s}{Th}\right) $, $\nu
_{0}=\int_{-1}^{1}K^{2}(u)\mathrm{d}u$. This completes the proof of the
theorem. }
\end{proof}

\begin{proof}[Proof of Theorem \ref{theorem2}]

{\small Let $\widetilde{\mathbf{F}}=(\widehat{\mathbf{F}},\ddot{\mathbf{F}})$%
, where $\ddot{\mathbf{F}}$ includes the last $\widetilde{r}-r$ columns of $%
\widetilde{\mathbf{F}}$. First, if we just focus on the first $r$ columns of 
$\widetilde{\mathbf{F}}$, we still have 
\begin{equation*}
\frac{1}{\sqrt{T}}\Vert \widehat{\mathbf{F}}-\mathbf{F}\mathbf{H}\Vert
=O_{P}(1/\sqrt{T\wedge N}).
\end{equation*}
It follows that 
\begin{equation*}
\frac{1}{T}\ddot{\mathbf{F}}^{\top }\mathbf{F}=\frac{1}{T}\ddot{\mathbf{F}}
^{\top }(\mathbf{F}-\widehat{\mathbf{F}}\mathbf{H}^{-1})=O_{P}(1/\sqrt{
T\wedge N})
\end{equation*}
and 
\begin{equation*}
\mathbf{F}^{\top }\widetilde{\mathbf{F}}/T=(\mathbf{F}^{\top }\widehat{ 
\mathbf{F}}/T,\mathbf{F}^{\top }\ddot{\mathbf{F}}/T)=(\mathbf{F}^{\top } 
\widehat{\mathbf{F}}/T,O_{P}(1/\sqrt{T\wedge N})).
\end{equation*}
Let $\mathbf{W}_{TN}=\frac{1}{TN}\pmb{\Lambda}^{\top }\pmb{\Lambda}\mathbf{F}
^{\top }\widetilde{\mathbf{F}}$. By using similar arguments of the proofs of
Lemmas \ref{Lemma8}(b)--(c), we have that $\mathbf{W}_{TN}$ is of full row
rank and $\mathbf{W}_{TN}\mathbf{W}_{TN}^{\top }\rightarrow _{P}\pmb{\Sigma}
_{\lambda }\pmb{\Sigma}_{f}\pmb{\Sigma}_{\lambda }$, i.e., $\mathbf{W}_{TN} 
\mathbf{W}_{TN}^{\top }$ is asymptotically invertible. }

{\small Note that 
\begin{equation*}
\widetilde{\mathbf{F}}\cdot \frac{1}{TN}\widetilde{\mathbf{V}}=\frac{1}{TN} 
\mathbf{X}\mathbf{X}^{\top }\widetilde{\mathbf{F}}=\frac{1}{TN}(\mathbf{F} %
\pmb{\Lambda}^{\top }+\mathcal{E})(\mathbf{F}\pmb{\Lambda}^{\top }+\mathcal{%
E })^{\top }\widetilde{\mathbf{F}},
\end{equation*}
we have 
\begin{eqnarray*}
\mathbf{F}-\widetilde{\mathbf{F}}\widetilde{\mathbf{H}}^{+} &=&-\frac{1}{TN} 
\mathbf{F}\pmb{\Lambda}^{\top }\mathcal{E}^{\top }\widetilde{\mathbf{F}} 
\mathbf{W}_{TN}^{\top }(\mathbf{W}_{TN}\mathbf{W}_{TN}^{\top })^{-1}-\frac{1 
}{TN}\mathcal{E}\pmb{\Lambda}\mathbf{F}^{\top }\widetilde{\mathbf{F}}\mathbf{%
W}_{TN}^{\top }(\mathbf{W}_{TN}\mathbf{W}_{TN}^{\top })^{-1}  \notag \\
&&-\frac{1}{TN}\mathcal{E}\mathcal{E}^{\top }\widetilde{\mathbf{F}}\mathbf{W}
_{TN}^{\top }(\mathbf{W}_{TN}\mathbf{W}_{TN}^{\top })^{-1}=-\mathbf{W}_{1}- 
\mathbf{W}_{2}-\mathbf{W}_{3},
\end{eqnarray*}
where $\widetilde{\mathbf{H}}^{+}=\frac{1}{TN}\widetilde{\mathbf{V}}\mathbf{%
W }_{TN}^{\top }(\mathbf{W}_{TN}\mathbf{W}_{TN}^{\top })^{-1}$. By Lemma \ref%
{Lemma2}, we have $\frac{1}{\sqrt{T}}\Vert \mathbf{F}-\widetilde{ \mathbf{F}}%
\widetilde{\mathbf{H}}^{+}\Vert =O_{P}(1/\sqrt{T\wedge N})$. }

{\small As in the proof of Theorem \ref{theorem1}, we next expand $L_{NT}$
as follows 
\begin{eqnarray*}
L_{NT} &=&\frac{1}{T^{2}N^{2}}\sum_{i,j=1}^{N}\pmb{\lambda}_{i}^{\top }%
\mathbf{F} ^{\top }\mathbf{M}_{\widetilde{\mathbf{F}}}\mathbf{K}_{h}\mathbf{M%
}_{ \widetilde{\mathbf{F}}}\mathbf{F}\pmb{\lambda}_{j}+\frac{1}{T^{2}N^{2}}
\sum_{i,j=1}^{N}\mathcal{E}_{i}^{\top }\mathbf{M}_{\widetilde{\mathbf{F}}} 
\mathbf{K}_{h}\mathbf{M}_{\widetilde{\mathbf{F}}}\mathcal{E}_{j}  \notag \\
&&+\frac{2}{T^{2}N^{2}}\sum_{i,j=1}^{N}\mathcal{E}_{i}^{\top }\mathbf{M}_{ 
\widetilde{\mathbf{F}}}\mathbf{K}_{h}\mathbf{M}_{\widetilde{\mathbf{F}}} 
\mathbf{F}\pmb{\lambda}_{j}  \notag \\
&\eqqcolon& L_{T,1}+L_{T,2}+L_{T,3}.
\end{eqnarray*}
}

{\small First, we study $L_{T,1}$ by making the following decomposition: 
\begin{eqnarray*}
L_{T,1} &=&\frac{1}{T^{2}N^{2}}\sum_{i,j=1}^{N}\pmb{\lambda}_{i}^{\top }( 
\mathbf{F}-\widetilde{\mathbf{F}}\widetilde{\mathbf{H}}^{+})^{\top }\mathbf{%
M }_{\widetilde{\mathbf{F}}}\mathbf{K}_{h}\mathbf{M}_{\widetilde{\mathbf{F}}
}( \mathbf{F}-\widetilde{\mathbf{F}}\widetilde{\mathbf{H}}^{+})\pmb{\lambda}
_{j}  \notag \\
&=&\frac{1}{T^{2}}\overline{\pmb{\lambda}}^{\top }\mathbf{W}_{1}^{\top } 
\mathbf{M}_{\widetilde{\mathbf{F}}}\mathbf{K}_{h}\mathbf{M}_{\widetilde{ 
\mathbf{F}}}\mathbf{W}_{1}\overline{\pmb{\lambda}}+\frac{1}{T^{2}}\overline{ %
\pmb{\lambda}}^{\top }\mathbf{W}_{2}^{\top }\mathbf{M}_{\widetilde{\mathbf{F}
}}\mathbf{K}_{h}\mathbf{M}_{\widetilde{\mathbf{F}}}\mathbf{W}_{2}\overline{ %
\pmb{\lambda}}  \notag \\
&&+\frac{1}{T^{2}}\overline{\pmb{\lambda}}^{\top }\mathbf{W}_{3}^{\top } 
\mathbf{M}_{\widetilde{\mathbf{F}}}\mathbf{K}_{h}\mathbf{M}_{\widetilde{ 
\mathbf{F}}}\mathbf{W}_{3}\overline{\pmb{\lambda}}+\frac{2}{T^{2}}\overline{ %
\pmb{\lambda}}^{\top }\mathbf{W}_{1}^{\top }\mathbf{M}_{\widetilde{\mathbf{F}
}}\mathbf{K}_{h}\mathbf{M}_{\widetilde{\mathbf{F}}}\mathbf{W}_{2}\overline{ %
\pmb{\lambda}}  \notag \\
&&+\frac{2}{T^{2}}\overline{\pmb{\lambda}}^{\top }\mathbf{W}_{1}^{\top } 
\mathbf{M}_{\widetilde{\mathbf{F}}}\mathbf{K}_{h}\mathbf{M}_{\widetilde{ 
\mathbf{F}}}\mathbf{W}_{3}\overline{\pmb{\lambda}}+\frac{2}{T^{2}}\overline{ %
\pmb{\lambda}}^{\top }\mathbf{W}_{2}^{\top }\mathbf{M}_{\widetilde{\mathbf{F}
}}\mathbf{K}_{h}\mathbf{M}_{\widetilde{\mathbf{F}}}\mathbf{W}_{3}\overline{ %
\pmb{\lambda}}  \notag \\
&\eqqcolon&L_{T,11}+L_{T,12}+L_{T,13}+2L_{T,14}+2L_{T,15}+2L_{T,16},
\end{eqnarray*}
where $\overline{\pmb{\lambda}}=\frac{1}{N}\sum_{i=1}^{N}\pmb{\lambda}_{i}$.
Below, we consider the terms on the r.h.s. of the last displayed equation
one by one. For $L_{T,11}$, by using the identity $\widetilde{\mathbf{F}} 
\mathbf{W}_{TN}^{\top }=\mathbf{F}(\pmb{\Lambda}^{\top }\pmb{\Lambda}/N)- 
\mathbf{M}_{\widetilde{\mathbf{F}}}\mathbf{F}(\pmb{\Lambda}^{\top } %
\pmb{\Lambda}/N)$ and using Lemmas \ref{Lemma2}(c)--(d), 
\begin{eqnarray*}
&& |L_{T,11}| \leq \frac{1}{T^{2}N^{2}}\Vert \mathbf{F}^{\top }\mathbf{M}_{ 
\widetilde{\mathbf{F}}}\mathbf{K}_{h}\mathbf{M}_{\widetilde{\mathbf{F}}} 
\mathbf{F}\Vert _{2}\left\Vert \frac{1}{T}\pmb{\Lambda}^{\top }\mathcal{E}
^{\top }\widetilde{\mathbf{F}}\mathbf{W}_{TN}^{\top }(\mathbf{W}_{TN}\mathbf{%
\ W}_{TN}^{\top })^{-1}\overline{\pmb{\lambda}}\right\Vert ^{2}  \notag \\
&\leq &\frac{2}{T^{2}N^{2}}\Vert \mathbf{K}_{h}\Vert _{2}\Vert \mathbf{M}_{ 
\widetilde{\mathbf{F}}}\Vert _{2}^{2}\Vert \mathbf{F}-\widetilde{\mathbf{F}} 
\widetilde{\mathbf{H}}^{+}\Vert ^{2}\left\Vert \frac{1}{T}\pmb{\Lambda}
^{\top }\mathcal{E}^{\top }\mathbf{F}(\pmb{\Lambda}^{\top }\pmb{\Lambda}/N)( 
\mathbf{W}_{TN}\mathbf{W}_{TN}^{\top })^{-1}\overline{\pmb{\lambda}}
\right\Vert ^{2}  \notag \\
&&+\frac{2}{T^{2}N^{2}}\Vert \mathbf{K}_{h}\Vert _{2}\Vert \mathbf{M}_{ 
\widetilde{\mathbf{F}}}\Vert _{2}^{2}\Vert \mathbf{F}-\widetilde{\mathbf{F}} 
\widetilde{\mathbf{H}}^{+}\Vert ^{2}\left\Vert \frac{1}{T}\pmb{\Lambda}
^{\top }\mathcal{E}^{\top }\mathbf{M}_{\widetilde{\mathbf{F}}}(\mathbf{F}- 
\widetilde{\mathbf{F}}\widetilde{\mathbf{H}}^{+})(\pmb{\Lambda}^{\top } %
\pmb{\Lambda}/N)(\mathbf{W}_{TN}\mathbf{W}_{TN}^{\top })^{-1}\overline{ %
\pmb{\lambda}}\right\Vert ^{2}  \notag \\
&\leq &O_{P}(1)\frac{1}{T^{2}N^{2}}\cdot T\cdot \frac{T}{T\wedge N}\cdot 
\frac{N}{T}+O_{P}(1)\frac{1}{T^{2}N^{2}}\cdot T\cdot \frac{T}{T\wedge N}
\cdot \frac{N}{T}\frac{T}{T\wedge N}  \notag \\
&=&O_{P}\left( \frac{1}{TN(T\wedge N)}+\frac{1}{N(T\wedge N)^{2}}\right)
=o_{P}\left( \frac{1}{TN\sqrt{h}}\right) .
\end{eqnarray*}
For $L_{T,12}$, since $(\mathbf{W}_{TN}\mathbf{W}_{TN}^{\top })^{-1}\mathbf{%
W }_{TN}\frac{1}{T}\widetilde{\mathbf{F}}^{\top }\mathbf{F}=\pmb{\Sigma}_{ %
\pmb{\Lambda}}^{-1}$ and $\pmb{\Sigma}_{\pmb{\Lambda}}=\frac{1}{N} %
\pmb{\Lambda}^{\top }\pmb{\Lambda}$, we have 
\begin{eqnarray*}
L_{T,12} &=&\frac{1}{T^{2}N^{2}}\overline{\pmb{\lambda}}^{\top }\pmb{\Sigma}
_{\pmb{\Lambda}}^{-1}\pmb{\Lambda}^{\top }\mathcal{E}^{\top }\mathbf{M}_{ 
\widehat{\mathbf{F}}}\mathbf{K}_{h}\mathbf{M}_{\widehat{\mathbf{F}}}\mathcal{%
E}\pmb{\Lambda}\pmb{\Sigma}_{\pmb{\Lambda}}^{-1}\overline{\pmb{\lambda}} 
\notag \\
&=&\frac{1}{T^{2}N^{2}}\sum_{i,j=1}^{N}\overline{\pmb{\lambda}}^{\top } %
\pmb{\Sigma}_{\pmb{\Lambda}}^{-1}\pmb{\lambda}_{i}\mathcal{E}_{i}^{\top } 
\mathbf{M}_{\widehat{\mathbf{F}}}\mathbf{K}_{h}\mathbf{M}_{\widehat{\mathbf{%
F }}}\mathcal{E}_{j}\pmb{\lambda}_{j}^{\top }\pmb{\Sigma}_{\pmb{\Lambda}
}^{-1} \overline{\pmb{\lambda}},
\end{eqnarray*}
which can be merged with $L_{T,2}$ together later on. For $L_{T,13}$, by
using the identity $\widetilde{\mathbf{F}}\mathbf{W}_{TN}^{\top }=\mathbf{F}
( \pmb{\Lambda}^{\top }\pmb{\Lambda}/N)-\mathbf{M}_{\widetilde{\mathbf{F}}} 
\mathbf{F}(\pmb{\Lambda}^{\top }\pmb{\Lambda}/N)$ and using Lemmas \ref%
{Lemma2}(e) and (b), we have 
\begin{eqnarray*}
|L_{T,13}| &\leq &O_{P}(1)\frac{2}{T^{2}N^{2}}\Vert \mathbf{M}_{\widetilde{ 
\mathbf{F}}}\mathbf{K}_{h}\mathbf{M}_{\widetilde{\mathbf{F}}}\Vert
_{2}\left\Vert \frac{1}{T}\mathcal{E}\mathcal{E}^{\top }\mathbf{F}
\right\Vert ^{2}  \notag \\
&&+O_{P}(1)\frac{2}{T^{2}N^{2}}\Vert \mathbf{M}_{\widetilde{\mathbf{F}}} 
\mathbf{K}_{h}\mathbf{M}_{\widetilde{\mathbf{F}}}\Vert _{2}\left\Vert \frac{
1 }{T}\mathcal{E}\mathcal{E}^{\top }\mathbf{M}_{\widetilde{\mathbf{F}}}( 
\mathbf{F}-\widetilde{\mathbf{F}}\widetilde{\mathbf{H}}^{+})\right\Vert ^{2}
\notag \\
&\leq &O_{P}(1)\frac{2}{T^{2}N^{2}}\cdot T\cdot \frac{N^{2}}{T\wedge N}
+O_{P}(1)\frac{2}{T^{2}N^{2}}\cdot T\cdot \frac{TN^{2}}{(T\wedge N)^{2}} 
\notag \\
&=&O_{P}\left( \frac{1}{T(T\wedge N)}+\frac{1}{(T\wedge N)^{2}}\right)
=o_{P}\left( 1/(TN\sqrt{h})\right) .
\end{eqnarray*}
By the development of $L_{T,11},$ $L_{T,12}$ and $L_{T,13}$ and the CS
inequality, we know $L_{T,14}$, $L_{T,15}$ and $L_{T,16}$ are all
asymptotically negligible. }

{\small Similarly, we can show that 
\begin{equation*}
L_{T,3}=-\frac{2}{T^{2}N^{2}}\sum_{i,j=1}^{N}\mathcal{E}_{i}^{\top }\mathbf{%
M }_{\widetilde{\mathbf{F}}}\mathbf{K}_{h}\mathbf{M}_{\widetilde{\mathbf{F}}
} \mathcal{E}_{j}\pmb{\lambda}_{j}^{\top }\pmb{\Sigma}_{\pmb{\Lambda}}^{-1} 
\overline{\pmb{\lambda}}+o_{P}\left( 1/(TN\sqrt{h})\right) .
\end{equation*}
It follows that 
\begin{equation*}
L_{NT}=\frac{1}{T^{2}N^{2}}\sum_{i,j=1}^{N}(1-a_{i})(1-a_{j})\mathcal{E}
_{i}^{\top }\mathbf{M}_{\widetilde{\mathbf{F}}}\mathbf{K}_{h}\mathbf{M}_{ 
\widetilde{\mathbf{F}}}\mathcal{E}_{j}+o_{P}\left( 1/(TN\sqrt{h})\right) ,
\end{equation*}
where $a_{i}=\pmb{\lambda}_{i}^{\top }\pmb{\Sigma}_{\pmb{\Lambda}}^{-1} 
\overline{\pmb{\lambda}}.$ }

{\small Note further that 
\begin{eqnarray*}
&&\frac{1}{T^{2}N^{2}}\sum_{i,j=1}^{N}(1-a_{i})(1-a_{j})\mathcal{E}
_{i}^{\top }\mathbf{M}_{\widetilde{\mathbf{F}}}\mathbf{K}_{h}\mathbf{M}_{ 
\widetilde{\mathbf{F}}}\mathcal{E}_{j}  \notag \\
&=&\frac{1}{T^{2}N^{2}}\sum_{i,j=1}^{N}(1-a_{i})(1-a_{j})\mathcal{E}
_{i}^{\top }\mathbf{K}_{h}\mathcal{E}_{j}+\frac{1}{T^{2}N^{2}}
\sum_{i,j=1}^{N}(1-a_{i})(1-a_{j})\mathcal{E}_{i}^{\top }\frac{1}{T} 
\widetilde{\mathbf{F}}\widetilde{\mathbf{F}}^{\top }\mathbf{K}_{h}\frac{1}{T}
\widetilde{\mathbf{F}}\widetilde{\mathbf{F}}^{\top }\mathcal{E}_{j}  \notag
\\
&&-\frac{2}{T^{2}N^{2}}\sum_{i,j=1}^{N}(1-a_{i})(1-a_{j})\mathcal{E}
_{i}^{\top }\mathbf{K}_{h}\frac{1}{T}\widetilde{\mathbf{F}}\widetilde{ 
\mathbf{F}}^{\top }\mathcal{E}_{j}  \notag \\
&\eqqcolon&B_{1}+B_{2}-2B_{3}.
\end{eqnarray*}
We consider $B_{2}$. First, let $\pmb{\Sigma}_{\mathbf{F}\pmb{\Lambda}}= 
\frac{1}{TN}\mathbf{F}\pmb{\Lambda}^{\top }\pmb{\Lambda}\mathbf{F}^{\top }.$
By Assumptions \ref{Assumption2} and \ref{Assumption3}, there exists a
positive fixed constant (say, $c_{1}^{\ast }$) satisfying that $\lambda
_{\max }(\pmb{\Sigma}_{\mathbf{F}\pmb{\Lambda}})\leq c_{1}^{\ast }.$ Second,
we let $\mathbf{F}^{\ast }$ be a $T\times (T-r)$ matrix such that 
\begin{equation*}
\frac{1}{T}(\mathbf{F}^{\ast },\mathbf{F}\mathbf{D})^{\top }(\mathbf{F}
^{\ast },\mathbf{F}\mathbf{D})=\left( 
\begin{array}{cc}
\mathbf{I}_{T-r} & \mathbf{0} \notag \\ 
\mathbf{0} & \mathbf{I}_{r}%
\end{array}
\right) ,
\end{equation*}
where $\mathbf{D}$ is an $r\times r$ rotation matrix such that $\frac{1}{T} 
\mathbf{D}^{\top }\mathbf{F}^{\top }\mathbf{F}\mathbf{D}=\mathbf{I}_{r}$.
Now, write 
\begin{equation*}
\frac{1}{TN}\mathbf{X}\mathbf{X}^{\top }=\pmb{\Sigma}_{\mathbf{F} %
\pmb{\Lambda}}+\left( \frac{1}{TN}\mathbf{X}\mathbf{X}^{\top }-\pmb{\Sigma}
_{ \mathbf{F}\pmb{\Lambda}}\right) \eqqcolon\pmb{\Sigma}_{\mathbf{F} %
\pmb{\Lambda}}+\Delta \pmb{\Sigma}_{\mathbf{F}\pmb{\Lambda}},
\end{equation*}
where the definition of $\Delta \pmb{\Sigma}_{\mathbf{F}\pmb{\Lambda}}$ is
obvious.\medskip }

{\small Having introduced the above notations, we are now ready to proceed
further. Note that $\frac{1}{\sqrt{T}}\mathbf{F}^{\ast }$, $\frac{1}{\sqrt{T}%
}\mathbf{F}\mathbf{D}$, $\pmb{\Sigma}_{\mathbf{F}\pmb{\Lambda}}$ and $\Delta %
\pmb{\Sigma}_{\mathbf{F}\pmb{\Lambda}}$ are corresponding to $\mathbf{Q}_{1}$
, $\mathbf{Q}_{2}$, $\mathbf{A}$ and $\mathbf{E}$ of Lemma \ref{Lemma7}.
Thus, using Lemma \ref{Lemma7}, we obtain that }

{\small 
\begin{equation*}
\widetilde{\mathbf{F}}^{\ast }\coloneqq\frac{1}{\sqrt{T}}\left( \mathbf{F}
^{\ast }+\mathbf{F}\mathbf{D}\mathbf{P}\right) (\mathbf{I}_{T-r}+\mathbf{P}
^{\top }\mathbf{P})^{-1/2},
\end{equation*}
which is corresponding to $\mathbf{Q}_{1}^{0}$ in Lemma \ref{Lemma7}.
Moreover, 
\begin{equation*}
\Vert \mathbf{P}\Vert \leq \frac{4}{\text{\normalfont sep}(0,\frac{1}{T} 
\mathbf{F}^{\top }\pmb{\Sigma}_{\mathbf{F}\pmb{\Lambda}}\mathbf{F})}\cdot
\Vert \Delta \pmb{\Sigma}_{\mathbf{F}\pmb{\Lambda}}\Vert \leq O_{P}(1)\Vert
\Delta \pmb{\Sigma}_{\mathbf{F}\pmb{\Lambda}}\Vert =O_{P}\left( 1/\sqrt{
T\wedge N}\right) ,
\end{equation*}
where the last line follows from the development of Lemma \ref{Lemma2}. }

{\small Since $\widetilde{\mathbf{F}}^{\ast }$ is an orthonormal basis for a
subspace that is invariant for $\pmb{\Sigma}_{\mathbf{F}\pmb{\Lambda}
}+\Delta \pmb{\Sigma}_{\mathbf{F}\pmb{\Lambda}}=\frac{1}{TN}\mathbf{X} 
\mathbf{X}^{\top }$, studying $\frac{1}{\sqrt{T}}\ddot{\mathbf{F}}$ is
equivalent to investigating $\widetilde{\mathbf{F}}^{\ast }$. Then we write 
\begin{eqnarray*}
\left\Vert \widetilde{\mathbf{F}}^{\ast }-\frac{1}{\sqrt{T}}\mathbf{F}^{\ast
}\right\Vert &=&\frac{1}{\sqrt{T}}\Vert [ \mathbf{F}^{\ast }+\mathbf{F} 
\mathbf{D}\mathbf{P}-\mathbf{F}^{\ast }(\mathbf{I}_{T-r}+\mathbf{P}^{\top } 
\mathbf{P})^{1/2}](\mathbf{I}_{T-r}+\mathbf{P}^{\top }\mathbf{P})^{-1/2}\Vert
\notag \\
&\leq &\frac{1}{\sqrt{T}}\Vert \mathbf{F}^{\ast }(\mathbf{I}_{T-r}-(\mathbf{%
I }_{T-r}+\mathbf{P}^{\top }\mathbf{P})^{1/2})(\mathbf{I}_{T-r}+\mathbf{P}
^{\top }\mathbf{P})^{-1/2}\Vert  \notag \\
&&+\frac{1}{\sqrt{T}}\Vert \mathbf{F}\mathbf{D}\mathbf{P}(\mathbf{I}_{T-r}+ 
\mathbf{P}^{\top }\mathbf{P})^{-1/2}\Vert  \notag \\
&\leq &\Vert (\mathbf{I}_{T-r}-(\mathbf{I}_{T-r}+\mathbf{P}^{\top }\mathbf{P}
)^{1/2})(\mathbf{I}_{T-r}+\mathbf{P}^{\top }\mathbf{P})^{-1/2}\Vert +\Vert 
\mathbf{P}(\mathbf{I}_{T-r}+\mathbf{P}^{\top }\mathbf{P})^{-1/2}\Vert  \notag
\\
&\leq &\frac{\Vert \mathbf{I}_{T-r}-(\mathbf{I}_{T-r}+\mathbf{P}^{\top } 
\mathbf{P})^{1/2}\Vert +\Vert \mathbf{P}\Vert }{\lambda _{\min }^{1/2}( 
\mathbf{I}_{T-r}+\mathbf{P}^{\top }\mathbf{P})}=O_{P}\left( 1/\sqrt{T\wedge
N }\right) .
\end{eqnarray*}
Then, we are able to write 
\begin{eqnarray*}
B_{2} &=&\frac{1}{T^{4}N^{2}}\sum_{i,j=1}^{N}(1-a_{i})(1-a_{j})\mathcal{E}
_{i}^{\top }\left\{ (\widehat{\mathbf{F}}-\mathbf{F}\mathbf{H},\ddot{\mathbf{%
F}}-\ddot{\mathbf{F}}^{\ast })\widetilde{\mathbf{F}}^{\top }\mathbf{K}_{h} 
\widetilde{\mathbf{F}}(\widehat{\mathbf{F}}-\mathbf{F}\mathbf{H},\ddot{ 
\mathbf{F}}-\ddot{\mathbf{F}}^{\ast })^{\top }\mathcal{E}_{j}\right.  \notag
\\
&&\left. +(\mathbf{F}\mathbf{H},\ddot{\mathbf{F}}^{\ast })\widetilde{\mathbf{%
F}}^{\top }\mathbf{K}_{h}\widetilde{\mathbf{F}}(\mathbf{F}\mathbf{H},\ddot{ 
\mathbf{F}}^{\ast })^{\top }+2(\widehat{\mathbf{F}}-\mathbf{F}\mathbf{H}, 
\ddot{\mathbf{F}}-\ddot{\mathbf{F}}^{\ast })\widetilde{\mathbf{F}}^{\top } 
\mathbf{K}_{h}\widetilde{\mathbf{F}}(\mathbf{F}\mathbf{H},\ddot{\mathbf{F}}
^{\ast })^{\top }\right\} \mathcal{E}_{j}  \notag \\
&\eqqcolon & B_{2,1}+B_{2,2}+2B_{2,3},
\end{eqnarray*}
where $\ddot{\mathbf{F}}^{\ast }$ are the $\widetilde{r}-r$ columns of $%
\mathbf{F}^{\ast }$ corresponding to $\ddot{\mathbf{F}}$. In view of the
aforementioned developments, the rest proof is the same as that in Theorem %
\ref{theorem1}. For example, 
\begin{eqnarray*}
B_{2,1} &\leq &\left\Vert \frac{1}{TN}\sum_{i=1}^{N}(1-a_{i})\mathcal{E}
_{i}^{\top }\right\Vert ^{2}\left\Vert (\widehat{\mathbf{F}}-\mathbf{F} 
\mathbf{H},\ddot{\mathbf{F}}-\ddot{\mathbf{F}}^{\ast })\right\Vert ^{2}\frac{
1}{T^{2}}\left\Vert \widetilde{\mathbf{F}}^{\top }\mathbf{K}_{h}\widetilde{ 
\mathbf{F}}\right\Vert  \notag \\
&=&O_{P}\left( \frac{1}{TN}\right) \times O_{P}\left( \frac{T}{T\wedge N}
\right) \times O_{P}(1)=o_{P}\left( 1/(TN\sqrt{h})\right) .
\end{eqnarray*}
Then the proof is now completed. }
\end{proof}

\begin{proof}[Proof of Proposition \ref{Coro1}]

{\small By the proof of Theorems \ref{theorem1} and \ref{theorem2} with $Th$
replaced by $l$, for any fixed $\widetilde{r}\geq r$ we have 
\begin{eqnarray*}
\widehat{\sigma }_{\varepsilon ,a}^{2} &=&\sum_{k=-l}^{l}\widehat{\sigma }
_{\varepsilon ,a,k}^{2}a(k/l)=\frac{1}{TN}\sum_{t,s=1}^{T}\sum_{i,j=1}^{N} 
\widehat{\varepsilon }_{it}\widehat{\varepsilon }_{js}a((t-s)/l)  \notag \\
&=&\frac{1}{TN}\sum_{i,j=1}^{N}\sum_{t,s=1}^{T}(1-a_{i})(1-a_{j})\varepsilon
_{it}\varepsilon_{js}a\left( (t-s)/l\right) +o_{P}\left( \sqrt{l/T} \right) .
\end{eqnarray*}
By Lemma \ref{Lemma6} with $Th$ replaced by $l$, we have 
\begin{equation*}
\frac{1}{TN}\sum_{i,j=1}^{N}\sum_{t,s=1}^{T}(1-a_{i})(1-a_{j})\varepsilon
_{it}\varepsilon_{js}a\left( \frac{t-s}{l}\right) =\frac{1}{T}
\sum_{t,s=1}^{T}E(\overline{\varepsilon }_{t,a}\overline{\varepsilon }
_{s,a})a\left( (t-s)/l\right) +O_{P}(l/T),
\end{equation*}
where $\overline{\varepsilon }_{t,a}=\frac{1}{\sqrt{N}}
\sum_{i=1}^{N}(1-a_{i})\varepsilon _{it}$. }

{\small Next, we complete the proof by calculating the bias term induced by
truncation. Write 
\begin{eqnarray*}
&&l^{q}\left( \frac{1}{T}\sum_{t,s=1}^{T}E(\overline{\varepsilon }_{t,a} 
\overline{\varepsilon }_{s,a})a\left( (t-s)/l\right) -\sigma _{\varepsilon
,a}^{2}\right)  \notag \\
&=&l^{q}\sum_{k=-l}^{l}\left[ a\left( k/l\right) -1\right] E(\overline{
\varepsilon }_{k,a}\overline{\varepsilon }_{0,a})-2l^{q}\sum_{k=l+1}^{\infty
}E(\overline{\varepsilon }_{k,a}\overline{\varepsilon }_{0,a})\eqqcolon %
I_{1}-2I_{2}.
\end{eqnarray*}%
We first consider $I_{1}$. By Assumption \ref{Assumption5}, $\forall
\epsilon >0$, we choose $\nu _{\epsilon }>0$ such that 
\begin{equation*}
|k/l|<\nu _{\epsilon }\quad \text{and}\quad \left\vert \frac{1-a(k/l)}{
|k/l|^{q}}-\bar{c}_{q}\right\vert <\epsilon .
\end{equation*}%
Letting $l_{T}^{\ast }=\lfloor \nu _{\epsilon }l\rfloor $, we write 
\begin{equation*}
I_{1}=\sum_{k=-l_{T}^{\ast }}^{l_{T}^{\ast }}\frac{a\left( k/l\right) -1}{
\left\vert k/l\right\vert ^{q}}|k|^{q}E(\overline{\varepsilon }_{k,a} 
\overline{\varepsilon }_{0,a})+2\sum_{k=l_{T}^{\ast }+1}^{l}\frac{a\left(
k/\ell \right) -1}{\left\vert k/\ell \right\vert ^{q}}|k|^{q}E(\overline{
\varepsilon }_{k,a}\overline{\varepsilon }_{0,a}).
\end{equation*}%
Then, it is easy to see that the first term of the r.h.s. converges to $-%
\bar{c}_{q}\sum_{k=-\infty }^{\infty }|k|^{q}E(\overline{\varepsilon }_{k,a}%
\overline{\varepsilon }_{0,a})<\infty $ under Assumption \ref{Assumption5}
provided that $\sum_{k=-\infty }^{\infty }|k|^{2}E(\overline{\varepsilon }%
_{k,a}\overline{\varepsilon }_{0,a})<\infty $. We next verify this condition
based on Assumption \ref{Assumption4}. Note that $\overline{\varepsilon }%
_{t,a}=\sum_{j=0}^{\infty }\mathscr{P}_{t-j}(\overline{\varepsilon }_{t,a})$
and write 
\begin{eqnarray*}
&&\sum_{k=1}^{\infty }k^{2}|E[\overline{\varepsilon }_{0,a}\overline{
\varepsilon }_{k,a}]|=\sum_{k=1}^{\infty }k^{q}\left\vert E\left[ \left(
\sum_{j=0}^{\infty }\mathscr{P}_{-j}(\overline{\varepsilon }_{0,a})\right)
\left( \sum_{j=0}^{\infty }\mathscr{P}_{k-j}(\overline{\varepsilon }
_{k,a})\right) \right] \right\vert  \notag \\
&=&\sum_{k=1}^{\infty }k^{2}\left\vert E\left[ \sum_{j=0}^{\infty }(E[ 
\overline{\varepsilon }_{0,a}\,|\, \mathscr{E}_{-j}]-E[\overline{\varepsilon 
} _{0,a}\,|\, \mathscr{E}_{-j-1}])\cdot (E[\overline{\varepsilon }%
_{k,a}\,|\, \mathscr{E}_{-j}]-E[\overline{\varepsilon }_{k,a}\,|\, %
\mathscr{E}_{-j-1}]) \right] \right\vert  \notag \\
&\leq &\sum_{k=1}^{\infty }k^{2}\sum_{j=0}^{\infty }\left\vert E[\overline{
\varepsilon }_{0,a}\,|\, \mathscr{E}_{-j}]-E[\overline{\varepsilon }
_{0,a}\,|\, \mathscr{E}_{-j-1}]\right\vert _{2}\cdot \left\vert E[\overline{
\varepsilon }_{k,a}\,|\, \mathscr{E}_{-j}]-E[\overline{\varepsilon }
_{k,a}\,|\, \mathscr{E}_{-j-1}]\right\vert _{2}  \notag \\
&=&\sum_{k=1}^{\infty }k^{2}\sum_{j=0}^{\infty }\left\vert E[\overline{
\varepsilon }_{j,a}\,|\, \mathscr{E}_{0}]-E[\overline{\varepsilon }
_{j,a}^{\ast }\,|\, \mathscr{E}_{0}]\right\vert _{2}\cdot \left\vert E[ 
\overline{\varepsilon }_{k+j,a}\,|\, \mathscr{E}_{0}]-E[\overline{%
\varepsilon }_{k+j,a}^{\ast }\,|\, \mathscr{E}_{0}]\right\vert _{2}  \notag
\\
&\leq &\sum_{k=1}^{\infty }k^{2}\sum_{j=0}^{\infty }\left\vert \overline{
\varepsilon }_{j,a}-\overline{\varepsilon }_{j,a}^{\ast }\right\vert
_{2}\left\vert \overline{\varepsilon }_{k+j,a}-\overline{\varepsilon }
_{k+j,a}^{\ast }\right\vert _{2}<\infty ,  \notag
\end{eqnarray*}%
where the second equality follows from the fact that $\mathscr{P}_{t-j}(%
\overline{\varepsilon }_{t,a})$' are martingale differences, the first
inequality follows from the CS inequality, the third equality follows from
stationarity, and the second inequality follows from Jensen's inequality.
Next, for $I_{2},$ we have 
\begin{equation*}
|I_{2}|\leq \frac{l^{q}}{l^{2}}\sum_{k=l+1}^{\infty }k^{2}|E(\overline{
\varepsilon }_{0,a}\overline{\varepsilon }_{k,a})|\rightarrow 0,
\end{equation*}%
where the last steps follows from the fact that $\frac{l^{q}}{l^{2}}$ is
bounded and $l\rightarrow \infty $. This completes the proof of the
proposition. }
\end{proof}

\begin{proof}[Proof of Theorem \ref{theorem3}]

{\small We first consider the case of $\widetilde{r}=r$. Note that under the
local alternative \eqref{Eq2.8}, we have $\mathbf{X}=\mathbf{F}\pmb{\Lambda}%
^{\top }+\mathcal{E}^{\ast }$ with $\mathcal{E}^{\ast }=\{\varepsilon
_{it}+a_{TN} \mathbf{g}_{it}^{\top }\mathbf{f}_{t}\}_{T\times N}$. Hence, as
in the proof of Theorem \ref{theorem1}, we can decompose $\mathbf{F}-%
\widehat{\mathbf{F}} \mathbf{H}^{-1}$ as follows: 
\begin{equation*}
\mathbf{F}-\widehat{\mathbf{F}}\mathbf{H}^{-1}=-\frac{1}{TN}\mathbf{F} %
\pmb{\Lambda}^{\top }\mathcal{E}^{\ast ,\top }\widehat{\mathbf{F}} %
\pmb{\Sigma}_{\widehat{\mathbf{F}}}^{-1}\pmb{\Sigma}_{\pmb{\Lambda}}^{-1}- 
\frac{1}{N}\mathcal{E}^{\ast }\pmb{\Lambda}\pmb{\Sigma}_{\pmb{\Lambda}
}^{-1}- \frac{1}{TN}\mathcal{E}^{\ast }\mathcal{E}^{\ast ,\top }\widehat{ 
\mathbf{F}} \pmb{\Sigma}_{\widehat{\mathbf{F}}}^{-1}\pmb{\Sigma}_{ %
\pmb{\Lambda}}^{-1}.
\end{equation*}
In addition, by definition, we expand $L_{NT}$ as follows: 
\begin{eqnarray*}
L_{NT}&=&\frac{1}{T^{2}N^{2}}\sum_{i,j=1}^{N}\pmb{\lambda}_{i}^{\top }%
\mathbf{F} ^{\top }\mathbf{M}_{\widehat{\mathbf{F}}}\mathbf{K}_{h}\mathbf{M}%
_{\widehat{ \mathbf{F}}}\mathbf{F}\pmb{\lambda}_{j}+\frac{1}{T^{2}N^{2}}%
\sum_{i,j=1}^{N} \mathcal{E}_{i}^{\ast ,\top }\mathbf{M}_{\widehat{\mathbf{F}%
}}\mathbf{K}_{h} \mathbf{M}_{\widehat{\mathbf{F}}}\mathcal{E}_{j}^{\ast } 
\notag \\
&&+\frac{2}{T^{2}N^{2}}\sum_{i,j=1}^{N}\mathcal{E}_{i}^{\ast ,\top }\mathbf{%
M }_{\widehat{\mathbf{F}}}\mathbf{K}_{h}\mathbf{M}_{\widehat{\mathbf{F}}} 
\mathbf{F}\pmb{\lambda}_{j}.
\end{eqnarray*}
Then by using Lemma \ref{Lemma9} and identical arguments as used in the
proof of Theorem \ref{theorem1}, we have 
\begin{eqnarray*}
L_{NT}&=&\frac{1}{T^{2}N^{2}}\sum_{i,j=1}^{N}%
\sum_{t,s=1}^{T}(1-a_{i})(1-a_{j}) \left\{ \varepsilon _{it}\varepsilon
_{js}+a_{TN}^{2}\mathbf{g}_{it}^{\top } \mathbf{f}_{t}\mathbf{g}_{js}^{\top }%
\mathbf{f}_{s}+2a_{TN}\mathbf{g} _{it}^{\top }\mathbf{f}_{t}\varepsilon
_{js}\right\} h^{-1}K_{ts}  \notag \\
&&+o_{P}\left( \frac{1}{TN\sqrt{h}}\right)  \notag \\
&\eqqcolon &J_{1}+J_{2}+2J_{3}+o_{P}\left( \frac{1}{TN\sqrt{h}}\right) ,
\end{eqnarray*}
where $\varepsilon _{it}^{\ast }=\varepsilon _{it}+a_{TN}\mathbf{g}
_{it}^{\top }\mathbf{f}_{t}$. }

{\small For $J_{1}$, by using Lemma \ref{Lemma6}, we have 
\begin{equation*}
TN\sqrt{h}\left( J_{1}-(TNh)^{-1}E(Q_{T})\right) \rightarrow _{D}N(0,2\nu
_{0}\sigma _{\varepsilon ,a}^{4}),
\end{equation*}
where $Q_{T}=\frac{1}{TN}\sum_{i,j=1}^{N}\sum_{t,s=1}^{T}(1-a_{i})(1-a_{j})
\varepsilon _{it}\varepsilon _{js}K_{ts}$ and $\nu _{0}=\int_{-1}^{1}K^{2}(u)%
\mathrm{d}u$. For $J_{2}$, by using Lemmas \ref{Lemma6} and \ref{Lemma1}(b),
we have 
\begin{eqnarray*}
&&TN\sqrt{h}J_{2} =\frac{1}{T^{2}N^{2}}\sum_{i,j=1}^{N}
\sum_{t,s=1}^{T}(1-a_{i})(1-a_{j})\mathbf{g}_{it}^{\top }\mathbf{f}_{t} 
\mathbf{g}_{js}^{\top }\mathbf{f}_{s}h^{-1}K_{ts}  \notag \\
&=&\frac{1}{T^{2}N^{2}}\sum_{i,j=1}^{N}\sum_{t,s=1}^{T}(1-a_{i})(1-a_{j}) 
\notag \\
&&\times \left\{ \mathbf{g}_{it}^{\top }E(\mathbf{f}_{t})\mathbf{g}
_{js}^{\top }E(\mathbf{f}_{s})+\mathbf{g}_{it}^{\top }(\mathbf{f}_{t}-E( 
\mathbf{f}_{t}))\mathbf{g}_{js}^{\top }(\mathbf{f}_{s}-E(\mathbf{f}_{s}))+2 
\mathbf{g}_{it}^{\top }(\mathbf{f}_{t}-E(\mathbf{f}_{t}))\mathbf{g}
_{js}^{\top }E(\mathbf{f}_{s})\right\} h^{-1}K_{ts}  \notag \\
&=&\frac{1}{T^{2}N^{2}}\sum_{i,j=1}^{N}\sum_{t,s=1}^{T}(1-a_{i})(1-a_{j}) 
\mathbf{g}_{it}^{\top }E(\mathbf{f}_{t})\mathbf{g}_{js}^{\top }E(\mathbf{f}
_{s})h^{-1}K_{ts}+O_{P}\left( \frac{1}{Th}\right) +O_{P}\left( \frac{1}{ 
\sqrt{T}}\right)  \notag \\
&\rightarrow &_{P}\int_{-1}^{1}K(\tau )\mathrm{d}\tau \int_{0}^{1}\left[ 
\overline{\mathbf{g}}_{a}(\tau )^{\top }E(\mathbf{f}_{t})\right] ^{2}\mathrm{%
d}\tau ,
\end{eqnarray*}
where $\overline{\mathbf{g}}_{a}(\tau )= \lim_{N\to \infty}\frac{1}{N}%
\sum_{i=1}^{N}(1-a_{i}) \mathbf{g}_{i}(\tau )$. Similarly, by using Lemma %
\ref{Lemma6}, we have 
\begin{equation*}
J_{3}=a_{TN}O_{P}\left(\frac{1}{Th\sqrt{N}}\right)=o_{P}(1/(TN\sqrt{h}))
\end{equation*}
if $Th^{3/2}\rightarrow \infty $. }

{\small For the case of the case of $\widetilde{r}>r$, we can decompose $%
\mathbf{F}- \widetilde{\mathbf{F}}\widetilde{\mathbf{H}}^{+}$ as 
\begin{eqnarray*}
\mathbf{F}-\widetilde{\mathbf{F}}\widetilde{\mathbf{H}}^{+} &=&-\frac{1}{TN} 
\mathbf{F}\pmb{\Lambda}^{\top }\mathcal{E}^{\ast ,\top }\widetilde{\mathbf{F}
}\mathbf{W}_{TN}^{\top }(\mathbf{W}_{TN}\mathbf{W}_{TN}^{\top })^{-1}-\frac{
1 }{TN}\mathcal{E}^{\ast }\pmb{\Lambda}\mathbf{F}^{\top }\widetilde{\mathbf{%
F }} \mathbf{W}_{TN}^{\top }(\mathbf{W}_{TN}\mathbf{W}_{TN}^{\top })^{-1} 
\notag \\
&&-\frac{1}{TN}\mathcal{E}^{\ast }\mathcal{E}^{\ast ,\top }\widetilde{ 
\mathbf{F}}\mathbf{W}_{TN}^{\top }(\mathbf{W}_{TN}\mathbf{W}_{TN}^{\top
})^{-1}=-\mathbf{W}_{1}-\mathbf{W}_{2}-\mathbf{W}_{3},
\end{eqnarray*}
where $\widetilde{\mathbf{H}}^{+}=\frac{1}{TN}\widetilde{\mathbf{V}}\mathbf{%
W }_{TN}^{\top }(\mathbf{W}_{TN}\mathbf{W}_{TN}^{\top })^{-1}$ and $\mathbf{%
W } _{TN}=\frac{1}{TN}\pmb{\Lambda}^{\top }\pmb{\Lambda}\mathbf{F}^{\top } 
\widetilde{\mathbf{F}}$. Then by using Lemma \ref{Lemma9} and identical
arguments as used in the proof of Theorem \ref{theorem2}, we have 
\begin{eqnarray*}
L_{NT}&=&\frac{1}{T^{2}N^{2}}\sum_{i,j=1}^{N}%
\sum_{t,s=1}^{T}(1-a_{i})(1-a_{j}) \left\{ \varepsilon _{it}\varepsilon
_{js}+a_{TN}^{2}\mathbf{g}_{it}^{\top } \mathbf{f}_{t}\mathbf{g}_{js}^{\top }%
\mathbf{f}_{s}+2a_{TN}\mathbf{g} _{it}^{\top }\mathbf{f}_{t}\varepsilon
_{js}\right\} h^{-1}K_{ts}  \notag \\
&&+o_{P}\left( \frac{1}{TN\sqrt{h}}\right)  \notag \\
&\eqqcolon &J_{1}+J_{2}+2J_{3}+o_{P}\left( \frac{1}{TN\sqrt{h}}\right) .
\end{eqnarray*}
The rest proof is the same as that for the case of $\widetilde{r}=r$. Thus,
we omit the details. This completes the proof of the theorem. }
\end{proof}

\begin{proof}[Proof of Proposition \ref{prop2}]

{\small By using Lemma \ref{Lemma10} and using similar arguments as the
proof of Theorem \ref{theorem3}, for any fixed $\widetilde{r}\geq r$ we have 
\begin{eqnarray*}
L_{NT} &=&\frac{1}{T^{2}N^{2}}\sum_{i,j=1}^{N}%
\sum_{t,s=1}^{T}(1-a_{i})(1-a_{j}) \left\{ \varepsilon _{it}\varepsilon
_{js}+a_{TN}^{2}\mathbf{g}_{it}^{\top } \mathbf{f}_{t}\mathbf{g}_{js}^{\top }%
\mathbf{f}_{s}+2a_{TN}\mathbf{g} _{it}^{\top }\mathbf{f}_{t}\varepsilon
_{js}\right\} h^{-1}K_{ts}  \notag \\
&&+o_{P}\left( \frac{1}{TN\sqrt{h}}\right)  \notag \\
&\eqqcolon &J_{1}+J_{2}+2J_{3}+o_{P}\left( \frac{1}{TN\sqrt{h}}\right) .
\end{eqnarray*}
}

{\small For $J_{1}$, again, by using Lemma \ref{Lemma6}, we have 
\begin{equation*}
TN\sqrt{h}\left( J_{1}-(TNh)^{-1}E(Q_{T})\right) \rightarrow _{D}N(0,2\nu
_{0}\sigma _{\varepsilon ,a}^{4}),
\end{equation*}
where $Q_{T}=\frac{1}{TN}\sum_{i,j=1}^{N}\sum_{t,s=1}^{T}(1-a_{i})(1-a_{j})
\varepsilon _{it}\varepsilon _{js}K_{ts}$ and $\nu _{0}=\int_{-1}^{1}K^{2}(u)%
\mathrm{d}u$. }

{\small For $J_2$, by using Lemma \ref{Lemma6} and $E(\mathbf{f}_t)=0$, we
have 
\begin{eqnarray*}
TN\sqrt{h}J_2 &=& \frac{1}{TN^{2}}\sum_{i,j=1}^{N}%
\sum_{t,s=1}^{T}(1-a_{i})(1-a_{j}) \mathbf{g}_{it}^{\top } \mathbf{f}_{t}%
\mathbf{g}_{js}^{\top }\mathbf{f}_{s}K_{ts}  \notag \\
&\to_P&\frac{1}{T}\sum_{t,s=1}^{T}E\left(\overline{\mathbf{g}}_a(t/T)^\top 
\mathbf{f}_t\mathbf{f}_s^\top\overline{\mathbf{g}}_a(s/T)\right)K_{ts}.
\end{eqnarray*}
Similarly, by using Lemma \ref{Lemma6} and $E(\mathbf{f}_{t}\overline{%
\mathbf{e}}_{s,a})=0$, we have 
\begin{equation*}
TN\sqrt{h}J_3 = h^{1/4} \frac{1}{T\sqrt{h}}\sum_{t,s=1}^{T} \overline{%
\mathbf{g}}_{a}(t/T)^{\top} \mathbf{f}_{t}\overline{\mathbf{e}}%
_{s,a}K_{ts}=O_P(h^{1/4}).
\end{equation*}
}

{\small The proof is now completed. }
\end{proof}

\begin{proof}[Proof of Theorem \ref{theorem4}]

{\small Note that $\mathbf{X}=\mathcal{F}\pmb{\Theta}^{\top }+\mathcal{E}%
^{\dagger }$ and $\mathcal{E}^{\dagger }=\mathcal{E}+\mathbb{F}\mathbb{A}%
^{(-J),\top }$. As in the proof of Theorem \ref{theorem1}, we expand $L_{NT}$
as follows: 
\begin{eqnarray*}
L_{NT}&=&\frac{1}{T^{2}N^{2}}\sum_{i,j=1}^{N}\pmb{\theta}_{i}^{\top }%
\mathcal{F} ^{\top }\mathbf{M}_{\widehat{\mathbf{F}}}\mathbf{K}_{h}\mathbf{M}%
_{\widehat{ \mathbf{F}}}\mathcal{F}\pmb{\theta}_{j}+\frac{1}{T^{2}N^{2}}%
\sum_{i,j=1}^{N} \mathcal{E}_{i}^{{\dagger },\top }\mathbf{M}_{\widehat{%
\mathbf{F}}}\mathbf{K} _{h}\mathbf{M}_{\widehat{\mathbf{F}}}\mathcal{E}%
_{j}^{\dagger }  \notag \\
&&+\frac{2}{T^{2}N^{2}}\sum_{i,j=1}^{N}\mathcal{E}_{i}^{\dagger ,\top } 
\mathbf{M}_{\widehat{\mathbf{F}}}\mathbf{K}_{h}\mathbf{M}_{\widehat{\mathbf{%
F }}}\mathcal{F}\pmb{\theta}_{j}  \notag \\
&\eqqcolon & L_{T,1}+L_{T,2}+L_{T,3},
\end{eqnarray*}
where $\pmb{\theta}_{i}$ denotes the $i^{th}$ column of $\pmb{\Theta}^{\top
} $ and $\mathcal{E}_{i}^{{\dagger }}$ denotes the $i^{th}$ column of $%
\mathcal{E}^{{\dagger },\top }$. }

{\small Consider $L_{T,2}$ first, by Lemmas \ref{Lemma8}(a) and \ref{Lemma2}%
(a), we have 
\begin{equation*}
\left\Vert \frac{1}{T^{2}N^{2}}\sum_{i,j=1}^{N}\mathcal{E}_{i}^{\dagger
,\top }\mathbf{M}_{\widehat{\mathbf{F}}}\mathbf{K}_{h}\mathbf{M}_{\widehat{ 
\mathbf{F}}}\mathcal{E}_{j}^{\dagger }\right\Vert \leq \left\Vert \frac{1}{
TN }\sum_{i=1}^{N}\mathcal{E}_{i}^{\dagger ,\top }\right\Vert ^{2}\Vert 
\mathbf{M}_{\widehat{\mathbf{F}}}\Vert _{2}^{2}\Vert \mathbf{K}_{h}\Vert
_{2}=o_{P}\left( \frac{1}{T}\right) O(T)=o_{P}(1).
\end{equation*}
Similarly, for $L_{T,3}$, we have 
\begin{equation*}
\Vert L_{T,3}\Vert \leq \left\Vert \frac{1}{TN}\sum_{i=1}^{N}\mathcal{E}
_{i}^{\dagger ,\top }\right\Vert \Vert \mathbf{M}_{\widehat{\mathbf{F}}
}\Vert _{2}^{2}\frac{1}{T}\Vert \mathbf{K}_{h}\Vert _{2}\Vert \mathcal{F}
\Vert \Vert \frac{1}{N}\sum_{i=1}^{N}\pmb{\theta}_{j}\Vert =o_{P}(1/\sqrt{T}
)O_{P}(\sqrt{T})=o_{P}(1).
\end{equation*}
}

{\small Next, we consider $L_{T,1}$. We first decompose $\frac{1}{\sqrt{T}}%
\mathbf{M} _{\widehat{\mathbf{F}}}\mathcal{F}$ as follows 
\begin{eqnarray*}
\mathbf{M}_{\widehat{\mathbf{F}}}\mathcal{F} &=&\mathcal{F}-\frac{1}{T} 
\widehat{\mathbf{F}}\widehat{\mathbf{F}}^{\top }\mathcal{F}  \notag \\
&=&(\mathcal{F}-\frac{1}{T}\mathcal{F}\mathcal{H}\mathcal{H}^{\top }\mathcal{%
F}^{\top }\mathcal{F})-\frac{1}{T}(\widehat{\mathbf{F}}-\mathcal{F}\mathcal{%
H })(\widehat{\mathbf{F}}-\mathcal{F}\mathcal{H})^{\top }\mathcal{F}  \notag
\\
&&-\frac{1}{T}(\widehat{\mathbf{F}}-\mathcal{F}\mathcal{H})\mathcal{H}^{\top
}\mathcal{F}^{\top }\mathcal{F}-\frac{1}{T}\mathcal{F}\mathcal{H}(\widehat{ 
\mathbf{F}}-\mathcal{F}\mathcal{H})^{\top }\mathcal{F}  \notag \\
&\eqqcolon &\mathbf{I}_{1}+\mathbf{I}_{2}+\mathbf{I}_{3}+\mathbf{I}_{4},
\end{eqnarray*}
where $\mathcal{H}=\left( \frac{\pmb{\Theta}^{\top }\pmb{\Theta}}{N}\right)
\left( \frac{\mathcal{F}^{\top }\widehat{\mathbf{F}}}{T}\right) \left( \frac{
1}{TN}\widehat{\mathbf{V}}\right) ^{-1}$. By Lemma \ref{Lemma8} (b), $\Vert 
\frac{1}{\sqrt{T}}\mathbf{I}_{2}\Vert =o_{P}(1)$, $\Vert \frac{1}{\sqrt{T}} 
\mathbf{I}_{3}\Vert =o_{P}(1)$ and $\Vert \frac{1}{\sqrt{T}}\mathbf{I}
_{4}\Vert =o_{P}(1)$. Let $\mathbf{R}_{TN}=(\pmb{\Theta}^{\top }\pmb{\Theta}
/N)^{1/2}\mathcal{F}^{\top }\widehat{\mathbf{F}}/T$ and 
\begin{equation*}
\mathbf{B}_{TN}=(\pmb{\Theta}^{\top }\pmb{\Theta}/N)^{1/2}(\mathcal{F}^{\top
}\mathcal{F}/T)(\pmb{\Theta}^{\top }\pmb{\Theta}/N)^{1/2}.
\end{equation*}
For $\mathbf{I}_{1}$, by using Lemmas \ref{Lemma8} (c)--(d) and $%
\pmb{\Upsilon}^{\top }\pmb{\Delta}=\mathbb{V}\pmb{\Upsilon}^{\top }$, we
have 
\begin{eqnarray*}
\frac{1}{\sqrt{T}}(\mathcal{F}-\frac{1}{T}\mathcal{F}\mathcal{H}\mathcal{H}
^{\top }\mathcal{F}^{\top }\mathcal{F}) &=&\frac{1}{\sqrt{T}}\mathcal{F}( %
\pmb{\Theta}^{\top }\pmb{\Theta}/N)^{1/2}(\mathbf{I}_{J}-\mathbf{R}_{TN}( 
\frac{1}{TN}\widehat{\mathbf{V}})^{-2}\mathbf{R}_{TN}^{\top }\mathbf{B}
_{TN})(\pmb{\Theta}^{\top }\pmb{\Theta}/N)^{-1/2}  \notag \\
&=&\frac{1}{\sqrt{T}}\mathcal{F}(\mathbf{S}^{(J)}/\sqrt{TN})(\mathbf{I}_{J}- %
\pmb{\Upsilon}\pmb{\Upsilon}^{\top })(\mathbf{S}^{(J)}/\sqrt{TN}
)^{-1}+o_{P}(1).
\end{eqnarray*}
It follows that 
\begin{eqnarray*}
L_{T,1} &=&\frac{1}{T^{2}}\overline{\pmb{\theta}}^{\top }(\mathbf{S}^{(J)}/ 
\sqrt{TN})^{-1}(\mathbf{I}_{J}-\pmb{\Upsilon}\pmb{\Upsilon}^{\top })(\mathbf{%
S}^{(J)}/\sqrt{TN})  \notag \\
&&\times \mathcal{F}^{\top }\mathbf{K}_{h}\mathcal{F}(\mathbf{S}^{(J)}/\sqrt{
TN})(\mathbf{I}_{J}-\pmb{\Upsilon}\pmb{\Upsilon}^{\top })(\mathbf{S}^{(J)}/ 
\sqrt{TN})^{-1}\overline{\pmb{\theta}}+o_{P}(1),
\end{eqnarray*}
where $\overline{\pmb{\theta}}=\frac{1}{N}\sum_{i=1}^{N}\pmb{\theta}_{i}$.
In addition, by using similar arguments as used in the proof of Lemma \ref%
{Lemma6}, we can show the convergence of the term $\frac{1}{T^{2}}\mathcal{F}
^{\top }\mathbf{K}_{h}\mathcal{F}$. This completes the proof of the theorem. 
}
\end{proof}

{\small
\setlength{\bibsep}{3.pt plus 0ex}
\bibliography{paper-ref.bib}
}

\newpage 

\setcounter{page}{1}

\begin{center}
{\large \textbf{Online Appendix B to 
\textquotedblleft A Robust Residual-Based Test for Structural Changes in
Factor Models\textquotedblright }} 

{\small \medskip }

{\small \textsc{Bin Peng}$^{a}$ and \textsc{Liangjun Su}$^{b}$ and \textsc{%
Yayi Yan}$^{c}$ }

{\small \medskip }

{\small $^{a}$Monash University, $^{b}$Tsinghua University, $^{c}$Shanghai University of
Finance and Economics}
\end{center}

\small

\renewcommand{\theequation}{B.\arabic{equation}}
\renewcommand{\thesection}{B.\arabic{section}}
\renewcommand{\thefigure}{B.\arabic{figure}}
\renewcommand{\thetable}{B.\arabic{table}}
\renewcommand{\thelemma}{B.\arabic{lemma}}
\renewcommand{\theremark}{B.\arabic{remark}}
\renewcommand{\thecorollary}{B.\arabic{corollary}}
\renewcommand{\theassumption}{B.\arabic{assumption}}

\setcounter{equation}{0}
\setcounter{lemma}{0}
\setcounter{section}{0}
\setcounter{table}{0}
\setcounter{figure}{0}
\setcounter{remark}{0}
\setcounter{corollary}{0}
\setcounter{assumption}{0}

{\small This appendix includes three sections. Section \ref{App.B1} contains
the proofs of the technical lemmas of Appendix \ref{App.A2}. Section \ref%
{App.B2} verifies Assumption \ref{Assumption4} in the paper. Section \ref%
{App.B3} contains some additional simulation results. }

\section{Proofs of the Technical Lemmas\label{App.B1}}

{\small 
\begin{proof}[Proof of Lemma \ref{Lemma1}]

\noindent (a). Define the projection operator $\mathcal{P}_{t}(\cdot
)=E[\cdot \,|\, \mathscr{E}_{t}]-E[\cdot \,|\, \mathscr{E}_{t-1}]$. Note that $%
\varepsilon _{it}\varepsilon _{jt}-E(\varepsilon _{it}\varepsilon
_{jt})=\sum_{k=0}^{\infty }\mathcal{P}_{t-k}(\varepsilon _{it}\varepsilon
_{jt})$. Let $\varepsilon _{it,\{k\}}$ be the coupled version of $%
\varepsilon _{it}$ replacing $\mathbf{e}_{t-k}$ with $\mathbf{e}%
_{t-k}^{\prime }$. By using Jensen inequality and using Assumption \ref%
{Assumption4}, we have 
\begin{eqnarray*}
|\mathcal{P}_{t-k}(\varepsilon _{it}\varepsilon _{jt})|_{2}
&=&|E(\varepsilon _{it}\varepsilon _{jt}\,|\, \mathscr{E}%
_{t-k})-E(\varepsilon _{it,\{t-k\}}\varepsilon _{jt,\{t-k\}}\,|\, %
\mathscr{E}_{t-k})|_{2} \notag \\
&\leq &|\varepsilon _{it}\varepsilon _{jt}-\varepsilon
_{it,\{t-k\}}\varepsilon _{jt,\{t-k\}}|_{2} \notag \\
&\leq &\left( |\varepsilon _{it}-\varepsilon
_{it,\{t-k\}}|_{4}+|\varepsilon _{jt}-\varepsilon
_{jt,\{t-k\}}|_{4}\right) (|\varepsilon _{it}|_{4}+|\varepsilon
_{jt}|_{4})=O(\lambda _{4}^{\varepsilon }(k)).
\end{eqnarray*}%
Since $\{\mathcal{P}_{t-k}(\varepsilon _{it}\varepsilon_{jt})\}_{t\leq T}$
is a martingale difference sequence (m.d.s.), by Burkholder and Minkowski
inequalities, we have 
\begin{equation*}
\left\vert \sum_{t=1}^{T}\mathcal{P}_{t-k}(\varepsilon _{it}\varepsilon
_{jt})\right\vert _{2}^{2}\leq O(1)\sum_{t=1}^{T}|\mathcal{P}%
_{t-k}(\varepsilon _{it}\varepsilon _{jt})|_{2}^{2}=O\left( T(\lambda
_{4}^{\varepsilon }(k))^{2}\right) .
\end{equation*}%
It follows that 
\begin{equation*}
\left|\sum_{t=1}^{T}\sum_{k=0}^{\infty }\mathcal{P}_{t-k}(\varepsilon
_{it}\varepsilon _{jt})\right|_{2}\leq \sum_{k=0}^{\infty }\left|\sum_{t=1}^{T}%
\mathcal{P}_{t-k}(\varepsilon _{it}\varepsilon _{jt})\right|_{2}=O(\sqrt{T})
\end{equation*}%
since $\sum_{k=0}^{\infty }\lambda _{4}^{\varepsilon }(k)$ is bounded by
Assumption \ref{Assumption4}.

\medskip

\noindent (b). By Jensen inequality and Assumptions \ref{Assumption4} and %
\ref{Assumption3}, we have 
\begin{eqnarray*}
|\mathcal{P}_{t-k}(\mathbf{f}_{t}\overline{\varepsilon }_{t,\mathbf{v}%
})|_{2} &=&|E(\mathbf{f}_{t}\overline{\varepsilon }_{t,\mathbf{v}}\,|\, %
\mathscr{E}_{t-k})-E(\mathbf{f}_{t,\{t-k\}}\overline{\varepsilon }_{t,%
\mathbf{v},\{t-k\}}\,|\, \mathscr{E}_{t-k})|_{2} \notag \\
&\leq &|\mathbf{f}_{t}-\mathbf{f}_{t,\{t-k\}}|_{4}\cdot |\overline{%
\varepsilon }_{t,\mathbf{v}}|_{4}+|\mathbf{f}_{t,\{t-k\}}|_{4}\cdot |%
\overline{\varepsilon }_{t,\mathbf{v}}-\overline{\varepsilon }_{t,\mathbf{v}%
,\{t-k\}}|_{4}=O(k^{-\alpha }).
\end{eqnarray*}%
Then by using martingale decomposition and Burkholder inequality, the proof
of part (b) is similar to that in part (a) and thus omitted here.

\medskip

\noindent (c). By Assumptions \ref{Assumption4},
martingale decomposition and Burkholder inequality, the proof of part (c) is
identical to that in parts (a)--(b) and thus omitted here. 
\end{proof}
}

{\small 
\begin{proof}[Proof of Lemma \ref{Lemma2}]
		
\noindent (a). Note that $\mathbf{K}_h$ is a symmetric matrix, and thus $$\| 
\mathbf{K}_h\|_2 \leq \| \mathbf{K}_h\|_\infty = \max_{s\ge 1} \sum_{t=1}^T
h^{-1}K((t-s)/(Th)) = O(T).$$

\medskip

\noindent (b). Note that $\left\Vert \frac{1}{TN}\mathcal{E}\mathcal{E}%
^{\top }\right\Vert =\left\Vert \frac{1}{TN}\mathcal{E}^{\top }\mathcal{E}%
\right\Vert \leq \left\Vert \frac{1}{TN}\sum_{t=1}^{T}(\pmb{\varepsilon}_{t}%
\pmb{\varepsilon}_{t}^{\top }-\pmb{\Sigma_{\varepsilon}})\right\Vert +\Vert 
\frac{1}{N}\pmb{\Sigma_{\varepsilon}}\Vert $. Note that $\Vert \frac{1}{N}%
\pmb{\Sigma_{\varepsilon}}\Vert $ is bounded by $\frac{1}{N}\sqrt{N}\Vert %
\pmb{\Sigma_{\varepsilon}}\Vert _{2}=O(1/\sqrt{N})$. For the first term, by
using Lemma \ref{Lemma1}(a) we have 
\begin{eqnarray*}
E\left\Vert \frac{1}{TN}\sum_{t=1}^{T}(\pmb{\varepsilon}_{t}\pmb{\varepsilon}%
_{t}^{\top }-\pmb{\Sigma_{\varepsilon}})\right\Vert  &\leq &\left\{
E\left\Vert \frac{1}{TN}\sum_{t=1}^{T}(\pmb{\varepsilon}_{t}\pmb{\varepsilon}%
_{t}^{\top }-\pmb{\Sigma_{\varepsilon}})\right\Vert ^{2}\right\} ^{1/2} \notag \\
&=&\frac{1}{N}\sqrt{\sum_{i,j=1}^{N}E\left( \frac{1}{T}\sum_{t=1}^{T}[%
	\varepsilon _{it}\varepsilon _{jt}-E(\varepsilon _{it}\varepsilon
	_{jt})]\right) ^{2}}=O(1/\sqrt{T}).
\end{eqnarray*}%
Then the result holds.

\medskip

\noindent (c). By Lemma \ref{Lemma1}(b), we have 
\begin{equation*}
\left\Vert \frac{1}{TN}\pmb{\Lambda}^{\top }\mathcal{E}^{\top }\mathbf{F}%
\right\Vert =\left\Vert \frac{1}{TN}\sum_{t=1}^{T}\sum_{i=1}^{N}\pmb{\lambda}%
_{i}\varepsilon _{it}\mathbf{f}_{t}^{\top }\right\Vert =O_{P}\left( \frac{1%
}{\sqrt{TN}}\right) .
\end{equation*}

\medskip

\noindent (d). By Assumption \ref{Assumption2} and the condition that $\Vert %
\pmb{\Sigma_{\varepsilon}}\Vert _{2}\leq c_{3}<\infty $ under Assumption \ref%
{Assumption4}(a) 
\begin{eqnarray*}
E\left\Vert \frac{1}{TN}\mathcal{E}\pmb{\Lambda}\right\Vert ^{2} &=&\frac{1}{%
	T^{2}N^{2}}\sum_{t=1}^{T}E\left\{ \mathrm{tr}\left( \pmb{\Lambda}^{\top }%
\pmb{\varepsilon}_{t}\pmb{\varepsilon}_{t}^{\top }\pmb{\Lambda}\right)
\right\}\le \frac{1}{TN}\Vert \pmb{\Sigma_{\varepsilon}}\Vert _{2}\mathrm{tr}%
\left( N^{-1}\pmb{\Lambda}^{\top }\pmb{\Lambda}\right) =O\left( \frac{1}{TN}%
\right) .
\end{eqnarray*}

		\medskip
		
\noindent (e). Note that 
\begin{equation*}
E\left\Vert \frac{1}{TN}\mathcal{E}\mathcal{E}^{\top }\mathbf{F}\right\Vert
\leq \left\{ \sum_{s=1}^{T}E\left\Vert \frac{1}{TN}\sum_{t=1}^{T}\mathbf{f}%
_{t}\pmb{\varepsilon}_{t}^{\top }\pmb{\varepsilon}_{s}\right\Vert
^{2}\right\} ^{1/2}
\end{equation*}%
and 
\begin{equation*}
\left\vert \frac{1}{TN}\sum_{t=1}^{T}\mathbf{f}_{t}\pmb{\varepsilon}%
_{t}^{\top }\pmb{\varepsilon}_{s}\right\vert _{2}\leq \left\vert \frac{1}{TN}%
\sum_{t=1}^{T}\mathbf{f}_{t}(\pmb{\varepsilon}_{t}^{\top }\pmb{\varepsilon}%
_{s}-E(\pmb{\varepsilon}_{t}^{\top }\pmb{\varepsilon}_{s}))\right\vert
_{2}+\left\vert \frac{1}{TN}\sum_{t=1}^{T}\mathbf{f}_{t}E(\pmb{\varepsilon}%
_{t}^{\top }\pmb{\varepsilon}_{s})\right\vert _{2}.
\end{equation*}%
By Lemma \ref{Lemma1}(c), the first term on the r.h.s. is $O(1/\sqrt{TN})$.
For the second term, we have that for $\alpha >1$ 
\begin{equation*}
\left|\frac{1}{TN}\sum_{t=1}^{T}\mathbf{f}_{t}E(\pmb{\varepsilon}_{t}^{\top }%
\pmb{\varepsilon}_{s})\right|_{2}\leq O(1)\frac{1}{T}\sum_{t=1}^{T}\max_{i}|E(%
\varepsilon _{it}\varepsilon _{is})|=O(1)\frac{1}{T}\sum_{t=1}^{T}|t-s|^{-%
	\alpha }=O\left( 1/T\right) 
\end{equation*}%
provided that 
\begin{eqnarray*}
|E[\varepsilon _{i0}\varepsilon _{ij}]| &=&\left\vert E\left[ \left(
\sum_{l=0}^{\infty }\mathscr{P}_{-l}(\varepsilon _{i0})\right) \left(
\sum_{l=0}^{\infty }\mathscr{P}_{j-l}(\varepsilon _{ij})\right) \right]
\right\vert \notag \\
&=&\left\vert E\left[ \sum_{l=0}^{\infty }(E[\varepsilon _{i0}\,|\, %
\mathscr{E}_{-l}]-E[\varepsilon _{i0}\,|\, \mathscr{E}_{-l-1}])\cdot
(E[\varepsilon _{ij}\,|\, \mathscr{E}_{-l}]-E[\varepsilon _{ij}\,|\, %
\mathscr{E}_{-l-1}])\right] \right\vert \notag \\
&\leq &\sum_{l=0}^{\infty }\Vert E[\varepsilon _{i0}\,|\, \mathscr{E}%
_{-l}]-E[\varepsilon _{i0}\,|\, \mathscr{E}_{-l-1}]\Vert _{2}\cdot \Vert
E[\varepsilon _{ij}\,|\, \mathscr{E}_{-l}]-E[\varepsilon _{ij}\,|\, %
\mathscr{E}_{-l-1}]\Vert _{2} \notag \\
&=&\sum_{l=0}^{\infty }\Vert E[\varepsilon _{il}\,|\, \mathscr{E}%
_{0}]-E[\varepsilon _{il}^{\ast }\,|\, \mathscr{E}_{0}]\Vert _{2}\cdot \Vert
E[\varepsilon _{i,j+l}\,|\, \mathscr{E}_{0}]-E[\varepsilon _{i,j+l}^{\ast
}\,|\, \mathscr{E}_{0}]\Vert _{2} \notag \\
&\leq &\sum_{l=0}^{\infty }\Vert \varepsilon _{il}-\varepsilon _{il}^{\ast
}\Vert _{2}\Vert \varepsilon _{i,j+l}-\varepsilon _{i,j+l}^{\ast }\Vert
_{2}\leq O(j^{-\alpha }).
\end{eqnarray*}%
Then we have $E\left\Vert \frac{1}{TN}\mathcal{E}\mathcal{E}^{\top }\mathbf{F%
}\right\Vert =O({1}/\sqrt{T}+1/\sqrt{N})$.

\medskip

\noindent (f). As in the proof of part (c), part (f) follows directly from
Lemma \ref{Lemma1}(b).

\medskip

\noindent (g). As in the proof of part (c), part (g) follows directly from
Lemma \ref{Lemma1}(b).

\end{proof}
}

{\small 
\begin{proof}[Proof of Lemma \ref{Lemma3}]
		
\noindent (a). Write

\begin{eqnarray*}
\widehat{\mathbf{F}}\cdot \frac{1}{TN}\widehat{\mathbf{V}} &=&\frac{1}{TN}%
\mathbf{X}\mathbf{X}^{\top }\widehat{\mathbf{F}}=\frac{1}{TN}(\mathbf{F}%
\pmb{\Lambda}^{\top }+\mathcal{E})(\mathbf{F}\pmb{\Lambda}^{\top }+\mathcal{E%
})^{\top }\widehat{\mathbf{F}} \notag \\
&=&\frac{1}{TN}\mathbf{F}\pmb{\Lambda}^{\top }\pmb{\Lambda}\mathbf{F}^{\top }%
\widehat{\mathbf{F}}+\frac{1}{TN}\mathbf{F}\pmb{\Lambda}^{\top }\mathcal{E}%
^{\top }\widehat{\mathbf{F}}+\frac{1}{TN}\mathcal{E}\pmb{\Lambda}\mathbf{F}%
^{\top }\widehat{\mathbf{F}}+\frac{1}{TN}\mathcal{E}\mathcal{E}^{\top }%
\widehat{\mathbf{F}}.
\end{eqnarray*}%
Note that 
\begin{equation*}
\frac{1}{\sqrt{T}}\Vert \frac{1}{TN}\mathbf{F}\pmb{\Lambda}^{\top }\mathcal{E%
}^{\top }\widehat{\mathbf{F}}\Vert =O_{P}\left( \frac{1}{\sqrt{N}}\right) \
\ \text{ and }\ \ \frac{1}{\sqrt{T}}\Vert \frac{1}{TN}\mathcal{E}\mathcal{E}%
^{\top }\widehat{\mathbf{F}}\Vert =O_{P}\left( \frac{1}{\sqrt{T\wedge N}}%
\right) 
\end{equation*}%
by Lemma \ref{Lemma2}(d) and (b), respectively and the fact that $\Vert 
\mathbf{F}/\sqrt{T}\Vert =O_{P}(1)$ and $\Vert \widehat{\mathbf{F}}/\sqrt{T}%
\Vert =\sqrt{r}$. Thus, the first result follows.

\medskip

\noindent (b). Write 
\begin{equation*}
\frac{1}{T}\mathbf{F}^{\top }(\widehat{\mathbf{F}}-\mathbf{F}\mathbf{H})%
\frac{1}{TN}\widehat{\mathbf{V}}=\frac{1}{T^{2}N}\mathbf{F}^{\top }\mathbf{F}%
\pmb{\Lambda}^{\top }\mathcal{E}^{\top }\widehat{\mathbf{F}}+\frac{1}{T^{2}N}%
\mathbf{F}^{\top }\mathcal{E}\pmb{\Lambda}\mathbf{F}^{\top }\widehat{\mathbf{%
		F}}+\frac{1}{T^{2}N}\mathbf{F}^{\top }\mathcal{E}\mathcal{E}^{\top }\widehat{%
	\mathbf{F}}.
\end{equation*}%
For the first term, by using Lemma \ref{Lemma2}(c)--(d) and the result in
part (a), we have 
\begin{eqnarray*}
\frac{1}{T^{2}N}\mathbf{F}^{\top }\mathbf{F}\pmb{\Lambda}^{\top }\mathcal{E}%
^{\top }\widehat{\mathbf{F}} &=&\frac{1}{T^{2}N}\mathbf{F}^{\top }\mathbf{F}%
\pmb{\Lambda}^{\top }\mathcal{E}^{\top }\mathbf{F}\mathbf{H}+\frac{1}{T^{2}N}%
\mathbf{F}^{\top }\mathbf{F}\pmb{\Lambda}^{\top }\mathcal{E}^{\top }(%
\widehat{\mathbf{F}}-\mathbf{F}\mathbf{H}) \notag \\
&=&O_{P}(1/\sqrt{TN})+O_{P}(\sqrt{T}/\sqrt{TN(T\wedge N)}).
\end{eqnarray*}%
Similarly, by Lemma \ref{Lemma2}(c) and (e), we have $\frac{1}{T^{2}N}%
\mathbf{F}^{\top }\mathcal{E}\pmb{\Lambda}\mathbf{F}^{\top }\widehat{\mathbf{%
	F}}=O_{P}(1/\sqrt{TN})$ and $\frac{1}{T^{2}N}\mathbf{F}^{\top }\mathcal{E}%
\mathcal{E}^{\top }\widehat{\mathbf{F}}=O_{P}(1/\sqrt{T(T\wedge N)})$. We
then obtain the second result.

\medskip

\noindent (c). Write 
\begin{eqnarray*}
\frac{1}{TN}\mathcal{E}\mathcal{E}^{\top }(\widehat{\mathbf{F}}-\mathbf{F}%
\mathbf{H})\widehat{\mathbf{V}}/(TN) &=&\frac{1}{T^{2}N^{2}}\mathcal{E}%
\mathcal{E}^{\top }\mathbf{F}\pmb{\Lambda}^{\top }\mathcal{E}^{\top }%
\widehat{\mathbf{F}}+\frac{1}{T^{2}N^{2}}\mathcal{E}\mathcal{E}^{\top }%
\mathcal{E}\pmb{\Lambda}\mathbf{F}^{\top }\widehat{\mathbf{F}} \notag \\
&&+\frac{1}{T^{2}N^{2}}\mathcal{E}\mathcal{E}^{\top }\mathcal{E}\mathcal{E}%
^{\top }\widehat{\mathbf{F}}.
\end{eqnarray*}%
For the first term, 
\begin{eqnarray*}
\frac{1}{T^{2}N^{2}}\mathcal{E}\mathcal{E}^{\top }\mathbf{F}\pmb{\Lambda}%
^{\top }\mathcal{E}^{\top }\widehat{\mathbf{F}} &=&\frac{1}{T^{2}N^{2}}%
\mathcal{E}\mathcal{E}^{\top }\mathbf{F}\pmb{\Lambda}^{\top }\mathcal{E}%
^{\top }\mathbf{F}\mathbf{H}+\frac{1}{T^{2}N^{2}}\mathcal{E}\mathcal{E}%
^{\top }\mathbf{F}\pmb{\Lambda}^{\top }\mathcal{E}^{\top }(\widehat{\mathbf{F%
}}-\mathbf{F}\mathbf{H}) \notag \\
&=&O_{P}\left( \frac{1}{\sqrt{TN(T\wedge N)}}\right) +O_{P}\left( \frac{%
	\sqrt{T}}{\sqrt{TN}(T\wedge N)}\right) 
\end{eqnarray*}%
by using Lemmas \ref{Lemma2}(c), (e) and \ref{Lemma3}(a). Similarly, by
Lemmas \ref{Lemma2} and \ref{Lemma3}(a), we can show that the second term is 
$O_{P}\left( \sqrt{T}/\sqrt{N(T\wedge N)}\right) $ and the third term is $%
O_{P}(\sqrt{T}/(T\wedge N)^{3/2})$. It follows that $\Vert \frac{1}{TN}%
\mathcal{E}\mathcal{E}^{\top }(\widehat{\mathbf{F}}-\mathbf{F}\mathbf{H}%
)\Vert =O_{P}(\sqrt{T}/(T\wedge N)^{3/2}).$

\medskip

\noindent (d). Write 
\begin{eqnarray*}
&&\sum_{i=1}^{N}(1-a_{i})\mathcal{E}_{i}^{\top }(\widehat{\mathbf{F}}-%
\mathbf{F}\mathbf{H}) \notag \\
&=&\sum_{i=1}^{N}(1-a_{i})\mathcal{E}_{i}^{\top }\left( \frac{1}{TN}\mathbf{F%
}\pmb{\Lambda}^{\top }\mathcal{E}^{\top }\widehat{\mathbf{F}}+\frac{1}{TN}%
\mathcal{E}\pmb{\Lambda}\mathbf{F}^{\top }\widehat{\mathbf{F}}+\frac{1}{TN}%
\mathcal{E}\mathcal{E}^{\top }\widehat{\mathbf{F}}\right) (\frac{1}{TN}%
\widehat{\mathbf{V}})^{-1}.
\end{eqnarray*}%
For the first term, by using Lemmas \ref{Lemma2}(f) and (c) and \ref{Lemma3}%
(a), we have 
\begin{eqnarray*}
&&\left\Vert \frac{1}{T^{2}N^{2}}\sum_{i=1}^{N}(1-a_{i})\mathcal{E}%
_{i}^{\top }\mathbf{F}\pmb{\Lambda}^{\top }\mathcal{E}^{\top }\widehat{%
	\mathbf{F}}\right\Vert \notag \\
&\leq &\left\Vert \frac{1}{T^{2}N^{2}}\sum_{i=1}^{N}(1-a_{i})\mathcal{E}%
_{i}^{\top }\mathbf{F}\pmb{\Lambda}^{\top }\mathcal{E}^{\top }\mathbf{F}%
\mathbf{H}\right\Vert +\left\Vert \frac{1}{T^{2}N^{2}}\sum_{i=1}^{N}(1-a_{i})%
\mathcal{E}_{i}^{\top }\mathbf{F}\pmb{\Lambda}^{\top }\mathcal{E}^{\top }(%
\widehat{\mathbf{F}}-\mathbf{F}\mathbf{H})\right\Vert  \notag \\
&=&O_{P}(1/(TN))+O_{P}(\sqrt{T}/(TN\sqrt{T\wedge N})).
\end{eqnarray*}%
For the second term, by using Lemmas \ref{Lemma2} (g) and (d), we have 
\begin{equation*}
\left\Vert \frac{1}{T^{2}N^{2}}\sum_{i=1}^{N}(1-a_{i})\mathcal{E}_{i}^{\top }%
\mathcal{E}\pmb{\Lambda}\mathbf{F}^{\top }\widehat{\mathbf{F}}\right\Vert
=O_{P}(1/N).
\end{equation*}%
Similarly, for the third term, we have 
\begin{equation*}
\left\Vert \frac{1}{T^{2}N^{2}}\sum_{i=1}^{N}(1-a_{i})\mathcal{E}_{i}^{\top }%
\mathcal{E}\mathcal{E}^{\top }\widehat{\mathbf{F}}\right\Vert =O_{P}(1/(%
\sqrt{TN}\sqrt{T\wedge N}))+O_{P}(\sqrt{T}/(\sqrt{TN}(T\wedge N))).
\end{equation*}%
This completes the proof. 
\end{proof}
}

{\small 
\begin{proof} [Proof of Lemma \ref{Lemma4}]

\noindent (a). Define 
\begin{equation*}
D_{t,j}=E[\overline{\varepsilon }_{t,\mathbf{v}}\,|\, \mathscr{E}_{t-j,t}]-E[%
\overline{\varepsilon }_{t,\mathbf{v}}\,|\, \mathscr{E}_{t-j+1,t}].
\end{equation*}%
Then $\{D_{t,j}\}_{t\geq 1}$ forms a m.d.s. with respect to $\mathscr{E}%
_{t-j,\infty }$, and $\Vert D_{t,j}\Vert _{4}\leq \lambda _{4}^{\varepsilon
}(j)$. Similar to the proof of Lemma \ref{Lemma1}, by Burkholder's
inequality and Minkowski's inequality, we have 
\begin{equation*}
\left\Vert \sum_{t=1}^{T}b_{t}D_{t,j}\right\Vert _{4}^{2}\leq
O(1)\sum_{t=1}^{T}\Vert b_{t}D_{t,j}\Vert
_{4}^{2}=O(1)\sum_{t=1}^{T}|b_{t}|^{2}\lambda _{4}^{\varepsilon ,2}(j).
\end{equation*}%
Since $\overline{\varepsilon }_{t,\mathbf{v}}-\widetilde{\overline{%
\varepsilon }}_{t,\mathbf{v}}=\sum_{j=m+1}^{\infty }D_{t,j}$, the result
follows.

\medskip

\noindent (b). In what follows, let 
\begin{eqnarray*}
Z_{t} &=&\sum_{j=1}^{t}K(\frac{t+1-j}{Th})\overline{\varepsilon }_{j,\mathbf{%
v}},\quad \widetilde{Z}_{t}=\sum_{j=1}^{t}K(\frac{t+1-j}{Th})\widetilde{%
\overline{\varepsilon }}_{j,\mathbf{v}},\text{ and} \notag \\
J_{T}^{\diamond } &=&\sum_{1\leq j<j^{\prime }\leq T}K(\frac{j^{\prime }-j}{%
Th})\overline{\varepsilon }_{j^{\prime },\mathbf{v}}\widetilde{\overline{%
\varepsilon }}_{j,\mathbf{v}}=\sum_{t=2}^{T}\overline{\varepsilon }_{t,%
\mathbf{v}}\widetilde{Z}_{t-1}.
\end{eqnarray*}%
Let $Z_{t,\{k\}},\overline{\varepsilon }_{t,\mathbf{v},\{k\}}$ be the
coupled version of $Z_{t}$ and $\overline{\varepsilon }_{t,\mathbf{v}}$
replacing $\mathbf{e}_{k}$ with $\mathbf{e}_{k}^{\prime }$. By Jensen's
inequality, 
\begin{eqnarray*}
|\mathscr{P}_{k}(J_{T}-J_{T}^{\diamond })|_{2} &\leq &\left\vert
\sum_{t=2}^{T}[\overline{\varepsilon }_{t,\mathbf{v}}(Z_{t-1}-\widetilde{Z}%
_{t-1})-\overline{\varepsilon }_{t,\mathbf{v},\{k\}}(Z_{t-1,\{k\}}-%
\widetilde{Z}_{t-1,\{k\}})]\right\vert _{2} \notag \\
&\leq &\left\vert \sum_{t=2}^{T}[\overline{\varepsilon }_{t,\mathbf{v}%
,\{k\}}(Z_{t-1}-\widetilde{Z}_{t-1}-Z_{t-1,\{k\}}+\widetilde{Z}%
_{t-1,\{k\}})]\right\vert _{2} \notag \\
&&+\left\vert \sum_{t=2}^{T}[\overline{\varepsilon }_{t,\mathbf{v}}-%
\overline{\varepsilon }_{t,\mathbf{v},\{k\}}](Z_{t-1}-\widetilde{Z}%
_{t-1})\right\vert _{2}\notag \\
&\eqqcolon& I_{1}+I_{2}.
\end{eqnarray*}

We consider $I_{1}$ first. Noting that $|\widetilde{\overline{\varepsilon }}%
_{t,\mathbf{v}}-\widetilde{\overline{\varepsilon }}_{t,\mathbf{v}%
,\{k\}}|_{4}\leq \lambda _{4}^{\varepsilon }(t-k)$ and $|\overline{%
\varepsilon }_{t,\mathbf{v}}-\widetilde{\overline{\varepsilon }}_{t,\mathbf{v%
}}|_{4}\leq (\sum_{j=m+1}^{\infty }\lambda _{4}^{\varepsilon ,2}(j))^{1/2},$
we obtain that 
\begin{equation*}
|\widetilde{\overline{\varepsilon }}_{t,\mathbf{v}}-\widetilde{\overline{%
\varepsilon }}_{t,\mathbf{v},\{k\}}-\overline{\varepsilon }_{t,\mathbf{v}}+%
\widetilde{\overline{\varepsilon }}_{t,\mathbf{v}}|_{4}\leq 2\min \left(
\lambda _{4}^{\varepsilon }(t-k),(\sum_{j=m+1}^{\infty }\lambda
_{4}^{\varepsilon ,2}(j))^{1/2}\right) .
\end{equation*}

By Lemma \ref{Lemma1} (a), we have 
\begin{eqnarray*}
&&I_{1}\leq \left\vert \sum_{t=1}^{T-1}(\widetilde{\overline{\varepsilon }}%
_{t,\mathbf{v}}-\widetilde{\overline{\varepsilon }}_{t,\mathbf{v},\{k\}}-%
\overline{\varepsilon }_{t,\mathbf{v}}+\widetilde{\overline{\varepsilon }}%
_{t,\mathbf{v}})\sum_{s=t+1}^{T}K(\frac{t-s}{Th})\overline{\varepsilon }_{t,%
\mathbf{v},\{k\}}\right\vert _{2} \notag \\
&\leq &\sum_{t=1}^{T-1}\left\vert \widetilde{\overline{\varepsilon }}_{t,%
\mathbf{v}}-\widetilde{\overline{\varepsilon }}_{t,\mathbf{v},\{k\}}-%
\overline{\varepsilon }_{t,\mathbf{v}}+\widetilde{\overline{\varepsilon }}%
_{t,\mathbf{v}}\right\vert _{\delta }\cdot \left\vert \sum_{s=t+1}^{T}K(%
\frac{t-s}{Th})\overline{\varepsilon }_{t,\mathbf{v},\{k\}}\right\vert
_{\delta } \notag \\
&=&O\left( (\sum_{t=1}^{T}K^{2}(\frac{t}{Th}))^{1/2}\right)
\sum_{t=1}^{T-1}\min \left( \lambda _{4}^{\varepsilon
}(t-k),(\sum_{j=m+1}^{\infty }\lambda _{4}^{\varepsilon ,2}(j))^{1/2}\right)
.
\end{eqnarray*}

Now, we consider $I_{2}$. By part (a), we have 
\begin{equation*}
\max_{1\leq t\leq T}|Z_{t-1}-\widetilde{Z}_{t-1}|_{4}\leq
O(1)(\sum_{t=1}^{T}K^{2}(\frac{t}{Th}))^{1/2}\sum_{j=m+1}^{\infty }\lambda
_{4}^{\varepsilon }(j).
\end{equation*}%
Note that $I_{2}$ is actually $I_{2,k}$ where the sub-index $k$ is
suppressed previously for notational simplicity. Hence, we have 
\begin{eqnarray*}
\sum_{k=-\infty }^{T}I_{2,k}^{2} &\leq &O(1)\sum_{k=-\infty }^{T}\left(
\sum_{t=1}^{T}K^{2}(\frac{t}{Th})\right) \left( \sum_{j=m+1}^{\infty
}\lambda _{4}^{\varepsilon }(j)\right) ^{2}(\sum_{t=1}^{T}\lambda
_{4}^{\varepsilon }(t-k))^{2} \notag \\
&\leq &O(1)\left( \sum_{t=1}^{T}K^{2}(\frac{t}{Th})\right) \left(
\sum_{j=m+1}^{\infty }\lambda _{4}^{\varepsilon }(j)\right)
^{2}\sum_{j=0}^{\infty }\lambda _{4}^{\varepsilon
}(j)\sum_{t=1}^{T}\sum_{k=-\infty }^{T}\lambda _{4}^{\varepsilon }(t-k) \notag \\
&=&O(1)T\left( \sum_{t=1}^{T}K^{2}(\frac{t}{Th})\right) \left(
\sum_{j=m+1}^{\infty }\lambda _{4}^{\varepsilon }(j)\right) ^{2}.
\end{eqnarray*}%
Similarly, $\sum_{k=-\infty }^{T}I_{1,k}^{2}=O(1)T\left( \sum_{t=1}^{T}K^{2}(%
\frac{t}{Th})\right) d_{m}^{2}$. Since $\sum_{j=m+1}^{\infty }\lambda
_{4}^{\varepsilon }(j)\leq d_{m}$, we have 
\begin{equation*}
|J_{T}-E(J_{T})-J_{T}^{\diamond }-E(J_{T}^{\diamond })|_{2}^{2}\leq
O(1)\sum_{k=-\infty }^{T}|\mathscr{P}_{k}(J_{T}-J_{T}^{\diamond
})|_{2}^{2}=O(1)T\left( \sum_{t=1}^{T}K^{2}(\frac{t}{Th})\right) d_{m}^{2}.
\end{equation*}%
Similarly, we have $|\widetilde{J}_{T}-E(\widetilde{J}_{T})-J_{T}^{\diamond
}-E(J_{T}^{\diamond })|_{2}^{2}=O(1)T\left( \sum_{t=1}^{T}K^{2}(\frac{t}{Th}%
)\right) d_{m}^{2}$. This completes the proof.

\medskip

\noindent (c). In what follows, let $U_{t}=E[h_{t}\,|\,\mathscr{E}_{t-1}]$
for notational simplicity. Note that we actually have 
\begin{equation*}
h_{t}=\sum_{j=0}^{m}E[\widetilde{\overline{\varepsilon }}_{j+t,a}\,|\, %
\mathscr{E}_{t}],
\end{equation*}%
as $E[\widetilde{\overline{\varepsilon }}_{j+t,a}\,|\,\mathscr{E}_{t}]=0$
for $j>m$. Also, note that $\{H_{t}\}$ is an $m$-dependent m.d.s. with
respect to $\mathscr{E}_{t}$. Note that $\widetilde{\overline{\varepsilon }}%
_{t,\mathbf{v}}=h_{t}-E[h_{t+1}\,|\, \mathscr{E}_{t}]$ and $|U_{t}|_{2}\leq
\left\vert \sum_{j=0}^{m}\widetilde{\overline{\varepsilon }}%
_{j+t,a}\right\vert _{2}=O(\sqrt{m})$ by Lemma \ref{Lemma1}. Then we have 
\begin{eqnarray*}
&&\left\vert \sum_{s=1}^{t-8m}K(\frac{t-s}{Th})(\widetilde{\overline{%
\varepsilon }}_{s,a}-H_{s})\right\vert _{2}=\left\vert \sum_{s=1}^{t-8m}K(%
\frac{t-s}{Th})(U_{s}-U_{s+1})\right\vert _{2} \notag \\
&=&\left\vert K(\frac{t-1}{Th})U_{1}-K(\frac{8m}{Th})U_{t-8m+1}+%
\sum_{s=2}^{t-8m}(K(\frac{t-s}{Th})-K(\frac{t+1-s}{Th}))U_{s}\right\vert _{2}
\notag \\
&\leq &O(1)\sqrt{m}\max_{1\leq t\leq T-1}|K(\frac{t}{Th})|+\left\vert
\sum_{s=2}^{t-8m}(K(\frac{t-s}{Th})-K(\frac{t+1-s}{Th}))U_{s}\right\vert
_{2}.
\end{eqnarray*}%
In addition, noting that $U_{t}=\sum_{l=1}^{m}\mathscr{P}_{t-l}[U_{t}]$ and $%
\{\mathscr{P}_{t-l}[U_{t}]\}_{t}$ is an m.d.s., we have 
\begin{equation*}
\left\vert \sum_{s=2}^{t-8m}(K(\frac{t-s}{Th})-K(\frac{t+1-s}{Th}))%
\mathscr{P}_{s-l}[U_{s}]\right\vert _{2}^{2}=\sum_{s=2}^{t-8m}(K(\frac{t-s}{%
Th})-K(\frac{t+1-s}{Th}))^{2}|\mathscr{P}_{s-l}[U_{s}]|_{2}^{2},
\end{equation*}%
which further yields that 
\begin{eqnarray*}
&&\left\vert \sum_{s=2}^{t-8m}(K(\frac{t-s}{Th})-K(\frac{t+1-s}{Th}%
))U_{s}\right\vert _{2} \notag \\
&\leq &\sum_{l=1}^{m}\left\vert \sum_{s=2}^{t-8m}(K(\frac{t-s}{Th})-K(\frac{%
t+1-s}{Th}))\mathscr{P}_{s-l}[U_{s}]\right\vert _{2} \notag \\
&\leq &\sum_{l=1}^{m}|\mathscr{P}_{0}[U_{l}]|_{2}\left( \sum_{s=2}^{t-8m}(K(%
\frac{t-s}{Th})-K(\frac{t+1-s}{Th}))^{2}\right) ^{1/2} \notag \\
&=&O(m)\left( \sum_{s=2}^{t-8m}(K(\frac{t-s}{Th})-K(\frac{t+1-s}{Th}%
))^{2}\right) ^{1/2}.
\end{eqnarray*}%
Hence, we have 
\begin{eqnarray*}
&&\left\vert \sum_{s=1}^{t-8m}K(\frac{t-s}{Th})(\widetilde{\overline{%
\varepsilon }}_{s,a}-H_{s})\right\vert _{2} \notag \\
&=&O(\sqrt{m})\left( \max_{1\leq t\leq T-1}K^{2}(\frac{t}{Th}%
)+m\sum_{t=1}^{T-1}\left( K(\frac{t}{Th})-K(\frac{t-1}{Th})\right)
^{2}\right) ^{1/2}.
\end{eqnarray*}%
Similarly, we have 
\begin{eqnarray*}
&&\left\vert \sum_{t=s+8m}^{T}K(\frac{t-s}{Th})(\widetilde{\overline{%
\varepsilon }}_{t,a}-H_{t})\right\vert _{2} \notag \\
&=&O(\sqrt{m})\left( \max_{1\leq t\leq T-1}K^{2}(\frac{t}{Th}%
)+m\sum_{t=1}^{T-1}\left( K(\frac{t}{Th})-K(\frac{t-1}{Th})\right)
^{2}\right) ^{1/2}.
\end{eqnarray*}

Let $W_{1,t}=\widetilde{\overline{\varepsilon }}_{t,a}\sum_{s=1}^{t-8m}K(%
\frac{t-s}{Th})(\widetilde{\overline{\varepsilon }}_{s,a}-H_{s})^{\top }$.
Then $W_{1,t},W_{1,t+4m},W_{1,t+8m}\ldots $ are martingale difference
sequences. By the above developments, we have 
\begin{equation*}
|W_{1,t}|_{2}=O(\sqrt{m})\left( \max_{1\leq t\leq T-1}K^{2}(\frac{t}{Th}%
)+m\sum_{t=1}^{T-1}\left( K(\frac{t}{Th})-K(\frac{t-1}{Th})\right)
^{2}\right) ^{1/2}
\end{equation*}%
and by Lemma \ref{Lemma1} 
\begin{eqnarray*}
\left\vert \sum_{t=1}^{T}W_{1,t}\right\vert _{2} &\leq
&\sum_{i=1}^{4m-1}\left\Vert \sum_{l=0}^{\lfloor (T-i)/(4m)\rfloor
}W_{1,t+4ml}\right\Vert _{2} \notag \\
&=&O(m\sqrt{T})\left( \max_{1\leq t\leq T-1}K^{2}(\frac{t}{Th}%
)+m\sum_{t=1}^{T-1}\left( K(\frac{t}{Th})-K(\frac{t-1}{Th})\right)
^{2}\right) ^{1/2}.
\end{eqnarray*}

Let $W_{t}=W_{1,t}+W_{2,t}$, where $W_{2,t}=\widetilde{\overline{\varepsilon 
}}_{t,\mathbf{v}}\sum_{s=t-8m+1}^{t-1}K(\frac{t-s}{Th})(\widetilde{\overline{%
\varepsilon }}_{s,a}-H_{s})^{\top }$ are $12m$-dependent. Similarly, we have 
\begin{equation*}
|W_{2,t}|_{2}=O(\sqrt{m})\left( \max_{1\leq t\leq T-1}K^{2}(\frac{t}{Th}%
)+m\sum_{t=1}^{T-1}\left( K(\frac{t}{Th})-K(\frac{t-1}{Th})\right)
^{2}\right) ^{1/2}
\end{equation*}%
and 
\begin{equation*}
\left\vert \sum_{t=1}^{T}W_{2,t}-E(W_{2,t})\right\vert _{2}=O(m\sqrt{T}%
)\left( \max_{1\leq t\leq T-1}K^{2}(\frac{t}{Th})+m\sum_{t=1}^{T-1}\left( K(%
\frac{t}{Th})-K(\frac{t-1}{Th})\right) ^{2}\right) ^{1/2}.
\end{equation*}%
It follows that 
\begin{equation*}
\left\vert \sum_{t=1}^{T}W_{t}-E(W_{t})\right\vert _{2}=O(m\sqrt{T})\left(
\max_{1\leq t\leq T-1}K^{2}(\frac{t}{Th})+m\sum_{t=1}^{T-1}\left( K(\frac{t}{%
Th})-K(\frac{t-1}{Th})\right) ^{2}\right) ^{1/2}.
\end{equation*}%
The proof is now completed. 
\end{proof}
}

{\small 
\begin{proof} [Proof of Lemma \ref{Lemma5}]

We will prove this lemma by using the martingale central limit theorem
(CLT); see, e.g., \cite{hall1980martingale}.

Note that $H_{t}$ and $H_{s}$ are mutually independent if $|t-s|\geq
m+1$. Also, since $\{H_{t}\sum_{s=t-m+1}^{t-1}H_{s}w_{s,t}\}_{t}$ is
an m.d.s., by Burkholder inequality and Minkowski inequality, if $%
m/(Th)\rightarrow 0$, we have 
\begin{eqnarray*}
\left\vert \sum_{t=2}^{T}H_{t}\sum_{s=t-m+1}^{t-1}H_{s}w_{s,t}\right\vert
_{2}^{2} &\leq &O(1)\sum_{t=2}^{T}\left\vert
H_{t}\sum_{s=t-m+1}^{t-1}H_{s}w_{s,t}\right\vert _{2}^{2}\leq
O(1)\sum_{t=2}^{T}\left\vert H_{t}\right\vert _{4}^{2}\left\vert
\sum_{s=t-m+1}^{t-1}H_{s}w_{s,t}\right\vert _{4}^{2} \notag \\
&\leq &O(1)\sum_{t=2}^{T}\left\vert H_{t}\right\vert
_{4}^{2}\sum_{s=t-m+1}^{t-1}\left\vert H_{s}w_{s,t}\right\vert _{4}^{2}\leq
O(1)\sum_{t=2}^{T}\sum_{s=t-m+1}^{t-1}w_{s,t}^{2} \notag \\
&=&O(Tm/(T^{2}h))=o(1).
\end{eqnarray*}%
Then we can easily verify the Lindeberg condition since 
\begin{equation*}
\sum_{t=1+m}^{T}\left\vert H_{t}\sum_{s=t-m}^{t-1}H_{s}w_{s,t}\right\vert
_{4}^{4}\leq \sum_{t=1+m}^{T}\left\vert H_{t}\right\vert _{4}^{4}\left\vert
\sum_{s=t-m}^{t-1}H_{s}w_{s,t}\right\vert _{4}^{4}\leq
\sum_{t=1+m}^{T}\left\vert H_{t}\right\vert _{4}^{4}\left(
\sum_{s=t-m}^{t-1}\left\vert H_{s}w_{s,t}\right\vert _{4}^{2}\right)
^{2}=o(1).
\end{equation*}%
We next verify the convergence of conditional variance, i.e., 
\begin{equation*}
\sum_{t=m+1}^{T}E[(H_{t}\sum_{s=1}^{t-m}H_{s}w_{s,t})^{2}\,|\, \mathscr{E}%
_{t-1}]\rightarrow
_{P}\sum_{t=m+1}^{T}E[(H_{t}\sum_{s=1}^{t-m}H_{s}w_{s,t})^{2}],
\end{equation*}%
then the result follows from the martingale CLT.

Let $W_{t-1}=\sum_{s=1}^{t-m}H_{s}w_{s,t}$. We will prove this convergence
result by showing that 
\begin{equation*}
\left\vert \sum_{t=m+1}^{T}\left[ E(H_{t}^{2}\,|\, \mathscr{E}%
_{t-1})-E(H_{t}^{2})\right] W_{t-1}^{2}\right\vert _{1}=o(1)
\end{equation*}%
and 
\begin{equation*}
\left\vert \sum_{t=m+1}^{T}E(H_{t}^{2})\left[ W_{t-1}^{2}-E(W_{t-1}^{2})%
\right] \right\vert _{1}=o(1),
\end{equation*}%
respectively.

First, we consider $\left\vert \sum_{t=m+1}^{T}\left[ E(H_{t}^{2}\,|\, %
\mathscr{E}_{t-1})-E(H_{t}^{2})\right] W_{t-1}^{2}\right\vert _{1}$. Let $%
H_{t}^{\ast }=E(H_{t}^{2}\,|\, \mathscr{E}_{t-1})-E(H_{t}^{2})$ and thus $%
E(H_{t}^{\ast })=0$. For $0\leq j\leq m-1$, let $U(j)=\sum_{t=m+1}^{T}\left( 
\mathcal{P}_{t-j}H_{t}^{\ast }\right) {W}_{t-1}^{2}$, then 
\begin{equation*}
\sum_{t=m+1}^{T}\left[ E(H_{t}^{2}\,|\, \mathscr{E}_{t-1})-E(H_{t}^{2})\right]
W_{t-1}^{2}=\sum_{t=m+1}^{T}\left( \sum_{j=0}^{m-1}\mathcal{P}%
_{t-j}H_{t}^{\ast }\right) W_{t-1}^{2}=\sum_{j=0}^{m-1}U(j).
\end{equation*}%
Note that $\left\{ \left( \mathcal{P}_{t-j}H_{t}^{\ast }\right)
W_{t-1}^{2}\right\} _{t}$ forms a martingale difference sequence since 
\begin{equation*}
E\left\{ \left( \mathcal{P}_{t-j}H_{t}^{\ast }\right) W_{t-1}^{2}|\mathscr{E}%
_{t-j-1}\right\} =\left[ E(H_{t}^{\ast }|\mathscr{E}_{t-j-1})-E(H_{t}^{\ast
}|\mathscr{E}_{t-j-1})\right] W_{t-1}^{2}=0.
\end{equation*}%
By using Burkholder inequality and the fact that $|\mathcal{P}%
_{t-j}H_{t}^{\ast }|_{2}\leq 2|H_{t}^{\ast }|_{2}<\infty $, 
\begin{eqnarray*}
|U(j)|_{2}^{2} &\leq &O(1)\sum_{t=m+1}^{T}|\left( \mathcal{P}%
_{t-j}H_{t}^{\ast }\right) W_{t-1}^{2}|_{2}^{2}\leq
O(1)\sum_{t=m+1}^{T}|W_{t-1}|_{4}^{4} \notag \\
&\leq &O(1)\sum_{t=m+1}^{T}(\sum_{s=1}^{t-m}w_{s,t}^{2})^{2}=O\left(
T^{-1}\right) .
\end{eqnarray*}%
It follows that 
\begin{equation*}
\left\vert \sum_{t=m+1}^{T}\left[ E(H_{t}^{2}\,|\, \mathscr{E}%
_{t-1})-E(H_{t}^{2})\right] W_{t-1}^{2}\right\vert _{1}=o(1).
\end{equation*}

Now, we consider $\left\vert \sum_{t=2}^{T}E({H}_{t}^{2})\left[
W_{t-1}^{2}-E(W_{t-1}^{2})\right] \right\vert _{1}$. For notational
simplicity, we omit the constant $E(H_{t}^{2})$ and write 
\begin{eqnarray*}
\sum_{t=m+1}^{T}[W_{t-1}^{2}-E\left( W_{t-1}^{2}\right) ]
&=&\sum_{t=m+1}^{T}\sum_{s=1}^{t-m}\left( H_{s}^{2}-E\left( H_{s}^{2}\right)
\right)
w_{s,t}^{2}+2\sum_{t=m+2}^{T}\sum_{s_{1}=2}^{t-m}\sum_{s_{2}=1}^{s_{1}-1}
\left[ H_{s_{1}}H_{s_{2}}w_{s_{1},t}w_{s_{2},t}\right] \notag  \\
&\eqqcolon &I_{3}+2I_{4}.
\end{eqnarray*}%
Since $T\sum_{t=1}^{T}w_{s,t}^{2}=O(1)$, as in the proof of $U_{j}$, by
using martingale decomposition and Burkholder inequality, we have 
\begin{eqnarray*}
I_{3} &=&\sum_{t=m+1}^{T}\sum_{s=1}^{t-m}\left(
H_{s}^{2}-E(H_{s}^{2})\right) w_{s,t}^{2} \notag \\
&=&\frac{1}{T}\sum_{s=1}^{T-m}\left( H_{s}^{2}-E(H_{s}^{2})\right) \left(
T\sum_{t=s+m}^{T}w_{s,t}^{2}\right) =O_{P}(1/\sqrt{T}).
\end{eqnarray*}%
For $I_{4},$ by the CS inequality and Burkholder inequality, we have 
\begin{eqnarray*}
|I_{4}|_{2}^{2} &\leq &O(1)\sum_{s_{1}=2}^{T-m}\left\vert
H_{s_{1}}\sum_{s_{2}=1}^{s_{1}-1}H_{s_{2}}%
\sum_{t=s_{1}+m}^{T}w_{s_{1},t}w_{s_{2},t}\right\vert _{2}^{2} \notag \\
&\leq &O(1)\sum_{s_{1}=2}^{T-m}\left\vert H_{s_{1}}\right\vert
_{4}^{2}\left\vert
\sum_{s_{2}=1}^{s_{1}-1}H_{s_{2}}\sum_{t=s_{1}+m}^{T}w_{s_{1},t}w_{s_{2},t}%
\right\vert _{4}^{2} \notag \\
&\leq &O(1)\sum_{s_{1}=2}^{T-m}\left\vert H_{s_{1}}\right\vert
_{4}^{2}\sum_{s_{2}=1}^{s_{1}-1}\left\vert
H_{s_{2}}\sum_{t=s_{1}+m}^{T}w_{s_{1},t}w_{s_{2},t}\right\vert _{4}^{2} \notag \\
&=&O\left( \sum_{s_{1}=2}^{T-m}\sum_{s_{2}=1}^{s_{1}-1}\left(
\sum_{t=s_{1}+m}^{T}w_{s_{1},t}w_{s_{2},t}\right) ^{2}\right) =O(1/T).
\end{eqnarray*}%
Combining the above results, the proof is now completed. 
\end{proof}
}

{\small 
\begin{proof} [Proof of Lemma \ref{Lemma6}]

We shall prove this CLT for quadratic forms of panel data by using $m$%
-dependence approximation, blocking arguments and the CLT for martingales.

First note that $Q_T = \frac{2}{T}J_T + \frac{1}{T}\sum_{t=1}^{T}\overline{%
\varepsilon}_{t,a}^2$. By using Lemma \ref{Lemma1} (a), we have $\frac{1}{T}%
\sum_{t=1}^{T}[\overline{\varepsilon}_{t,a}^2-E(\overline{\varepsilon}%
_{t,a}^2)] = O_P(1/\sqrt{T})$. In addition, by the $m$-dependence
approximation result in Lemma \ref{Lemma4} (b), we have $\frac{2}{T}%
|J_T-E(J_T) - \widetilde{J}_T + E(\widetilde{J}_T)|_{2} = O(\sqrt{h}d_m)$.
Hence $\sqrt{1/h}\frac{2}{T}|J_T-E(J_T) - \widetilde{J}_T + E(\widetilde{J}%
_T)|_{2} \to 0$ as $d_m \to 0$ when $m\to \infty$. In addition, by Lemma \ref%
{Lemma4} (c), $\widetilde{J}_T - E(\widetilde{J}_T)$ can be approximated by $%
M_T$ as long as $m/(Th) \to 0$.

From the above analyses, we know that $\sqrt{1/h}(Q_{T}-E(Q_{T}))$ can be
approximated by $\frac{2}{T}\sqrt{1/h}M_{T}$. Note that $H_{t}$ is a
sequence of martingale differences, we then prove this result by using Lemma %
\ref{Lemma5}. Let $w_{s,5}=\frac{2}{T}\sqrt{1/h}K\left( \frac{t-s}{Th}%
\right) $, by using Lemma \ref{Lemma5}, we have 
\begin{equation*}
\frac{2}{T}\sqrt{1/h}M_{T}\rightarrow _{D}N(0,2\nu _{0}\sigma _{\varepsilon
,a}^{4})
\end{equation*}%
provided that 
\begin{equation*}
\frac{4}{T^{2}h}\sum_{t=2}^{T}\sum_{s=1}^{t-1}K^{2}\left( \frac{t-s}{Th}%
\right) \rightarrow 2\nu _{0}.
\end{equation*}%
The proof is now completed. 
\end{proof}
}

{\small 
\begin{proof} [Proof of Lemma \ref{Lemma7}]

The proof is given in Theorem 8.1.10 of \cite{GL2013} and is thus omitted. 
\end{proof}
}

{\small 
\begin{proof}[Proof of Lemma \ref{Lemma8}] 

\noindent (a). Let $\mathbf{V}=[\mathbf{v}_{1},\ldots ,\mathbf{v}_{R}]$, $%
\mathbf{v}_{j}=[\mathbf{v}_{j,1}^{\top },\ldots ,\mathbf{v}_{j,T}^{\top
}]^{\top }$ for $1\leq j\leq R$, $\mathbf{U}=[\mathbf{u}_{1},\ldots ,\mathbf{%
u}_{R}]$ and $\mathbf{u}_{j}=[u_{j,1},\ldots ,u_{j,N}]$ for $1\leq j\leq R$.
Note that $\mathbb{F}\mathbf{V}^{(-J)}$ is a $T\times (R-J)$ matrix with the 
$t^{th}$ row being $[\mathbf{f}_{t}^{\top }\mathbf{v}_{J+1,t},\mathbf{f}%
_{t}^{\top }\mathbf{v}_{R,t}]$. Write 
\begin{equation*}
\Vert \frac{1}{TN}\sum_{i=1}^{N}\mathcal{E}_{i}^{\dagger }\Vert \leq \Vert 
\frac{1}{TN}\sum_{i=1}^{N}\mathcal{E}_{i}\Vert +\left( \frac{1}{T^{2}N^{2}}%
\sum_{t=1}^{T}(\sum_{j=J+1}^{R}s_{TN,j}\sum_{i=1}^{N}\mathbf{f}_{t}^{\top }%
\mathbf{v}_{j,t}u_{j,i})^{2}\right) ^{1/2}.
\end{equation*}%
The first term is $O_{P}(1/\sqrt{TN})$ by Lemma \ref{Lemma2}(g). For the
second term, by the CS inequality and the fact that $%
\sum_{i=1}^{N}u_{j,i}^{2}=1$, $\frac{1}{TN}\sum_{j=J+1}^{T}s_{TN,j}^{2}=o(1)$
and $\sum_{t=1}^{T}\Vert \mathbf{v}_{j,t}\Vert ^{2}=1$, we have 
\begin{eqnarray*}
&&\frac{1}{T^{2}N^{2}}\sum_{t=1}^{T}(\sum_{j=J+1}^{R}s_{TN,j}\sum_{i=1}^{N}%
\mathbf{f}_{t}^{\top }\mathbf{v}_{j,t}u_{j,i})^{2} \notag \\
&\leq &\frac{1}{T^{2}N^{2}}\sum_{t=1}^{T}\sum_{j=J+1}^{R}(s_{TN,j}\mathbf{f}%
_{t}^{\top }\mathbf{v}_{j,t})^{2}(\sum_{i=1}^{N}u_{j,i})^{2} \notag \\
&\leq &\frac{1}{T^{2}N}\sum_{t=1}^{T}\sum_{j=J+1}^{R}(s_{TN,j}\mathbf{f}%
_{t}^{\top }\mathbf{v}_{j,t})^{2}=o_{P}(1/T).
\end{eqnarray*}%
Then we have $\Vert \frac{1}{TN}\sum_{i=1}^{N}\mathcal{E}_{i}^{\dagger
}\Vert =o_{P}(1/\sqrt{T})$.

\medskip

\noindent (b). Noting that $\widehat{\mathbf{F}}\widehat{\mathbf{V}}=\mathbf{%
X}\mathbf{X}^{\top }\widehat{\mathbf{F}}$ and $\mathbf{X}=\mathcal{F}%
\pmb{\Theta}^{\top }+\mathcal{E}^{\dagger }$, we have 
\begin{equation*}
\widehat{\mathbf{F}}-\mathcal{F}\left( \frac{\pmb{\Theta}^{\top }\pmb{\Theta}%
}{N}\right) \left( \frac{\mathcal{F}^{\top }\widehat{\mathbf{F}}}{T}\right)
\left( \frac{1}{TN}\widehat{\mathbf{V}}\right) ^{-1}=\frac{1}{TN}\left( 
\mathcal{F}\pmb{\Theta}^{\top }\mathcal{E}^{\dagger ,\top }+\mathcal{E}%
^{\dagger }\pmb{\Theta}\mathcal{F}^{\top }+\mathcal{E}^{\dagger }\mathcal{E}%
^{\dagger ,\top }\right) \widehat{\mathbf{F}}\left( \frac{1}{TN}\widehat{%
\mathbf{V}}\right) ^{-1}.
\end{equation*}%
For the first and second term on the r.h.s., noting that $\mathbf{U}%
^{(-J),\top }\mathbf{U}=\mathbf{0}$ and 
\begin{equation*}
E\Vert \mathcal{E}\pmb{\Theta}\Vert ^{2}=E\left\{ \mathrm{tr}\left(
\sum_{t=1}^{T}\pmb{\Theta}^{\top }\pmb{\varepsilon}_{t}\pmb{\varepsilon}%
_{t}^{\top }\pmb{\Theta}\right) \right\} \leq TN\Vert \pmb{\Sigma}%
_{\varepsilon }\Vert _{2}\mathrm{tr}\left( \pmb{\Theta}^{\top }\pmb{\Theta}%
/N\right) =O(TN),
\end{equation*}%
we have 
\begin{equation*}
\Vert \frac{1}{TN}\mathcal{E}^{\dagger }\pmb{\Theta}\Vert =\Vert \frac{1}{TN}%
\mathcal{E}\pmb{\Theta}\Vert =O_{P}(1/\sqrt{TN}).
\end{equation*}%
For the third term, using the facts that $\Vert \frac{1}{TN}\mathcal{E}%
\mathcal{E}^{\top }\Vert =O_{P}(1/\sqrt{T\wedge N})$ and 
\begin{equation*}
E\Vert \frac{1}{TN}\mathbb{F}\mathbb{A}^{(-J),\top }\mathbb{A}^{(-J)}\mathbb{%
F}^{\top }\Vert \leq \frac{1}{TN}\sum_{j=J+1}^{R}s_{TN,j}^{2}\sum_{t=1}^{T}%
\Vert \mathbf{v}_{j,t}\Vert ^{2}E\Vert \mathbf{f}_{t}\Vert ^{2}=o(1),
\end{equation*}%
we have $\Vert \frac{1}{TN}\mathcal{E}^{\dagger }\mathcal{E}^{\dagger ,\top
}\Vert =o_{P}(1)$. It follows that $\frac{1}{\sqrt{T}}\Vert \widehat{%
\mathbf{F}}-\mathcal{F}\mathcal{H}\Vert =o_{P}(1)$, where $\mathcal{H}%
= ( \frac{\pmb{\Theta}^{\top }\pmb{\Theta}}{N} ) ( \frac{%
\mathcal{F}^{\top }\widehat{\mathbf{F}}}{T} ) ( \frac{1}{TN}%
\widehat{\mathbf{V}} ) ^{-1}$.

\medskip

\noindent (c). Note that $\frac{1}{TN}\widehat{\mathbf{V}}=\frac{1}{T^{2}N}%
\widehat{\mathbf{F}}^{\top }\mathbf{X}\mathbf{X}^{\top }\widehat{\mathbf{F}}$
and $\mathbf{X}=\mathbb{F}\mathbb{A}^{(J),\top }+\mathbb{F}\mathbb{A}%
^{(-J),\top }+\mathcal{E}$. Then we have 
\begin{eqnarray*}
\frac{1}{TN}\widehat{\mathbf{V}} &=&\frac{1}{T^{2}N}\widehat{\mathbf{F}}%
^{\top }\mathbb{F}\mathbb{A}^{(J),\top }\mathbb{A}^{(J)}\mathbb{F}^{\top }%
\widehat{\mathbf{F}}+\frac{1}{T^{2}N}\widehat{\mathbf{F}}^{\top }\mathbb{F}%
\mathbb{A}^{(-J),\top }\mathbb{A}^{(-J)}\mathbb{F}^{\top }\widehat{\mathbf{F}%
}+\frac{1}{T^{2}N}\widehat{\mathbf{F}}^{\top }\mathcal{E}\mathcal{E}^{\top }%
\widehat{\mathbf{F}} \notag \\
&&+\text{interaction terms} \notag \\
&\eqqcolon & I_{5}+I_{6}+I_{7}+\text{interaction terms}.
\end{eqnarray*}%
We consider $I_{6}$ first. Let $\mathbf{V}=[\mathbf{v}_{1},\ldots ,\mathbf{v}%
_{R}]$ and $\mathbf{v}_{j}=[\mathbf{v}_{j,1}^{\top },\ldots ,\mathbf{v}%
_{j,T}^{\top }]^{\top }$ for $1\leq j\leq R$. Note that $\mathbb{F}\mathbf{V}%
^{(-J)}$ is a $T\times (R-J)$ matrix with the $t^{th}$ row being $[\mathbf{f}%
_{t}^{\top }\mathbf{v}_{J+1,t},\mathbf{f}_{t}^{\top }\mathbf{v}_{R,t}]$.
Then we have 
\begin{equation*}
E\Vert I_{6}\Vert \leq \widetilde{r}(TN)^{-1}E\Vert \mathbb{F}\mathbb{A}%
^{(-J),\top }\mathbb{A}^{(-J)}\mathbb{F}^{\top }\Vert \leq \frac{\widetilde{r%
}}{TN}\sum_{j=J+1}^{R}s_{TN,j}^{2}\sum_{t=1}^{T}\Vert \mathbf{v}_{j,t}\Vert
^{2}E\Vert \mathbf{f}_{t}\Vert ^{2}=o(1).
\end{equation*}%
As in the proof of part (a), we have $I_{7}=O_{P}(\frac{1}{\sqrt{T\wedge N}})
$. By the CS inequality, we can show that the cross product (interaction)
terms are all $o_{P}(1)$.

Let $\mathbf{R}_{TN}=(\pmb{\Theta}^{\top }\pmb{\Theta}/N)^{1/2}\mathcal{F}%
^{\top }\widehat{\mathbf{F}}/T$ and $\pmb{\Upsilon}_{TN}=\mathbf{R}_{TN}(%
\widehat{\mathbf{F}}^{\top }\mathcal{F}/T(\pmb{\Theta}^{\top }\pmb{\Theta}/N)%
\mathcal{F}^{\top }\widehat{\mathbf{F}}/T)^{-1/2}$. Thus $\pmb{\Upsilon}%
_{TN}^{\top }\pmb{\Upsilon}_{TN}=\mathbf{I}_{\widetilde{r}}$. In addition,
note that 
\begin{equation*}
\left( \mathbf{B}_{TN}+\mathbf{d}_{TN}(\mathbf{R}_{TN}^{\top }\mathbf{R}%
_{TN})^{-1}\mathbf{R}_{TN}^{\top }\right) \mathbf{R}_{TN}=\mathbf{R}_{TN}%
\frac{1}{TN}\widehat{\mathbf{V}}\quad \text{and}\quad \mathbf{B}_{TN}=%
\pmb{\Delta}+o_{P}(1),
\end{equation*}%
where $\mathbf{B}_{TN}=(\pmb{\Theta}^{\top }\pmb{\Theta}/N)^{1/2}(\mathcal{F}%
^{\top }\mathcal{F}/T)(\pmb{\Theta}^{\top }\pmb{\Theta}/N)^{1/2}$ and $%
\mathbf{d}_{TN}=o_{P}(1)$. Hence, each column of $\mathbf{R}_{TN}$ is
non-standardized eigenvector of the matrix $\mathbf{B}_{TN}+\mathbf{d}_{TN}(%
\mathbf{R}_{TN}^{\top }\mathbf{R}_{TN})^{-1}\mathbf{R}_{TN}^{\top }$. Then,
part (c) follows by using eigenvalues perturbation theory (e.g., p.203 in %
\citealp{stewart1990matrix}).

\medskip

\noindent (d). Part (d) follows from the proof of part (c) and the
eigenvector perturbation theory (see, e.g., Ch. V in %
\citealp{stewart1990matrix}). 
\end{proof}
}

{\small 
\begin{proof}[Proof of Lemma \ref{Lemma9}]

\noindent (a). Let $\mathbf{G}=\{\mathbf{g}_{it}^{\top }\mathbf{f}%
_{t}\}_{T\times N}$ be a $T\times N$ matrix. Then we have 
\begin{equation*}
\Vert \frac{1}{TN}\mathcal{E}^{\ast }\mathcal{E}^{\ast ,\top }\Vert \leq
\Vert \frac{1}{TN}\mathcal{E}\mathcal{E}^{\top }\Vert +a_{TN}^{2}\Vert \frac{%
1}{TN}\mathbf{G}\mathbf{G}^{\top }\Vert +2a_{TN}\Vert \frac{1}{TN}\mathcal{E}%
\mathbf{G}^{\top }\Vert .
\end{equation*}%
Part (a) follows from Lemma \ref{Lemma2}(b) and the facts that $\Vert \mathbf{G}\Vert =O_{P}(\sqrt{TN})$ and $a_{TN}=(TN)^{-1/2}h^{-1/4}=o(1/%
\sqrt{T\wedge N})$.

\medskip

\noindent (b). Note that $\Vert \frac{1}{TN}\pmb{\Lambda}^{\top }\mathcal{E}%
^{\ast ,\top }\mathbf{F}\Vert \leq \Vert \frac{1}{TN}\pmb{\Lambda}^{\top }%
\mathcal{E}^{\top }\mathbf{F}\Vert +a_{TN}\Vert \frac{1}{TN}\pmb{\Lambda}%
^{\top }\mathbf{G}^{\top }\mathbf{F}\Vert $. Part (b) follows directly from
Lemma \ref{Lemma2}(c) provided that $a_{TN}\Vert \frac{1}{TN}\pmb{\Lambda}%
^{\top }\mathbf{G}^{\top }\mathbf{F}\Vert =o_{P}(1/\sqrt{TN})$. By using
similar arguments as used in the proof of Lemma \ref{Lemma1}(a) and the
normalization condition that $\frac{1}{T}\sum_{t=1}^{T}\mathbf{g}_{it}=0$,
we have 
\begin{eqnarray*}
\Vert \frac{1}{TN}\pmb{\Lambda}^{\top }\mathbf{G}^{\top }\mathbf{F}\Vert 
&\leq &\left\Vert \frac{1}{TN}\sum_{i=1}^{N}\pmb{\lambda}_{i}\sum_{t=1}^{T}%
\mathbf{g}_{it}^{\top }E(\mathbf{f}_{t}\mathbf{f}_{t}^{\top })\right\Vert
+\left\Vert \frac{1}{TN}\sum_{i=1}^{N}\pmb{\lambda}_{i}\sum_{t=1}^{T}\mathbf{%
g}_{it}^{\top }(\mathbf{f}_{t}\mathbf{f}_{t}^{\top }-E(\mathbf{f}_{t}\mathbf{%
f}_{t}^{\top }))\right\Vert  \notag \\
&=&\left\Vert \frac{1}{TN}\sum_{t=1}^{T}\sum_{i=1}^{N}\pmb{\lambda}_{i}%
\mathbf{g}_{it}^{\top }E(\mathbf{f}_{t}\mathbf{f}_{t}^{\top })\right\Vert
+O_{P}(1/\sqrt{T})=O_{P}(1/\sqrt{T}).
\end{eqnarray*}%
This proves part (b). \medskip 

\noindent (c). Note that $\|\frac{1}{TN}\mathcal{E}^* \pmb{\Lambda} \|\leq \|
\frac{1}{TN}\mathcal{E} \pmb{\Lambda} \| + a_{TN}\|\frac{1}{TN}\mathbf{G}
\pmb{\Lambda} \|=O_P( \frac{1}{\sqrt{T N}}) + O_P( \frac{1}{\sqrt{T}}a_{TN})
= O_P( \frac{1}{\sqrt{T N}})$ since $a_{TN} = (TN)^{-1/2}h^{-1/4}=o(1/\sqrt{
T\wedge N})$.

\medskip

\noindent (d). Note that 
\begin{equation*}
\Vert \frac{1}{TN}\mathcal{E}^{\ast }\mathcal{E}^{\ast ,\top }\mathbf{F}
\Vert \leq \Vert \frac{1}{TN}\mathcal{E}\mathcal{E}^{\top }\mathbf{F}\Vert
+a_{TN}^{2}\Vert \frac{1}{TN}\mathbf{G}\mathbf{G}^{\top }\mathbf{F}\Vert
+a_{TN}\Vert \frac{1}{TN}\mathcal{E}\mathbf{G}^{\top }\mathbf{F}
\Vert +a_{TN}\Vert \frac{1}{TN}\mathbf{G}\mathcal{E}^{\top }\mathbf{F}
\Vert .
\end{equation*}%
Then part (d) follows from Lemma \ref{Lemma2}(e) and the fact that $a_{TN}=(TN)^{-1/2}h^{-1/4}=o(1/\sqrt{T\wedge N})$.

\medskip

\noindent (e). By Lemma \ref{Lemma2} (f) and similar arguments as used in
the the proof of part (b), we have 
\begin{eqnarray*}
\Vert \frac{1}{TN}\sum_{i=1}^{N}(1-a_{i})\mathcal{E}_{i}^{\ast ,\top }%
\mathbf{F}\Vert  &\leq &\Vert \frac{1}{TN}\sum_{i=1}^{N}(1-a_{i})\mathcal{E}%
_{i}^{\top }\mathbf{F}\Vert +a_{TN}\Vert \frac{1}{TN}\sum_{i=1}^{N}(1-a_{i})%
\sum_{t=1}^{T}\mathbf{g}_{it}^{\top }\mathbf{f}_{t}\mathbf{f}_{t}^{\top
}\Vert  \notag \\
&=&O_{P}(1/\sqrt{TN})+a_{TN}O_{P}(1/\sqrt{T})=O_{P}(1/\sqrt{TN}).
\end{eqnarray*}

\medskip

\noindent (f). Write%
\begin{eqnarray*}
\widehat{\mathbf{F}}\cdot \frac{1}{TN}\widehat{\mathbf{V}} &=&\frac{1}{TN}%
\mathbf{X}\mathbf{X}^{\top }\widehat{\mathbf{F}}=\frac{1}{TN}(\mathbf{F}%
\pmb{\Lambda}^{\top }+\mathcal{E}^{\ast })(\mathbf{F}\pmb{\Lambda}^{\top }+%
\mathcal{E})^{\ast ,\top }\widehat{\mathbf{F}} \notag \\
&=&\frac{1}{TN}\mathbf{F}\pmb{\Lambda}^{\top }\pmb{\Lambda}\mathbf{F}^{\top }%
\widehat{\mathbf{F}}+\frac{1}{TN}\mathbf{F}\pmb{\Lambda}^{\top }\mathcal{E}%
^{\ast ,\top }\widehat{\mathbf{F}}+\frac{1}{TN}\mathcal{E}^{\ast }%
\pmb{\Lambda}\mathbf{F}^{\top }\widehat{\mathbf{F}}+\frac{1}{TN}\mathcal{E}%
^{\ast }\mathcal{E}^{\ast ,\top }\widehat{\mathbf{F}}.
\end{eqnarray*}%
Note that 
\begin{equation*}
\frac{1}{\sqrt{T}}\Vert \frac{1}{TN}\mathbf{F}\pmb{\Lambda}^{\top }\mathcal{E%
}^{\ast ,\top }\widehat{\mathbf{F}}\Vert =O_{P}\left( \frac{1}{\sqrt{N}}%
\right) \ \ \text{ and }\ \ \frac{1}{\sqrt{T}}\Vert \frac{1}{TN}\mathcal{E}%
^{\ast }\mathcal{E}^{\ast ,\top }\widehat{\mathbf{F}}\Vert =O_{P}\left( 
\frac{1}{\sqrt{T\wedge N}}\right) 
\end{equation*}%
by Lemma \ref{Lemma9}(c) and \ref{Lemma9}(a) respectively. Then part (f)
follows.

\medskip

\noindent (g). Write 
\begin{equation*}
\frac{1}{T}\mathbf{F}^{\top }(\widehat{\mathbf{F}}-\mathbf{F}\mathbf{H})%
\frac{1}{TN}\widehat{\mathbf{V}}=\frac{1}{T^{2}N}\mathbf{F}^{\top }\mathbf{F}%
\pmb{\Lambda}^{\top }\mathcal{E}^{\ast ,\top }\widehat{\mathbf{F}}+\frac{1}{%
T^{2}N}\mathbf{F}^{\top }\mathcal{E}^{\ast }\pmb{\Lambda}\mathbf{F}^{\top }%
\widehat{\mathbf{F}}+\frac{1}{T^{2}N}\mathbf{F}^{\top }\mathcal{E}^{\ast }%
\mathcal{E}^{\ast ,\top }\widehat{\mathbf{F}}.
\end{equation*}%
For the first term on the r.h.s., by Lemmas \ref{Lemma9}(b) and (f), we have 
\begin{eqnarray*}
\frac{1}{T^{2}N}\mathbf{F}^{\top }\mathbf{F}\pmb{\Lambda}^{\top }\mathcal{E}%
^{\ast ,\top }\widehat{\mathbf{F}} &=&\frac{1}{T^{2}N}\mathbf{F}^{\top }%
\mathbf{F}\pmb{\Lambda}^{\top }\mathcal{E}^{\ast ,\top }\mathbf{F}\mathbf{H}+%
\frac{1}{T^{2}N}\mathbf{F}^{\top }\mathbf{F}\pmb{\Lambda}^{\top }\mathcal{E}%
^{\ast ,\top }(\widehat{\mathbf{F}}-\mathbf{F}\mathbf{H}) \notag \\
&=&O_{P}(1/\sqrt{TN})+O_{P}(1/\sqrt{N(T\wedge N)}).
\end{eqnarray*}%
Similarly, by Lemmas \ref{Lemma9}(b) and (d), we have $\frac{1}{T^{2}N}%
\mathbf{F}^{\top }\mathcal{E}^{\ast }\pmb{\Lambda}\mathbf{F}^{\top }\widehat{%
\mathbf{F}}=O_{P}(1/\sqrt{TN})$ and $\frac{1}{T^{2}N}\mathbf{F}^{\top }%
\mathcal{E}^{\ast }\mathcal{E}^{\ast ,\top }\widehat{\mathbf{F}}=O_{P}(1/%
\sqrt{T(T\wedge N)})$. Then part (g) follows.

\medskip

\noindent (h). Write 
\begin{eqnarray*}
\frac{1}{TN}\mathcal{E}^{\ast }\mathcal{E}^{\ast ,\top }(\widehat{\mathbf{F}}%
-\mathbf{F}\mathbf{H})\widehat{\mathbf{V}}/(TN) &=&\frac{1}{T^{2}N^{2}}%
\mathcal{E}^{\ast }\mathcal{E}^{\ast ,\top }\mathbf{F}\pmb{\Lambda}^{\top }%
\mathcal{E}^{\ast ,\top }\widehat{\mathbf{F}}+\frac{1}{T^{2}N^{2}}\mathcal{E}%
^{\ast }\mathcal{E}^{\ast ,\top }\mathcal{E}^{\ast }\pmb{\Lambda}\mathbf{F}%
^{\top }\widehat{\mathbf{F}} \notag \\
&&+\frac{1}{T^{2}N^{2}}\mathcal{E}^{\ast }\mathcal{E}^{\ast ,\top }\mathcal{E%
}^{\ast }\mathcal{E}^{\ast ,\top }\widehat{\mathbf{F}}.
\end{eqnarray*}%
For the first term on the r.h.s., 
\begin{eqnarray*}
\frac{1}{T^{2}N^{2}}\mathcal{E}^{\ast }\mathcal{E}^{\ast ,\top }\mathbf{F}%
\pmb{\Lambda}^{\top }\mathcal{E}^{\ast ,\top }\widehat{\mathbf{F}} &=&\frac{1%
}{T^{2}N^{2}}\mathcal{E}^{\ast }\mathcal{E}^{\ast ,\top }\mathbf{F}%
\pmb{\Lambda}^{\top }\mathcal{E}^{\ast ,\top }\mathbf{F}\mathbf{H}+\frac{1}{%
T^{2}N^{2}}\mathcal{E}^{\ast }\mathcal{E}^{\ast ,\top }\mathbf{F}%
\pmb{\Lambda}^{\top }\mathcal{E}^{\ast ,\top }(\widehat{\mathbf{F}}-\mathbf{F%
}\mathbf{H}) \notag \\
&=&O_{P}\left( \frac{1}{\sqrt{TN(T\wedge N)}}\right) +O_{P}\left( \frac{%
\sqrt{T}}{\sqrt{TN}(T\wedge N)}\right) 
\end{eqnarray*}%
by Lemmas \ref{Lemma9} (b)--(d) and (f). Similarly, by Lemmas \ref{Lemma9}
(c), (d) and (f) we can show that the second term is $O_{P}\left( \sqrt{T}/%
\sqrt{N(T\wedge N)}\right) $ and the third term is $O_{P}(\sqrt{T}/(T\wedge
N)^{3/2})$. It follows that $\Vert \frac{1}{TN}\mathcal{E}^{\ast }\mathcal{E}%
^{\ast ,\top }(\widehat{\mathbf{F}}-\mathbf{F}\mathbf{H})\Vert =O_{P}(\sqrt{T%
}/(T\wedge N)^{3/2}).$

\medskip

\noindent (i). Let $\mathbf{G}_{i}$ denote the $i^{th}$ column of $\mathbf{G}
$ in which $\mathbf{G}=\{\mathbf{g}_{it}^{\top }\mathbf{f}_{t}\}_{T\times N}$%
. Write 
\begin{eqnarray*}
&&\sum_{i=1}^{N}(1-a_{i})\mathcal{E}_{i}^{\ast ,\top }(\widehat{\mathbf{F}}-%
\mathbf{F}\mathbf{H}) \notag \\
&=&\sum_{i=1}^{N}(1-a_{i})\mathcal{E}_{i}^{\ast ,\top }\left( \frac{1}{TN}%
\mathbf{F}\pmb{\Lambda}^{\top }\mathcal{E}^{\ast ,\top }\widehat{\mathbf{F}}+%
\frac{1}{TN}\mathcal{E}^{\ast }\pmb{\Lambda}\mathbf{F}^{\top }\widehat{%
\mathbf{F}}+\frac{1}{TN}\mathcal{E}^{\ast }\mathcal{E}^{\ast ,\top }\widehat{%
\mathbf{F}}\right) (\frac{1}{TN}\widehat{\mathbf{V}})^{-1}.
\end{eqnarray*}%
For the first term on the r.h.s., we apply Lemmas \ref{Lemma9} (b) and
(e)--(f) to obtain 
\begin{eqnarray*}
&&\left\Vert \frac{1}{T^{2}N^{2}}\sum_{i=1}^{N}(1-a_{i})\mathcal{E}%
_{i}^{\ast ,\top }\mathbf{F}\pmb{\Lambda}^{\top }\mathcal{E}^{\ast ,\top }%
\widehat{\mathbf{F}}\right\Vert  \notag \\
&\leq &\left\Vert \frac{1}{T^{2}N^{2}}\sum_{i=1}^{N}(1-a_{i})\mathcal{E}%
_{i}^{\ast ,\top }\mathbf{F}\pmb{\Lambda}^{\top }\mathcal{E}^{\ast ,\top }%
\mathbf{F}\mathbf{H}\right\Vert +\left\Vert \frac{1}{T^{2}N^{2}}%
\sum_{i=1}^{N}(1-a_{i})\mathcal{E}_{i}^{\ast ,\top }\mathbf{F}\pmb{\Lambda}%
^{\top }\mathcal{E}^{\ast ,\top }(\widehat{\mathbf{F}}-\mathbf{F}\mathbf{H}%
)\right\Vert \notag \\
&=&O_{P}(1/(TN))+O_{P}\left( \sqrt{T}/(TN\sqrt{(T\wedge N)})\right) .
\end{eqnarray*}%
For the second term, by Lemmas \ref{Lemma2}(g) and \ref{Lemma9}(c), we have 
\begin{eqnarray*}
&&\left\Vert \frac{1}{T^{2}N^{2}}\sum_{i=1}^{N}(1-a_{i})\mathcal{E}%
_{i}^{\ast ,\top }\mathcal{E}^{\ast }\pmb{\Lambda}\mathbf{F}^{\top }\widehat{%
\mathbf{F}}\right\Vert \notag \\
&\leq &a_{TN}\left\Vert \frac{1}{T^{2}N^{2}}\sum_{i=1}^{N}(1-a_{i})\mathbf{G}%
_{i}^{\top }\mathcal{E}\pmb{\Lambda}\mathbf{F}^{\top }\widehat{\mathbf{F}}%
\right\Vert +a_{TN}^{2}\left\Vert \frac{1}{T^{2}N^{2}}\sum_{i=1}^{N}(1-a_{i})%
\mathbf{G}_{i}^{\top }\mathbf{G}\pmb{\Lambda}\mathbf{F}^{\top }\widehat{%
\mathbf{F}}\right\Vert \notag \\
&&+\left\Vert \frac{1}{T^{2}N^{2}}\sum_{i=1}^{N}(1-a_{i})\mathcal{E}%
_{i}^{\top }\mathcal{E}^{\ast }\pmb{\Lambda}\mathbf{F}^{\top }\widehat{%
\mathbf{F}}\right\Vert  \notag \\
&=&a_{TN}\left\Vert \frac{1}{T^{2}N^{2}}\sum_{i=1}^{N}(1-a_{i})\mathbf{G}%
_{i}^{\top }\mathcal{E}\pmb{\Lambda}\mathbf{F}^{\top }\widehat{\mathbf{F}}%
\right\Vert +O_{P}(a_{TN}^{2})+O_{P}(1/N).
\end{eqnarray*}%
In addition, since $|\mathbf{G}_{i}^{\top }\mathcal{E}\pmb{\Lambda}%
|_{2}=|\sum_{j=1}^{N}\sum_{t=1}^{T}\mathbf{g}_{it}^{\top }\mathbf{f}%
_{t}\varepsilon _{jt}\pmb{\lambda}_{j}^{\top }|_{2}=O(\sqrt{TN})$ by using
Lemma \ref{Lemma1} (b), we then have $a_{TN}\left\Vert \frac{1}{T^{2}N^{2}}%
\sum_{i=1}^{N}(1-a_{i})\mathbf{G}_{i}^{\top }\mathcal{E}\pmb{\Lambda}\mathbf{%
F}^{\top }\widehat{\mathbf{F}}\right\Vert =O_{P}(a_{TN}/\sqrt{TN})$.
Similarly, for the third term, we have 
\begin{equation*}
\left\Vert \frac{1}{T^{2}N^{2}}\sum_{i=1}^{N}(1-a_{i})\mathcal{E}_{i}^{\ast
,\top }\mathcal{E}^{\ast }\mathcal{E}^{\ast ,\top }\widehat{\mathbf{F}}%
\right\Vert =O_{P}(a_{TN}/\sqrt{T\wedge N})+O_{P}(a_{TN}\sqrt{T}/(T\wedge
N)).
\end{equation*}%
This completes the proof. 

\end{proof}
}

{\small 
\begin{proof}[Proof of Lemma \ref{Lemma10}]

\noindent (a). Let $\mathbf{G}=\{\mathbf{g}_{it}^{\top }\mathbf{f}
_{t}\}_{T\times N}$ be a $T\times N$ matrix. Then we have 
\begin{equation*}
\Vert \frac{1}{TN}\mathcal{E}^{\ast }\mathcal{E}^{\ast ,\top }\Vert \leq
\Vert \frac{1}{TN}\mathcal{E}\mathcal{E}^{\top }\Vert +a_{TN}^{2}\Vert \frac{%
	1}{TN}\mathbf{G}\mathbf{G}^{\top }\Vert +2a_{TN}\Vert \frac{1}{TN}\mathcal{E}%
\mathbf{G}^{\top }\Vert .
\end{equation*}%
Part (a) follows from Lemma \ref{Lemma2}(b) and the facts that $
\Vert \mathbf{G}\Vert =O_{P}(\sqrt{TN})$ and $a_{TN}=N^{-1/2}h^{1/4}=o(1/
\sqrt{T\wedge N})$.

\medskip

\noindent (b). Note that $\Vert \frac{1}{TN}\pmb{\Lambda}^{\top }\mathcal{E}%
^{\ast ,\top }\mathbf{F}\Vert \leq \Vert \frac{1}{TN}\pmb{\Lambda}^{\top }%
\mathcal{E}^{\top }\mathbf{F}\Vert +a_{TN}\Vert \frac{1}{TN}\pmb{\Lambda}%
^{\top }\mathbf{G}^{\top }\mathbf{F}\Vert $. Part (b) follows directly from
Lemma \ref{Lemma2}(c) provided that $a_{TN}\Vert \frac{1}{TN}\pmb{\Lambda}%
^{\top }\mathbf{G}^{\top }\mathbf{F}\Vert =o_{P}(1/\sqrt{TN})$. By using
similar arguments as used in the proof of Lemma \ref{Lemma1}(a) and the
normalization condition that $\frac{1}{T}\sum_{t=1}^{T}\mathbf{g}_{it}=0$,
we have 
\begin{eqnarray*}
\Vert \frac{1}{TN}\pmb{\Lambda}^{\top }\mathbf{G}^{\top }\mathbf{F}\Vert 
&\leq &\left\Vert \frac{1}{TN}\sum_{i=1}^{N}\pmb{\lambda}_{i}\sum_{t=1}^{T}%
\mathbf{g}_{it}^{\top }E(\mathbf{f}_{t}\mathbf{f}_{t}^{\top })\right\Vert
+\left\Vert \frac{1}{TN}\sum_{i=1}^{N}\pmb{\lambda}_{i}\sum_{t=1}^{T}\mathbf{%
	g}_{it}^{\top }(\mathbf{f}_{t}\mathbf{f}_{t}^{\top }-E(\mathbf{f}_{t}\mathbf{%
	f}_{t}^{\top }))\right\Vert  \notag \\
&=&\left\Vert \frac{1}{TN}\sum_{t=1}^{T}\sum_{i=1}^{N}\pmb{\lambda}_{i}%
\mathbf{g}_{it}^{\top }E(\mathbf{f}_{t}\mathbf{f}_{t}^{\top })\right\Vert
+O_{P}(1/\sqrt{T})=O_{P}(1/\sqrt{T}).
\end{eqnarray*}%
This proves part (b) given $a_{NT} = N^{-1/2}h^{1/4}$. \medskip 

\noindent (c). Note that $\|\frac{1}{TN}\mathcal{E}^* \pmb{\Lambda} \|\leq \|%
\frac{1}{TN}\mathcal{E} \pmb{\Lambda} \| + a_{TN}\|\frac{1}{TN}\mathbf{G}
\pmb{\Lambda} \|=O_P( \frac{1}{\sqrt{T N}}) + O_P( \frac{1}{\sqrt{T}}a_{TN})
= O_P( \frac{1}{\sqrt{T N}})$ since $a_{TN} = N^{-1/2}h^{1/4}=o(1/\sqrt{N})$.

\medskip

\noindent (d). Note that 
\begin{equation*}
\Vert \frac{1}{TN}\mathcal{E}^{\ast }\mathcal{E}^{\ast ,\top }\mathbf{F}
\Vert \leq \Vert \frac{1}{TN}\mathcal{E}\mathcal{E}^{\top }\mathbf{F}\Vert
+a_{TN}^{2}\Vert \frac{1}{TN}\mathbf{G}\mathbf{G}^{\top }\mathbf{F}\Vert
+a_{TN}\Vert \frac{1}{TN}\mathcal{E}\mathbf{G}^{\top }\mathbf{F}
\Vert +a_{TN}\Vert \frac{1}{TN}\mathbf{G}\mathcal{E}^{\top }\mathbf{F}
\Vert
\end{equation*}
and
$$
\left\|\frac{1}{\sqrt{NT}}\mathbf{G}^\top\mathbf{F}\right\|= \left\|\frac{1}{\sqrt{NT}}\sum_{t=1}^{T}\mathbf{
G}_{t}(\mathbf{f}_{t}\mathbf{f}_{t}^{\top }-E(\mathbf{f}_{t}\mathbf{
f}_{t}^{\top}))\right\| = O_P(1),
$$
where $\mathbf{G}_{t} = [\mathbf{g}_{1t},\ldots,\mathbf{g}_{Nt}]^\top$.

Then part (d) follows from Lemma \ref{Lemma2}(e) and the fact that $a_{TN}=N^{-1/2}h^{1/4}=o(1/\sqrt{T\wedge N})$.

\medskip

\noindent (e). By Lemma \ref{Lemma2} (f) and similar arguments as used in
the the proof of part (b), we have 
\begin{eqnarray*}
\Vert \frac{1}{TN}\sum_{i=1}^{N}(1-a_{i})\mathcal{E}_{i}^{\ast ,\top }%
\mathbf{F}\Vert  &\leq &\Vert \frac{1}{TN}\sum_{i=1}^{N}(1-a_{i})\mathcal{E}%
_{i}^{\top }\mathbf{F}\Vert +a_{TN}\Vert \frac{1}{TN}\sum_{i=1}^{N}(1-a_{i})%
\sum_{t=1}^{T}\mathbf{g}_{it}^{\top }\mathbf{f}_{t}\mathbf{f}_{t}^{\top
}\Vert  \notag \\
&=&O_{P}(1/\sqrt{TN})+a_{TN}O_{P}(1/\sqrt{T})=O_{P}(1/\sqrt{TN}).
\end{eqnarray*}

\medskip

\noindent (f). Write%
\begin{eqnarray*}
\widehat{\mathbf{F}}\cdot \frac{1}{TN}\widehat{\mathbf{V}} &=&\frac{1}{TN}%
\mathbf{X}\mathbf{X}^{\top }\widehat{\mathbf{F}}=\frac{1}{TN}(\mathbf{F}%
\pmb{\Lambda}^{\top }+\mathcal{E}^{\ast })(\mathbf{F}\pmb{\Lambda}^{\top }+%
\mathcal{E})^{\ast ,\top }\widehat{\mathbf{F}} \notag \\
&=&\frac{1}{TN}\mathbf{F}\pmb{\Lambda}^{\top }\pmb{\Lambda}\mathbf{F}^{\top }%
\widehat{\mathbf{F}}+\frac{1}{TN}\mathbf{F}\pmb{\Lambda}^{\top }\mathcal{E}%
^{\ast ,\top }\widehat{\mathbf{F}}+\frac{1}{TN}\mathcal{E}^{\ast }%
\pmb{\Lambda}\mathbf{F}^{\top }\widehat{\mathbf{F}}+\frac{1}{TN}\mathcal{E}%
^{\ast }\mathcal{E}^{\ast ,\top }\widehat{\mathbf{F}}.
\end{eqnarray*}%
Note that 
\begin{equation*}
\frac{1}{\sqrt{T}}\Vert \frac{1}{TN}\mathbf{F}\pmb{\Lambda}^{\top }\mathcal{E%
}^{\ast ,\top }\widehat{\mathbf{F}}\Vert =O_{P}\left( \frac{1}{\sqrt{N}}%
\right) \ \ \text{ and }\ \ \frac{1}{\sqrt{T}}\Vert \frac{1}{TN}\mathcal{E}%
^{\ast }\mathcal{E}^{\ast ,\top }\widehat{\mathbf{F}}\Vert =O_{P}\left( 
\frac{1}{\sqrt{T\wedge N}}\right) 
\end{equation*}%
by Lemma \ref{Lemma10}(c) and \ref{Lemma10}(a) respectively. Then part (f)
follows.

\medskip

\noindent (g). Write 
\begin{equation*}
\frac{1}{T}\mathbf{F}^{\top }(\widehat{\mathbf{F}}-\mathbf{F}\mathbf{H})%
\frac{1}{TN}\widehat{\mathbf{V}}=\frac{1}{T^{2}N}\mathbf{F}^{\top }\mathbf{F}%
\pmb{\Lambda}^{\top }\mathcal{E}^{\ast ,\top }\widehat{\mathbf{F}}+\frac{1}{%
	T^{2}N}\mathbf{F}^{\top }\mathcal{E}^{\ast }\pmb{\Lambda}\mathbf{F}^{\top }%
\widehat{\mathbf{F}}+\frac{1}{T^{2}N}\mathbf{F}^{\top }\mathcal{E}^{\ast }%
\mathcal{E}^{\ast ,\top }\widehat{\mathbf{F}}.
\end{equation*}%
For the first term on the r.h.s., by Lemmas \ref{Lemma10}(b) and (f), we have 
\begin{eqnarray*}
\frac{1}{T^{2}N}\mathbf{F}^{\top }\mathbf{F}\pmb{\Lambda}^{\top }\mathcal{E}%
^{\ast ,\top }\widehat{\mathbf{F}} &=&\frac{1}{T^{2}N}\mathbf{F}^{\top }%
\mathbf{F}\pmb{\Lambda}^{\top }\mathcal{E}^{\ast ,\top }\mathbf{F}\mathbf{H}+%
\frac{1}{T^{2}N}\mathbf{F}^{\top }\mathbf{F}\pmb{\Lambda}^{\top }\mathcal{E}%
^{\ast ,\top }(\widehat{\mathbf{F}}-\mathbf{F}\mathbf{H}) \notag \\
&=&O_{P}(1/\sqrt{TN})+O_{P}(1/\sqrt{N(T\wedge N)}).
\end{eqnarray*}%
In the same fashion, by Lemmas \ref{Lemma10}(b) and (d), we obtain  that  

\begin{eqnarray*}
&&\frac{1}{T^{2}N}\mathbf{F}^{\top }\mathcal{E}^{\ast }\pmb{\Lambda}\mathbf{F}^{\top }\widehat{\mathbf{F}}=O_{P}(1/\sqrt{TN}),\notag \\
&&\frac{1}{T^{2}N}\mathbf{F}^{\top }
\mathcal{E}^{\ast }\mathcal{E}^{\ast ,\top }\widehat{\mathbf{F}}=O_{P}(1/\sqrt{T(T\wedge N)}).
\end{eqnarray*}
Then part (g) follows.

\medskip

\noindent (h). Write 
\begin{eqnarray*}
\frac{1}{TN}\mathcal{E}^{\ast }\mathcal{E}^{\ast ,\top }(\widehat{\mathbf{F}}%
-\mathbf{F}\mathbf{H})\widehat{\mathbf{V}}/(TN) &=&\frac{1}{T^{2}N^{2}}%
\mathcal{E}^{\ast }\mathcal{E}^{\ast ,\top }\mathbf{F}\pmb{\Lambda}^{\top }%
\mathcal{E}^{\ast ,\top }\widehat{\mathbf{F}}+\frac{1}{T^{2}N^{2}}\mathcal{E}%
^{\ast }\mathcal{E}^{\ast ,\top }\mathcal{E}^{\ast }\pmb{\Lambda}\mathbf{F}%
^{\top }\widehat{\mathbf{F}} \notag \\
&&+\frac{1}{T^{2}N^{2}}\mathcal{E}^{\ast }\mathcal{E}^{\ast ,\top }\mathcal{E%
}^{\ast }\mathcal{E}^{\ast ,\top }\widehat{\mathbf{F}}.
\end{eqnarray*}%
For the first term on the r.h.s., 
\begin{eqnarray*}
\frac{1}{T^{2}N^{2}}\mathcal{E}^{\ast }\mathcal{E}^{\ast ,\top }\mathbf{F}%
\pmb{\Lambda}^{\top }\mathcal{E}^{\ast ,\top }\widehat{\mathbf{F}} &=&\frac{1%
}{T^{2}N^{2}}\mathcal{E}^{\ast }\mathcal{E}^{\ast ,\top }\mathbf{F}%
\pmb{\Lambda}^{\top }\mathcal{E}^{\ast ,\top }\mathbf{F}\mathbf{H}+\frac{1}{%
	T^{2}N^{2}}\mathcal{E}^{\ast }\mathcal{E}^{\ast ,\top }\mathbf{F}%
\pmb{\Lambda}^{\top }\mathcal{E}^{\ast ,\top }(\widehat{\mathbf{F}}-\mathbf{F%
}\mathbf{H}) \notag \\
&=&O_{P}\left( \frac{1}{\sqrt{TN(T\wedge N)}}\right) +O_{P}\left( \frac{%
	\sqrt{T}}{\sqrt{TN}(T\wedge N)}\right) 
\end{eqnarray*}%
by Lemmas \ref{Lemma10} (b)--(d) and (f). Similarly, by Lemmas \ref{Lemma10}
(c), (d) and (f) we can show that the second term is $O_{P}\left( \sqrt{T}/%
\sqrt{N(T\wedge N)}\right) $ and the third term is $O_{P}(\sqrt{T}/(T\wedge
N)^{3/2})$. It follows that $\Vert \frac{1}{TN}\mathcal{E}^{\ast }\mathcal{E}%
^{\ast ,\top }(\widehat{\mathbf{F}}-\mathbf{F}\mathbf{H})\Vert =O_{P}(\sqrt{T%
}/(T\wedge N)^{3/2}).$

\medskip

\noindent (i). Let $\mathbf{G}_{i}$ denote the $i^{th}$ column of $\mathbf{G}
$ in which $\mathbf{G}=\{\mathbf{g}_{it}^{\top }\mathbf{f}_{t}\}_{T\times N}$%
. Write 
\begin{eqnarray*}
&&\sum_{i=1}^{N}(1-a_{i})\mathcal{E}_{i}^{\ast ,\top }(\widehat{\mathbf{F}}-%
\mathbf{F}\mathbf{H}) \notag \\
&=&\sum_{i=1}^{N}(1-a_{i})\mathcal{E}_{i}^{\ast ,\top }\left( \frac{1}{TN}%
\mathbf{F}\pmb{\Lambda}^{\top }\mathcal{E}^{\ast ,\top }\widehat{\mathbf{F}}+%
\frac{1}{TN}\mathcal{E}^{\ast }\pmb{\Lambda}\mathbf{F}^{\top }\widehat{%
	\mathbf{F}}+\frac{1}{TN}\mathcal{E}^{\ast }\mathcal{E}^{\ast ,\top }\widehat{%
	\mathbf{F}}\right) (\frac{1}{TN}\widehat{\mathbf{V}})^{-1}.
\end{eqnarray*}%
For the first term on the r.h.s., we apply Lemmas \ref{Lemma10} (b) and
(e)--(f) to obtain 
\begin{eqnarray*}
&&\left\Vert \frac{1}{T^{2}N^{2}}\sum_{i=1}^{N}(1-a_{i})\mathcal{E}%
_{i}^{\ast ,\top }\mathbf{F}\pmb{\Lambda}^{\top }\mathcal{E}^{\ast ,\top }%
\widehat{\mathbf{F}}\right\Vert  \notag \\
&\leq &\left\Vert \frac{1}{T^{2}N^{2}}\sum_{i=1}^{N}(1-a_{i})\mathcal{E}%
_{i}^{\ast ,\top }\mathbf{F}\pmb{\Lambda}^{\top }\mathcal{E}^{\ast ,\top }%
\mathbf{F}\mathbf{H}\right\Vert +\left\Vert \frac{1}{T^{2}N^{2}}%
\sum_{i=1}^{N}(1-a_{i})\mathcal{E}_{i}^{\ast ,\top }\mathbf{F}\pmb{\Lambda}%
^{\top }\mathcal{E}^{\ast ,\top }(\widehat{\mathbf{F}}-\mathbf{F}\mathbf{H}%
)\right\Vert \notag \\
&=&O_{P}(1/(TN))+O_{P}\left( \sqrt{T}/(TN\sqrt{(T\wedge N)})\right) .
\end{eqnarray*}%
For the second term, by Lemmas \ref{Lemma2}(g) and \ref{Lemma10}(c), we have 
\begin{eqnarray*}
&&\left\Vert \frac{1}{T^{2}N^{2}}\sum_{i=1}^{N}(1-a_{i})\mathcal{E}%
_{i}^{\ast ,\top }\mathcal{E}^{\ast }\pmb{\Lambda}\mathbf{F}^{\top }\widehat{%
	\mathbf{F}}\right\Vert \notag \\
&\leq &a_{TN}\left\Vert \frac{1}{T^{2}N^{2}}\sum_{i=1}^{N}(1-a_{i})\mathbf{G}%
_{i}^{\top }\mathcal{E}\pmb{\Lambda}\mathbf{F}^{\top }\widehat{\mathbf{F}}%
\right\Vert +a_{TN}^{2}\left\Vert \frac{1}{T^{2}N^{2}}\sum_{i=1}^{N}(1-a_{i})%
\mathbf{G}_{i}^{\top }\mathbf{G}\pmb{\Lambda}\mathbf{F}^{\top }\widehat{%
	\mathbf{F}}\right\Vert \notag \\
&&+\left\Vert \frac{1}{T^{2}N^{2}}\sum_{i=1}^{N}(1-a_{i})\mathcal{E}%
_{i}^{\top }\mathcal{E}^{\ast }\pmb{\Lambda}\mathbf{F}^{\top }\widehat{%
	\mathbf{F}}\right\Vert  \notag \\
&=&a_{TN}\left\Vert \frac{1}{T^{2}N^{2}}\sum_{i=1}^{N}(1-a_{i})\mathbf{G}%
_{i}^{\top }\mathcal{E}\pmb{\Lambda}\mathbf{F}^{\top }\widehat{\mathbf{F}}%
\right\Vert +O_{P}(a_{TN}^{2})+O_{P}(1/N).
\end{eqnarray*}%
In addition, since $|\mathbf{G}_{i}^{\top }\mathcal{E}\pmb{\Lambda}%
|_{2}=|\sum_{j=1}^{N}\sum_{t=1}^{T}\mathbf{g}_{it}^{\top }\mathbf{f}%
_{t}\varepsilon _{jt}\pmb{\lambda}_{j}^{\top }|_{2}=O(\sqrt{TN})$ by using
Lemma \ref{Lemma1} (b), we then have $a_{TN}\left\Vert \frac{1}{T^{2}N^{2}}%
\sum_{i=1}^{N}(1-a_{i})\mathbf{G}_{i}^{\top }\mathcal{E}\pmb{\Lambda}\mathbf{%
F}^{\top }\widehat{\mathbf{F}}\right\Vert =O_{P}(a_{TN}/\sqrt{TN})$.
Similarly, for the third term, we have 
\begin{equation*}
\left\Vert \frac{1}{T^{2}N^{2}}\sum_{i=1}^{N}(1-a_{i})\mathcal{E}_{i}^{\ast
	,\top }\mathcal{E}^{\ast }\mathcal{E}^{\ast ,\top }\widehat{\mathbf{F}}%
\right\Vert =O_{P}(a_{TN}/\sqrt{T\wedge N})+O_{P}(a_{TN}\sqrt{T}/(T\wedge
N)).
\end{equation*}%
This completes the proof. 
	
\end{proof}
}

\section{Verification of Assumption \ref{Assumption4}\label{App.B2}}

{\small In this section we verify Assumption \ref{Assumption4} via a given
example.  Consider a high-dimensional MA($\infty $) process $\pmb{\varepsilon}
_{t}=\sum_{j=0}^{\infty }\mathbf{B}_{j}\mathbf{e}_{t-j},$ where $\mathbf{B}
_{j}$'s are $N\times N$ matrices. We make the following primitive conditions
on $\left\{ \mathbf{B}_{j}\right\} ,$ $\{\mathbf{e}_{t-j}\}$ and $\{\mathbf{f%
}_{t}\}.$ }

\begin{assumption}
{\small \label{AssumptionA1}(a) $\Vert \mathbf{B}_{j}\Vert _{2}=O(j^{-\alpha
})$ for some $\alpha >3$; }

{\small (b) $\{e_{it}\}$ is independent over $(i,t)$ with mean 0 and
variance 1, and $E|e_{it}|^{4}<\infty $, where $e_{it}$ is the $i^{th}$
element of $\mathbf{e}_{t}$. }
\end{assumption}

\begin{assumption}
{\small \label{AssumptionA2}$\{\mathbf{f}_{t}\}$ satisfy Assumption \ref%
{Assumption3} , $\{\mathbf{f}_{t}\}$ is independent of $\{\pmb{\varepsilon}%
_{t}\}$, $\Vert \mathbf{B}_{j}\Vert _{1}=O(j^{-\alpha })$ and $\Vert \mathbf{%
B}_{j}\Vert _{\infty }=O(j^{-\alpha })$. }
\end{assumption}

{\small Below we verify that $\{\pmb{\varepsilon}_{t}\}$ fulfills Assumption %
\ref{Assumption4}(a)--(b) under Assumption \ref{AssumptionA1}(a)--(b).
Similarly, Assumption \ref{Assumption4}(c) is satisfied under Assumptions %
\ref{AssumptionA1}--\ref{AssumptionA2}. }

\subsection{Verification of Assumption \ref{Assumption4}(a)--(b)}

{\small By Assumption \ref{AssumptionA1}(a)-(b) and the independence between 
$\{\mathbf{e}_{t-j}\}$ and $\{\mathbf{f}_{t}\}$, $E(\boldsymbol{\varepsilon }%
_{t}|\mathbf{f}_{t})=E(\boldsymbol{\varepsilon }_{t})=0$ and $E[\boldsymbol{%
\varepsilon }_{t}\boldsymbol{\varepsilon }_{t}^{\top }]=$$\sum_{j=0}^{\infty
}\mathbf{B}_{j}\mathbf{B}_{j}^{\prime }\eqqcolon \Sigma _{\varepsilon }$. By the
triangle inequality and the submultiplicative property of the spectral norm,
we have%
\begin{equation*}
\left\Vert \Sigma _{\varepsilon }\right\Vert _{2}=\left\Vert
\sum_{j=0}^{\infty }\mathbf{B}_{j}\mathbf{B}_{j}^{\prime }\right\Vert
_{2}\leq \sum_{j=0}^{\infty }\left\Vert \mathbf{B}_{j}\right\Vert
_{2}^{2}\lesssim \sum_{j=0}^{\infty }j^{-2\alpha }=O(1).
\end{equation*}%
} {\small This verifies Assumption \ref{Assumption4}(a).}

{\small Let $\mathbf{B}_{j}=\{B_{j,kl}\}_{k,l\in \lbrack N]}$ and $\mathbf{v}%
^{\top }\mathbf{B}_{j}=(B_{j,\,\centerdot 1},\ldots ,B_{j,\,\centerdot N})$
for any $\Vert \mathbf{v}\Vert =1$. As $\{e_{it}\}$ are independent over $i$%
, we can write 
\begin{eqnarray*}
&&E\left\vert \mathbf{v}^{\top }\mathbf{B}_{j}(\mathbf{e}_{0}-\mathbf{e}%
_{0}^{\prime })\right\vert ^{4}=E\left\vert \sum_{l=1}^{N}B_{j,\,\centerdot
l}^{2}(e_{l,0}-e_{l,0}^{\prime })^{2}\right\vert ^{2}\notag \\
&&+4E\left\vert \sum_{l=1}^{N-1}\sum_{k=l+1}^{N}B_{j,\,\centerdot
l}B_{j,\,\centerdot k}(e_{l,0}-e_{l,0}^{\prime })(e_{k,0}-e_{k,0}^{\prime
})\right\vert ^{2} \notag\\
&\leq &O(1)\left( \sum_{l=1}^{N}B_{j,\,\centerdot l}^{2}\right)
^{2}+O(1)\sum_{l=1}^{N-1}\sum_{k=l+1}^{N}B_{j,\,\centerdot
l}^{2}B_{j,\,\centerdot k}^{2} \notag\\
&\leq &O(1)\left( \mathbf{v}^{\top }\mathbf{B}_{j}\mathbf{B}_{j}^{\top }%
\mathbf{v}\right) ^{2}=O(\Vert \mathbf{B}_{j}\Vert _{2}^{4}),
\end{eqnarray*}%
where the first inequality follows from some direct calculations, and the
second inequality follows from the fact that $\sum_{l=1}^{N-1}%
\sum_{k=l+1}^{N}B_{j,\,\centerdot l}^{2}B_{j,\,\centerdot k}^{2}\leq
(\sum_{l=1}^{N}B_{j,\,\centerdot l}^{2})^{2}$. Based on the above
development, we have 
\begin{equation*}
\sup_{\mathbf{v}\in \mathbb{R}^{N},\Vert \mathbf{v}\Vert =1}\left\vert 
\mathbf{v}^{\top }(\pmb{\varepsilon}_{j}-\pmb{\varepsilon}_{j}^{\ast
})\right\vert _{4}=\sup_{\mathbf{v}\in \mathbb{R}^{N},\Vert \mathbf{v}\Vert
=1}\left\vert \mathbf{v}^{\top }\mathbf{B}_{j}(\mathbf{e}_{0}-\mathbf{e}%
_{0}^{\prime })\right\vert _{4}=O(|\mathbf{B}_{j}|_{2})=O(j^{-\alpha }).
\end{equation*}%
This verifies Assumption \ref{Assumption4}(b).}

\subsection{{Verification of Assumption  \ref{Assumption4}(c)}}

{\small Let $\mathbf{f}_{t,t-m}$ and $\varepsilon _{it,t-m}$ be the coupled
version of $\mathbf{f}_{t}$ and $\varepsilon _{it}$ with $\mathbf{e}_{t-m}$
replaced by $\mathbf{e}_{t-m}^{\prime }$. Since $\{\mathbf{f}_{t}\}$ is
independent of $\{\pmb{\varepsilon}_{t}\}$, by Jensen inequality, we have 
\begin{eqnarray*}
&&\left\vert \mathcal{P}_{t-m}\left( \mathbf{f}_{t}\frac{1}{\sqrt{N}}
\sum_{i=1}^{N}\varepsilon _{it}\varepsilon _{is}\right) \right\vert _{2} 
\notag \\
&\leq &|E(\mathbf{f}_{t}\,|\, \mathscr{E}_{t-m})-E(\mathbf{f}_{t}\,|\, %
\mathscr{E}_{t-m-1})|_{2}\left\vert \frac{1}{\sqrt{N}}\sum_{i=1}^{N}
\varepsilon _{it}\varepsilon _{is}\right\vert _{2}  \notag \\
&&+|\mathbf{f}_{t}|_{2}\left\vert E\left( \frac{1}{\sqrt{N}}
\sum_{i=1}^{N}\varepsilon _{it}\varepsilon _{is}\,|\, \mathscr{E}
_{t-m}\right) -E\left( \frac{1}{\sqrt{N}}\sum_{i=1}^{N}\varepsilon
_{it}\varepsilon _{is}\,|\, \mathscr{E}_{t-m-1}\right) \right\vert _{2} 
\notag \\
&\leq &|E(\mathbf{f}_{t}-\mathbf{f}_{t,t-m}\,|\, \mathscr{E}
_{t-m})|_{2}\left\vert \frac{1}{\sqrt{N}}\sum_{i=1}^{N}\varepsilon
_{it}\varepsilon_{is}\right\vert _{2}  \notag \\
&&+|\mathbf{f}_{t}|_{2}\left\vert E\left( \frac{1}{\sqrt{N}}
\sum_{i=1}^{N}(\varepsilon _{it}\varepsilon _{is}-\varepsilon
_{it,t-m}\varepsilon _{is,t-m})\,|\, \mathscr{E}_{t-m}\right) \right\vert
_{2}.
\end{eqnarray*}
Noting that $|E(\mathbf{f}_{t}-\mathbf{f}_{t,t-m}\,|\, \mathscr{E}
_{t-m})|_{2}\leq |\mathbf{f}_{m}-\mathbf{f}_{m}^{\ast }|_{2}=O(m^{-\alpha })$
and $|\frac{1}{\sqrt{N}}\sum_{i=1}^{N}\varepsilon _{it}\varepsilon
_{is}|_{2}=O(1)$, it is sufficient to show the order of 
\begin{equation*}
\left\vert E\left( \frac{1}{\sqrt{N}}\sum_{i=1}^{N}(\varepsilon
_{it}\varepsilon _{is}-\varepsilon _{it,t-m}\varepsilon _{is,t-m})\,|\, %
\mathscr{E}_{t-m}\right) \right\vert _{2}.
\end{equation*}
}

{\small Let $\mathbf{B}_{j,i\cdot}$ denote the $i^{th}$ column of $\mathbf{B}
_{j}^\top$. Hence, we have $\varepsilon_{it} = \sum_{j=0}^{\infty}\mathbf{B}
_{j,i\cdot}^\top\mathbf{e}_{t-j}$ and $\varepsilon_{it,t-m} = \sum_{j=0,\neq
m}^{\infty}\mathbf{B}_{j,i\cdot}^\top\mathbf{e}_{t-j} + \mathbf{B}%
_{m,i\cdot}^\top\mathbf{e}_{t-m}^\prime$. Then write 
\begin{eqnarray*}
&&\left|E\left(\frac{1}{\sqrt{N}}\sum_{i=1}^{N}(\varepsilon_{it}
\varepsilon_{is}-\varepsilon_{it,t-m}\varepsilon_{is,t-m})\,|\, \mathscr{E}
_{t-m} \right)\right|_2  \notag \\
&\leq& \left|E\left(\frac{1}{\sqrt{N}}\sum_{i=1}^{N}(\mathbf{B}
_{m,i\cdot}^\top\mathbf{e}_{t-m}^\prime\mathbf{B}_{s+m-t,i\cdot}^\top\mathbf{%
e}_{t-m}^\prime-\frac{1}{\sqrt{N}}\sum_{i=1}^{N}(\mathbf{B} _{m,i\cdot}^\top 
\mathbf{e}_{t-m}\mathbf{B}_{s+m-t,i\cdot}^\top\mathbf{e} _{t-m})\,|\, %
\mathscr{E}_{t-m} \right)\right|_2  \notag \\
&&+ \left|E\left(\frac{1}{\sqrt{N}}\sum_{i=1}^{N}\mathbf{B}_{m,i\cdot}^\top( 
\mathbf{e}_{t-m}^\prime-\mathbf{e}_{t-m})(\sum_{j=0,\neq s+m-t}^{\infty} 
\mathbf{B}_{j,i\cdot}^\top \mathbf{e}_{s-j})\,|\, \mathscr{E}_{t-m}
\right)\right|_2  \notag \\
&&+\left|E\left(\frac{1}{\sqrt{N}}\sum_{i=1}^{N}\mathbf{B}
_{s+m-t,i\cdot}^\top(\mathbf{e}_{t-m}^\prime-\mathbf{e}_{t-m})(\sum_{j=0,
\neq m}^{\infty}\mathbf{B}_{j,i\cdot}^\top \mathbf{e}_{t-j})\,|\, \mathscr{E}
_{t-m} \right)\right|_2  \notag \\
&\eqqcolon &E_1 + E_2 + E_3.
\end{eqnarray*}
}

{\small For $E_1$, we have 
\begin{eqnarray*}
E_1&=& \left|\frac{1}{\sqrt{N}}\sum_{i=1}^{N}
\sum_{l=1}^{N}B_{m,il}B_{s+m-t,il}(e_{l,t-m}^2-1)+\frac{1}{\sqrt{N}}
\sum_{i=1}^{N}\sum_{l=1}^{N}\sum_{k=1,\neq
l}^{N}B_{m,il}B_{s+m-t,il}e_{l,t-m}e_{k,t-m} \right|_2  \notag \\
&\leq&\left|\frac{1}{\sqrt{N}}\sum_{i=1}^{N}
\sum_{l=1}^{N}B_{m,il}B_{s+m-t,il}(e_{l,t-m}^2-1)\right|_2  \notag \\
&& + \left|\frac{1}{\sqrt{N}}\sum_{i=1}^{N}\sum_{l=1}^{N}\sum_{k=1,\neq
l}^{N}B_{m,il}B_{s+m-t,ik}e_{l,t-m}e_{k,t-m} \right|_2  \notag \\
&\eqqcolon&E_{11} + E_{12}.
\end{eqnarray*}
For $E_{11}$, by using Cauchy-Schwarz inequality and $\frac{1}{N}
\sum_{l,i=1}^{N}B_{m,il}^2 \leq \|\mathbf{B}_{m} \|_2^2$, we have 
\begin{eqnarray*}
E_{11} &=& \left( \sum_{l=1}^{N}\left(\frac{1}{\sqrt{N}}
\sum_{i=1}^{N}B_{m,il}B_{s+m-t,il}\right)^2E\left((e_{l,t-m}^2-1)^2\right)
\right)^{1/2}  \notag \\
&\leq& O(1) \left(\max_{1\leq l\leq N}|\sum_{i=1}^{N}B_{s+m-t,il}^2|\frac{1}{
N}\sum_{l,i=1}^{N}B_{m,il}^2\right)^{1/2} =O( \|\mathbf{B}_{m}
\|_2)=O(m^{-\alpha}).
\end{eqnarray*}
Next, for $E_{12}$, we have 
\begin{eqnarray*}
E_{12} &=& \left( \sum_{l=1}^{N}\sum_{k=1,\neq l}^{N}\left(\frac{1}{\sqrt{N}}
\sum_{i=1}^{N}B_{m,il}B_{s+m-t,ik}\right)^2E\left(e_{l,t-m}^2e_{k,t-m}^2
\right) \right)^{1/2}  \notag \\
&\leq& O(1) \left(\sum_{l=1}^{N}\sum_{k=1,\neq l}^{N}\left(\frac{1}{\sqrt{N}}
\sum_{i=1}^{N}B_{m,il}B_{s+m-t,ik}\right)^2\right)^{1/2}  \notag \\
&\leq& O(1) \left(\max_{k}\sum_{i=1}^{N}|B_{s+m-k,ik}|\max_{k}
\sum_{j=1}^{N}|B_{s+m-t,jk}|\max_{l}\sum_{j=1}^{N}|B_{m,jl}|\frac{1}{N}
\sum_{i=1}^{N}\sum_{l=1}^{N}|B_{m,il}|\right)^{1/2}  \notag \\
&=& O(\sqrt{ \|\mathbf{B}_{m} \|_1\|\mathbf{B}_{m} \|_\infty}
)=O(m^{-\alpha}).
\end{eqnarray*}
Similarly, we can show that $E_2$ is also $O(\|\mathbf{B}_{m}
\|_2)=O(m^{-\alpha})$ for some $\alpha > 3$. }

{\small Finally, we consider the third term $E_{3}$. By using arguments as
used in the study of $E_{1}$, we have 
\begin{eqnarray*}
&&\left\vert E\left( \frac{1}{\sqrt{N}}\sum_{i=1}^{N}\mathbf{B}%
_{s+m-t,i\cdot }^{\top }(\mathbf{e}_{t-m}^{\prime }-\mathbf{e}%
_{t-m})(\sum_{j=0,\neq m}^{\infty }\mathbf{B}_{j,i\cdot }^{\top }\mathbf{e}%
_{t-j})\,|\,\mathscr{E}_{t-m}\right) \right\vert _{2} \notag \\
&=&\left\vert \frac{1}{\sqrt{N}}\sum_{i=1}^{N}\mathbf{B}_{s+m-t,i\cdot
}^{\top }\mathbf{e}_{t-m}\sum_{j=m+1}^{\infty }\mathbf{B}_{j,i\cdot }^{\top }%
\mathbf{e}_{t-j}\right\vert _{2}\notag \\
&\leq &\sum_{j=m+1}^{\infty }\left\vert \frac{1}{\sqrt{N}}\sum_{i=1}^{N}%
\mathbf{B}_{s+m-t,i\cdot }^{\top }\mathbf{e}_{t-m}\mathbf{B}_{j,i\cdot
}^{\top }\mathbf{e}_{t-j}\right\vert _{2}=\sum_{j=m+1}^{\infty }O(m^{-\alpha
})=O(m^{-\alpha +1}).
\end{eqnarray*}%
This completes the verification of Assumption \ref{Assumption4}(c). }

\section{Additional Simulation Results \label{App.B3}}

{\small This section reports some additional simulation results.}

\subsection{Moving Block Bootstrap for the SW and FHW Tests}

{\small Following the suggestion of a referee, we examine the applicability
of the moving block bootstrap (MBB) procedure for SW and FHW tests via
simulation studies. The details of the moving block bootstrap (MBB)
procedure can be found in Appendix C.6 of \cite{fu2023testing}, in which
they suggest a MBB procedure in order to deal with serially correlated
errors. }

{\small We generate the data under DGP.S3 where the error terms exhibit both
serial and cross-sectional correlations. The simulation results for the 5\%
test are reported in Table \ref{TableB1.1}. Overall, we can see that the MBB
procedure is not appropriate for obtaining the critical values of the SW and
FHW tests. Specifically, the SW and FHW tests display severe size
distortions and their sizes converge to zero with the increase of $T$. This
result is similar to those in Appendix C.6 of \cite{fu2023testing}, in which
they find that SW and FHW tests tend to over-reject but their sizes decrease
with the increasing sample size. }

{\small 
\begin{table}[tbp]
\caption{Finite Sample Performance of the MBB Procedure for SW and FHW tests}
\label{TableB1.1}
\vspace{-3mm} \setlength{\tabcolsep}{4pt} \renewcommand{\arraystretch}{0.95} 
\begin{center}

\scalebox{0.9}{\begin{tabular}{l cc ccccc | ccccc}
\hline\hline
& &  & \multicolumn{5}{c}{$N = 50$} & \multicolumn{5}{c}{$N = 100$} \\
\hline
& & $\widetilde{r}$ &PSY & CDG & HI& SW & FHW &PSY & CDG & HI& SW & FHW\\
\multirow{16}{*}{DGP.S3} 
&\multirow{4}{*}{$T = 30$}  & 2  &    0.139  & 0.007  & 0.000  & 0.164  & 0.072  & 0.132  & 0.001  & 0.000  & 0.306  & 0.058  \\
& & 3                &    0.125  & 0.001  & 0.000  & 0.147  & 0.107  & 0.094  & 0.000  & 0.000  & 0.244  & 0.093  \\
& & 4                &    0.086  & 0.000  & 0.000  & 0.129  & 0.114  & 0.106  & 0.000  & 0.000  & 0.193  & 0.092  \\
& & 5                &    0.076  & 0.000  & 0.000  & 0.089  & 0.134  & 0.095  & 0.000  & 0.000  & 0.071  & 0.143  \\
\cline{3-13}
&\multirow{4}{*}{$T = 50$}  & 2  &    0.112  & 0.007  & 0.000  & 0.078  & 0.035  & 0.112  & 0.009  & 0.001  & 0.166  & 0.040  \\
& & 3                &    0.082  & 0.002  & 0.000  & 0.070  & 0.043  & 0.075  & 0.003  & 0.000  & 0.119  & 0.039  \\
& & 4                &    0.062  & 0.000  & 0.000  & 0.089  & 0.049  & 0.079  & 0.000  & 0.000  & 0.124  & 0.041  \\
& & 5                &    0.061  & 0.000  & 0.000  & 0.113  & 0.061  & 0.072  & 0.000  & 0.000  & 0.172  & 0.068  \\
\cline{3-13} 
&\multirow{4}{*}{$T = 100$}  & 2 &    0.085  & 0.016  & 0.005  & 0.019  & 0.017  & 0.099  & 0.021  & 0.007  & 0.052  & 0.018  \\
& & 3                &    0.075  & 0.007  & 0.000  & 0.007  & 0.012  & 0.066  & 0.017  & 0.001  & 0.013  & 0.004  \\
& & 4                &    0.065  & 0.000  & 0.000  & 0.009  & 0.022  & 0.074  & 0.003  & 0.000  & 0.010  & 0.006  \\
& & 5                &    0.062  & 0.001  & 0.000  & 0.009  & 0.024  & 0.067  & 0.000  & 0.000  & 0.028  & 0.008  \\
\cline{3-13}                                                                                                                                                           
&\multirow{4}{*}{$T =200$}   & 2 &    0.071  & 0.045  & 0.021  & 0.002  & 0.012  & 0.062  & 0.033  & 0.019  & 0.009  & 0.014  \\
& & 3                &    0.073  & 0.027  & 0.007  & 0.000  & 0.003  & 0.074  & 0.027  & 0.004  & 0.000  & 0.007  \\
& & 4                &    0.074  & 0.013  & 0.000  & 0.000  & 0.006  & 0.073  & 0.017  & 0.000  & 0.000  & 0.004  \\
& & 5                &    0.071  & 0.005  & 0.000  & 0.000  & 0.004  & 0.076  & 0.008  & 0.000  & 0.000  & 0.004  \\
\hline\hline
\end{tabular}} 

\end{center}
\end{table}
}

\subsection{Results for Zero-Mean Common Factors}

In this section, we examine the performance of our test with
zero-mean common factors. As implied by Proposition \ref{prop2}, our test
statistic is of order $Nh^{-1/2}$ under the global alternatives if the
common factors have zero mean. Here, we focus on DGP.G1--DGP.G3 except that
the common factors are generated by a vector autoregressive model with no
intercept. 

Table \ref{TableB1.2} reports the empirical rejections rates of
various tests under DGP.G1--DGP.G3 with zero-mean common factors where the
nomial level is 5\%. To conclude, although the empirical rejection rates of
our test converge to $1$ at a slower rate than SW and FHW tests, our test is
still powerful against DGP.G1--DGP.G3 with zero-mean common factors,
especially when $r=2$.

{\small 
\begin{table}[tbp]
\caption{Empirical Rejection Rates for DGP.G1--DGP.G3}
\label{TableB1.2}

\vspace{-3mm} \setlength{\tabcolsep}{4pt} \renewcommand{\arraystretch}{0.95} 
\begin{center}
\scalebox{0.9}{\begin{tabular}{c cc ccccc | ccccc}
\hline\hline
& &  & \multicolumn{5}{c}{$N = 50$} & \multicolumn{5}{c}{$N = 100$} \\
\hline
& & $\widetilde{r}$ &PSY & CDG & HI& SW & FHW &PSY & CDG & HI& SW & FHW\\
\multirow{12}{*}{DGP.G1} &\multirow{4}{*}{$T = 50$} 
& 2            &    0.778  & 0.009  & 0.001  & 0.715  & 0.984  & 0.939  & 0.012  & 0.000  & 0.878  & 0.999  \\
&& 3            &    0.246  & 0.006  & 0.000  & 0.238  & 0.381  & 0.167  & 0.008  & 0.000  & 0.248  & 0.307  \\
&& 4            &    0.147  & 0.002  & 0.000  & 0.201  & 0.191  & 0.132  & 0.000  & 0.000  & 0.250  & 0.196  \\
&& 5            &    0.136  & 0.000  & 0.000  & 0.189  & 0.114  & 0.112  & 0.000  & 0.000  & 0.252  & 0.129  \\
\cline{3-13} 
&\multirow{4}{*}{$T = 100$} 
& 2             &    0.957  & 0.023  & 0.008  & 0.988  & 1.000  & 0.995  & 0.030  & 0.011  & 1.000  & 1.000  \\
&& 3            &    0.198  & 0.729  & 0.002  & 0.208  & 0.414  & 0.122  & 0.936  & 0.002  & 0.179  & 0.277  \\
&& 4            &    0.151  & 0.018  & 0.000  & 0.179  & 0.272  & 0.105  & 0.006  & 0.000  & 0.177  & 0.190  \\
&& 5            &    0.142  & 0.002  & 0.000  & 0.168  & 0.179  & 0.101  & 0.001  & 0.000  & 0.199  & 0.145  \\
\cline{3-13}                                                                                                 
&\multirow{4}{*}{$T = 200$}                                                                                   
& 2            &    0.999  & 0.043  & 0.034  & 1.000  & 1.000  & 1.000  & 0.037  & 0.026  & 1.000  & 1.000  \\
&& 3            &    0.122  & 1.000  & 0.955  & 0.169  & 0.419  & 0.074  & 1.000  & 0.995  & 0.123  & 0.276  \\
&& 4            &    0.109  & 0.027  & 0.143  & 0.133  & 0.312  & 0.090  & 0.018  & 0.327  & 0.104  & 0.191  \\
&& 5            &    0.066  & 0.012  & 0.000  & 0.119  & 0.250  & 0.074  & 0.008  & 0.000  & 0.116  & 0.167  \\
\hline
\multirow{12}{*}{DGP.G2} 
&\multirow{4}{*}{$T = 50$} 
& 2            &    0.606  & 0.004  & 0.001  & 0.888  & 0.996  & 0.834  & 0.013  & 0.002  & 0.967  & 1.000  \\
&& 3            &    0.283  & 0.001  & 0.000  & 0.628  & 0.700  & 0.266  & 0.012  & 0.000  & 0.715  & 0.561  \\
&& 4            &    0.147  & 0.000  & 0.000  & 0.567  & 0.482  & 0.188  & 0.001  & 0.000  & 0.647  & 0.403  \\
&& 5            &    0.151  & 0.000  & 0.000  & 0.503  & 0.329  & 0.165  & 0.000  & 0.000  & 0.551  & 0.301  \\
\cline{3-13} 
&\multirow{4}{*}{$T = 100$} 
& 2            &    0.859  & 0.023  & 0.007  & 0.999  & 1.000  & 0.991  & 0.033  & 0.010  & 1.000  & 1.000  \\
&& 3            &    0.166  & 0.726  & 0.001  & 0.669  & 0.721  & 0.179  & 0.924  & 0.002  & 0.688  & 0.523  \\
&& 4            &    0.143  & 0.013  & 0.000  & 0.648  & 0.612  & 0.172  & 0.010  & 0.000  & 0.734  & 0.437  \\
&& 5            &    0.142  & 0.002  & 0.000  & 0.633  & 0.499  & 0.125  & 0.000  & 0.000  & 0.740  & 0.350  \\
\cline{3-13}                                                                                                 
&\multirow{4}{*}{$T = 200$}                                                                                   
& 2             &    0.986  & 0.033  & 0.033  & 1.000  & 1.000  & 1.000  & 0.037  & 0.024  & 1.000  & 1.000  \\
&& 3            &    0.126  & 0.998  & 0.966  & 0.586  & 0.724  & 0.104  & 1.000  & 0.996  & 0.663  & 0.471  \\
&& 4            &    0.100  & 0.037  & 0.201  & 0.585  & 0.640  & 0.102  & 0.018  & 0.402  & 0.706  & 0.418  \\
&& 5            &    0.099  & 0.018  & 0.000  & 0.597  & 0.563  & 0.079  & 0.009  & 0.000  & 0.747  & 0.370  \\
\hline
\multirow{12}{*}{DGP.G3} 
&\multirow{4}{*}{$T = 50$} 
& 2             &    0.707  & 0.014  & 0.001  & 0.849  & 0.982  & 0.848  & 0.009  & 0.002  & 0.942  & 0.999  \\
&& 3            &    0.321  & 0.009  & 0.000  & 0.510  & 0.450  & 0.286  & 0.009  & 0.000  & 0.636  & 0.445  \\
&& 4            &    0.202  & 0.002  & 0.000  & 0.440  & 0.260  & 0.221  & 0.000  & 0.000  & 0.562  & 0.313  \\
&& 5            &    0.154  & 0.000  & 0.000  & 0.393  & 0.172  & 0.179  & 0.000  & 0.000  & 0.499  & 0.190  \\
\cline{3-13} 
&\multirow{4}{*}{$T = 100$} 
& 2             &     0.906  & 0.019  & 0.008  & 0.991  & 1.000  & 0.984  & 0.022  & 0.004  & 1.000  & 1.000  \\
&& 3            &     0.250  & 0.637  & 0.001  & 0.551  & 0.516  & 0.244  & 0.876  & 0.001  & 0.654  & 0.431  \\
& &4            &     0.181  & 0.010  & 0.000  & 0.527  & 0.401  & 0.183  & 0.005  & 0.000  & 0.665  & 0.333  \\
& &5            &     0.176  & 0.002  & 0.000  & 0.530  & 0.283  & 0.169  & 0.001  & 0.000  & 0.688  & 0.253  \\
\cline{3-13}                                                                                                 
&\multirow{4}{*}{$T = 200$}                                                                                   
& 2             &    0.988  & 0.028  & 0.032  & 1.000  & 1.000  & 1.000  & 0.030  & 0.024  & 1.000  & 1.000  \\
&& 3            &    0.187  & 1.000  & 0.934  & 0.529  & 0.538  & 0.148  & 1.000  & 0.990  & 0.593  & 0.363  \\
& &4            &    0.153  & 0.035  & 0.108  & 0.530  & 0.420  & 0.132  & 0.015  & 0.297  & 0.645  & 0.311  \\
& &5            &    0.149  & 0.008  & 0.000  & 0.529  & 0.336  & 0.122  & 0.017  & 0.000  & 0.695  & 0.256  \\
\hline\hline
\end{tabular}}
\end{center}
\end{table}
}

\subsection{Bandwidth Sensitivity}

{\small In this section, we investigate the sensitivity of our test to the
choice of bandwidth. As before, we adopt the rule of thumb to set the
benchmark bandwidth: $h_{C}=(TN)^{-1/5}.$ We further consider $%
h_{L}=0.8h_{C} $ and $h_{R}=1.2h_{C}$ respectively in order to examine the
sensitivity. Here, we focus on DGP.S3, DGP.L3, DGP.L6 and DGP.G3, all of
which allows both TSA and CSD. This allows us examine the performance of the
test under both TSA and CSD. }

{\small Table \ref{TableB.3} presents the simulation results for $%
T=100,200,300$ and $N=100,200$. As we can see clearly from these tables, the
size and power performance of our test remains largely consistent with that
in the main text. This suggests that our test results are not sensitive to
the choice of bandwidth if the sample size is not so small. }

{\small 
\begin{table}[tbp]
\caption{Size and Power Performance with Different Bandwidth Sequences}
\label{TableB.3}
\setlength{\tabcolsep}{4pt} \renewcommand{\arraystretch}{0.95} 
\begin{center}
\scalebox{0.9}{\begin{tabular}{c cc ccc| ccc|ccc}
\hline\hline
&\multicolumn{1}{l}{} &  & \multicolumn{3}{c}{$T = 100$} & \multicolumn{3}{c}{$T = 200$} & \multicolumn{3}{c}{$T = 300$} \\
\hline
&\multicolumn{1}{l}{} & $\widetilde{r}$ & $h_L$ & $h_C$ & $h_R$ & $h_L$ & $h_C$ & $h_R$ & 
$h_L$ & $h_C$ & $h_R$ \\ 
\multirow{8}{*}{DGP.S3}
&\multirow{4}{*}{$N = 100$} 
& 2 &0.116  & 0.107  & 0.097  & 0.092  & 0.072  & 0.048  & 0.073  & 0.061  & 0.054  \\
&& 3 &0.120  & 0.089  & 0.090  & 0.076  & 0.064  & 0.066  & 0.048  & 0.046  & 0.048  \\
&& 4 &0.093  & 0.072  & 0.062  & 0.069  & 0.062  & 0.054  & 0.056  & 0.059  & 0.043  \\
&& 5 &0.079  & 0.052  & 0.068  & 0.071  & 0.049  & 0.046  & 0.086  & 0.054  & 0.063  \\
\cline{3-12}
&\multirow{4}{*}{$N = 200$} 
& 2 &    0.078  & 0.085  & 0.081  & 0.089  & 0.070  & 0.073  & 0.077  & 0.064  & 0.065  \\
&& 3 &    0.090  & 0.069  & 0.089  & 0.072  & 0.076  & 0.078  & 0.098  & 0.065  & 0.066  \\
&& 4 &    0.086  & 0.079  & 0.082  & 0.077  & 0.064  & 0.060  & 0.096  & 0.085  & 0.066  \\
&& 5 &    0.081  & 0.070  & 0.061  & 0.068  & 0.050  & 0.049  & 0.061  & 0.057  & 0.050  \\
\hline
\multirow{8}{*}{DGP.L3}
&\multirow{4}{*}{$N = 100$} 
& 2 &  0.225  & 0.185  & 0.166  & 0.224  & 0.205  & 0.173  & 0.213  & 0.208  & 0.201  \\
&& 3 &  0.219  & 0.180  & 0.149  & 0.207  & 0.205  & 0.162  & 0.216  & 0.205  & 0.192  \\
&& 4 &  0.155  & 0.166  & 0.132  & 0.216  & 0.161  & 0.129  & 0.199  & 0.222  & 0.172  \\
&& 5 &  0.137  & 0.151  & 0.110  & 0.152  & 0.136  & 0.146  & 0.193  & 0.160  & 0.160  \\
\cline{3-12}
&\multirow{4}{*}{$N = 200$} 
& 2 &    0.234  & 0.183  & 0.183  & 0.248  & 0.210  & 0.186  & 0.265  & 0.189  & 0.167  \\
&& 3 &    0.201  & 0.189  & 0.131  & 0.212  & 0.192  & 0.163  & 0.212  & 0.186  & 0.176  \\
&& 4 &    0.196  & 0.162  & 0.138  & 0.170  & 0.169  & 0.115  & 0.195  & 0.190  & 0.155  \\
&& 5 &    0.164  & 0.142  & 0.086  & 0.182  & 0.164  & 0.143  & 0.195  & 0.181  & 0.151  \\
\hline
\multirow{8}{*}{DGP.L6}
&\multirow{4}{*}{$N = 100$} 
& 2 &    0.267  & 0.237  & 0.223  & 0.303  & 0.254  & 0.233  & 0.259  & 0.250  & 0.233  \\
&& 3 &    0.247  & 0.189  & 0.205  & 0.239  & 0.242  & 0.220  & 0.227  & 0.215  & 0.202  \\
&& 4 &    0.216  & 0.195  & 0.181  & 0.217  & 0.212  & 0.205  & 0.172  & 0.209  & 0.152  \\
&& 5 &    0.184  & 0.157  & 0.148  & 0.230  & 0.186  & 0.178  & 0.159  & 0.193  & 0.132  \\
\cline{3-12}
&\multirow{4}{*}{$N = 200$} 
& 2 &    0.288  & 0.231  & 0.201  & 0.241  & 0.213  & 0.197  & 0.213  & 0.189  & 0.189  \\
&& 3 &    0.245  & 0.218  & 0.173  & 0.214  & 0.191  & 0.175  & 0.206  & 0.209  & 0.138  \\
&& 4 &    0.212  & 0.195  & 0.177  & 0.178  & 0.170  & 0.161  & 0.230  & 0.182  & 0.165  \\
&& 5 &    0.211  & 0.163  & 0.161  & 0.167  & 0.162  & 0.160  & 0.225  & 0.157  & 0.155  \\
\hline
\multirow{8}{*}{DGP.G3}
&\multirow{4}{*}{$N = 100$}
& 2 &    1.000  & 1.000  & 1.000  & 1.000  & 1.000  & 1.000  & 1.000  & 1.000  & 1.000  \\
&& 3 &    0.121  & 0.098  & 0.072  & 0.076  & 0.055  & 0.042  & 0.059  & 0.058  & 0.042  \\
&& 4 &    0.105  & 0.083  & 0.072  & 0.069  & 0.056  & 0.047  & 0.073  & 0.049  & 0.038  \\
&& 5 &    0.091  & 0.073  & 0.049  & 0.069  & 0.047  & 0.043  & 0.047  & 0.037  & 0.043  \\
\cline{3-12}
&\multirow{4}{*}{$N = 200$}
& 2 &    1.000  & 1.000  & 1.000  & 1.000  & 1.000  & 1.000  & 1.000  & 1.000  & 1.000  \\
&& 3 &    0.138  & 0.105  & 0.098  & 0.063  & 0.065  & 0.058  & 0.059  & 0.043  & 0.041  \\
&& 4 &    0.108  & 0.089  & 0.062  & 0.097  & 0.062  & 0.082  & 0.069  & 0.058  & 0.047  \\
&& 5 &    0.127  & 0.111  & 0.057  & 0.079  & 0.073  & 0.069  & 0.051  & 0.047  & 0.031  \\
\hline\hline
\end{tabular}} 
\end{center}
\end{table}
}

\end{document}